\documentclass[aos]{imsart}
\usepackage{wrapfig}
\usepackage{amsfonts}
\usepackage{amssymb}
\usepackage{amsmath,tabu}
\usepackage{amsthm}
\usepackage{comment}
\usepackage{diagbox}
\usepackage{graphics}
\usepackage{epstopdf}
\usepackage{pst-all,psfrag}
\usepackage[unicode]{hyperref}
\usepackage{indentfirst}
\usepackage{subcaption}
\usepackage{natbib}

\usepackage{verbatim}
\usepackage{eurosym}
\usepackage[utf8]{inputenc}
\usepackage[english]{babel}
\usepackage{chngcntr}
\usepackage{apptools}
\usepackage{tabularx}

\usepackage[textsize=footnotesize]{todonotes}

\usepackage{accents}
\newcommand{\dbtilde}[1]{\accentset{\approx}{#1}}

\AtAppendix{\counterwithin{theorem}{section}}

\newtheorem{theorem}{Theorem}
\newtheorem{proposition}[theorem]{Proposition}

\newtheorem{lemma}[theorem]{Lemma}

\newtheorem{remark}{Remark}

\newtheorem{definition}[theorem]{Definition}
\renewcommand{\proofname}{Proof.}
\renewcommand{\refname}{References}
\newcommand{\eps}{\varepsilon}
\newcommand{\p}{\mathfrak p}
\newcommand{\q}{\mathfrak q}
\newcommand{\T}{\mathcal T}
\newcommand{\J}{\mathbf J}
\newcommand{\V}{\mathcal V}
\newcommand{\R}{\mathcal R}
\renewcommand{\aa}{\mathfrak a}
\newcommand{\1}{I}
\newcommand{\ii}{\mathbf i}
\begin{document}

\begin{frontmatter}

\title{Cointegration in large VARs}

\begin{aug}

\author{\fnms{Anna} \snm{Bykhovskaya}\ead[label=e1,mark]{anna.bykhovskaya@wisc.edu}}
\and
\author{\fnms{Vadim} \snm{Gorin}\ead[label=e2,mark]{vadicgor@gmail.com}}

\address{University of Wisconsin-Madison, \printead{e1,e2}}

\end{aug}


\maketitle

\begin{abstract}
The paper analyses cointegration in vector autoregressive processes (VARs) for the cases when both the number of coordinates, $N$, and the number of time periods, $T$, are large and of the same order.  We propose a way to examine a VAR of order $1$ for the presence of cointegration based on a modification of the Johansen likelihood ratio test. The advantage of our procedure over the original Johansen test and its finite sample corrections is that our test does not suffer from over-rejection. This is achieved through novel asymptotic theorems for eigenvalues of matrices in the test statistic in the regime of proportionally growing $N$ and $T$. Our theoretical findings are supported by Monte Carlo simulations and an empirical illustration. Moreover, we find a surprising connection with multivariate analysis of variance (MANOVA) and explain why it emerges.
\end{abstract}


\end{frontmatter}

\section{Introduction}

\subsection{Motivation}
The importance of cointegration in economics stems from the seminal papers \citet{granger1981} and \citet{engle_granger1987}. For example, as they show, monthly rates on 1-month and 20-year treasury bonds are cointegrated, which means that they are both non-stationary, but they have a stationary linear combination. A lot of other variables in macroeconomics and finance such as price level, consumption, output, trade flows, interest rates and so on are non-stationary, and, thus, are potentially subject to cointegration. When dealing with non-stationary time series, it is always a question whether one should work with levels or with differences. For multivariate settings such as vector autoregression (VAR), the choice of model would depend on whether the series is cointegrated or not.

There are several ways to test for the presence of cointegration (see, e.g., \citet{maddala} for the detailed description of various methods). One popular approach relies on checking whether the residuals from regressing one of the coordinates on the remaining ones are stationary. It is based on \citet{engle_granger1987} and was later extended in \citet{phillips_ouliaris1990}. Another widely used technique is due to S{\o}ren Johansen \citep{johansen1988, johansen1991}\footnote{A related approach was also proposed in \citet{stock_watson1988}.}. This approach assumes VAR structure and relies on the likelihood ratio. It tests a null hypothesis of at most $\rho$ cointegrating relationships versus an alternative of between $\rho$ and $r>\rho$ cointegrating relationships. The Johansen test turns out to be related to the eigenvalues of some random matrix and has a non-standard asymptotic distribution.

Neither of the approaches is commonly used in the analysis of large systems. However, in many situations the data turns out to have both cross-sectional and time dimensions being large. Natural examples are given, e.g., by financial data (stock prices, exchange rates, etc.) or by monthly data on trade and investments between countries (where the number of pairs of countries is large). The main difficulty with applying the above cointegration testing approaches to high-dimensional settings is that both of them require $N$, the cross-sectional dimension of a time series, to be fixed and small. For the former approach, the larger is $N$, the more regressions we need to run and interpreting their results becomes ambiguous. For the latter approach, it turns out, the asymptotic theory stops providing a good approximation for moderate values of $T$. The test starts to over-reject the null of fewer cointegrations in favor of the alternative (see, for example, \citet{ho_sorensen1996}, \citet{gonzalo_pitarakis}). For instance, the simulations reported in Table 1 of \citet{gonzalo_pitarakis} indicate that the empirical rejection rate based on the 95\% asymptotic critical value for the no-cointegration hypothesis (constructed using the asymptotics of \citet{johansen1988, johansen1991}) is $20\%$ at $N=5$, $T=30$. It eventually decreases to the desired $5\%$ as $T$ grows, but even at $T=400$ the empirical size is still $6\%$.

After size distortions became clear econometricians developed various procedures for correcting over-rejection. Popular methods include finite sample correction (e.g., \citet{Reinsel_Ahn}, \citet{johansen_correction}) and bootstrap (e.g., \citet{Swensen}, \citet{cav_et_all}).  Such modifications help to restore correct size for moderate values of $T$ (e.g., $T=50$), when $N$ is of a smaller order (e.g., $N=5$). Yet, the question of larger $N$ remained open.


A recent ground-breaking paper \citet{onatski_ecta} shows that the over-rejection when testing the null of no cointegration can be mathematically explained by considering an alternative asymptotics in which both $N$ and $T$ go to infinity jointly and proportionally. When $N$ is large such asymptotic regime turns out to better suit the finite sample properties of the data. While \citet{onatski_ecta} point this out, they do not provide alternative testing procedures.

The observation that large VARs might have analyzable joint limits opens a new area of research. That is, one can try to develop various sophisticated joint asymptotics and derive, as a result, appropriate ways to test for cointegration in settings where both $N$ and $T$ are large. This, however, requires new tools rooted in the random matrix theory. In our paper we propose such cointegration tests and introduce asymptotic theorems which make testing possible. Let us describe our main results.

\subsection{Results}

By modifying the Johansen likelihood ratio (LR) test, we come up with a way to examine the presence of cointegration in a time series when the cross-sectional dimension, $N$, and the time dimension, $T$, are of the same order. We consider a vector autoregression (VAR) of order $1$ in the error correction representation
\begin{equation}\label{var_p_intro}
\Delta X_t=\mu+\Pi X_{t-1}+\eps_t,\qquad t=1,\ldots,T,
\end{equation}
where $\Delta X_t:=X_t-X_{t-1}$, errors $\{\eps_t\}$ are Gaussian with $N\times N$ covariance matrix $\Lambda$, and $\Pi$, $\mu$ are unknown parameters.\footnote{Section \ref{section_vark} discusses possible generalizations of our approach to VAR($k$) for general $k$.} We do not restrict $\Lambda$ to be diagonal, which means that any cross-sectional correlation structure is allowed. Such cross-sectional heteroscedasticity assumption may be important in applications, where one expects variables, e.g., countries, to be correlated. In contrast, many previous approaches relied on specific forms of covariance $\Lambda$, see, e.g., \citet{Breitung_Pesarann_2008}, \citet[Section 7]{Bai_Ng_2008} for reviews and \citet{ZhangPanGao_2018} for recent developments.

Our model allows for an arbitrary linear trend in $X_t$, which is a desirable feature when one goes to applications. Yet, it also encompasses as a special case a model without a trend ($\mu=0$). The way we deal with the data turns out to be invariant to $\mu$, and we do not propose any different procedure even if a researcher has a prior knowledge that there is no trend. (We discuss this in more detail in Section \ref{Section_no_trend}.)

We are interested in whether the $N$-dimensional non-stationary process $X_t$ is cointegrated or not. That is, whether there is a non-zero vector $\beta$ such that $\beta'X_t$ is trend stationary.  Our main contributions lie in the construction of the appropriate test, analysis of its asymptotic distribution, and computation of critical values (some of them are reported in Table \ref{airy_quantiles} of Section \ref{Section_asy_res}).

Our procedure for cointegration testing relies on the following steps. First, we de-trend $X_t$ by subtracting $\frac{t}{T}\left(X_T-X_0\right)$ from $X_t$. Then we follow Johansen's procedure and regress both first differences, $\Delta X_t$, and lagged de-trended $X_t$ on a constant. That is, we de-mean the data. Then we calculate the squared sample canonical correlations between the residuals from those regressions. This corresponds to the eigenvalues from the modification of the Johansen LR statistics obtained by using the lagged de-trended version of $X_t$ and not de-trended first differences $\Delta X_t$.

The de-trending procedure can be interpreted in the following way: we take the first differences of the observed variables, de-mean them, and re-sum back. This is equivalent (up to a constant $X_0$ which disappears after we further de-mean the sums) to our de-trending. Such interpretation is related to \citet{ERS} which analyzes unit root testing in the presence of a trend $d_t$. Under the null hypothesis of a unit root and when the trend $d_t$ is linear, \citet{ERS} suggests to de-mean the first differences.  

As discussed in \citet{phillips_detrending}, de-trending plays an important role in improving the performance of cointegration tests. In our procedure there is an additional reason why de-trending matters. It leads to an unexpected connection with the Jacobi ensemble, a familiar object in random matrix theory, but a novel one in econometrics.\footnote{We remark that de-trending is implicitly present in the proofs of \citet{onatski_ecta,onatski_joe}, yet in our work it becomes a significant ingredient, rather than a technical detail.} Eventually, this connection is the main technical tool leading to our asymptotic results and construction of the test.

We show that after proper rescaling our test under the null hypothesis of no cointegration converges to the sum of the first $r$ elements of the Airy$_1$ point process (we formally define that process in Section \ref{subsection_limit}). Airy$_1$ process is a known object in the random matrix theory and its marginal distributions can be computed in various ways. We also present in Section \ref{Section_simulations} a simulation study of the speed of convergence, which supports the limiting results even for moderate values of $N$ and $T$. In Table \ref{rej_rate} of that section we demonstrate the significant improvement over the finite sample behavior of the Johansen test and its corrections reported in \citet{gonzalo_pitarakis}.

Along with the size we also report power simulations in Section \ref{Section_simulations}. The power depends on the choice of the alternative and we perform several experiments with random rank $1$ matrix $\Pi$ and random initial condition $X_0$ of varying magnitudes. In many cases the power is very close to $100\%$. In particular, we found high power even for moderate $N,T$, e.g., $N\approx30,\, T\geq150$.  In general, the results are encouraging and we see quite large power envelopes.

We also present a small empirical illustration of our testing procedure on weekly S$\&$P100 log-prices over $10$ years. This gives us approximately $500$ observations across time and $T/N\approx5$, corresponding to high power and close to $5\%$ sample rejection rate in simulations. Log-prices are known to have a unit root, and we do not find any strong evidence that they are cointegrated.

\subsection{Techniques}

Let us briefly indicate technical aspects of our proofs and their mathematical novelty. The key observation that we make is that a small perturbation of the matrix arising in the modified Johansen test has an explicit distribution of its eigenvalues. This distribution is called Jacobi ensemble, and its usual appearances in statistics include sample canonical correlations for two sets of independent data (as opposed to highly dependent in our case, see the discussion in Section \ref{Section_Jacobi_prelim}) and multivariate analysis of variance (MANOVA), see, e.g., \citet{Muirhead_book}. Our method has two main components which are new, as far as we are aware. First, our perturbation of the model is based on a replacement of a certain permutation in the matrix formulation of the modified Johansen test by a uniformly random orthogonal matrix. The second ingredient is a challenging computation of matrix integrals\footnote{Our arguments have some similarities with proofs in \citet{Hua}.} leading to the identity of the law of the perturbed matrix with the Jacobi ensemble.

Put it otherwise, we discover an  exactly-solvable\footnote{A random model is called exactly-solvable or integrable (see \citet{borodin_gorin_review} for an overview) if there exist explicit formulas for the expectations of non-trivial random variables describing the system. Such formulas provide a basis for asymptotic analysis of the systems of interest. This is in contrast to generic models where explicit formulas are rarely available.} model in a small neighborhood of the (random) matrix of the Johansen cointegration test. Let us center our attention on this exactly-solvable model. It can be used as an initial point for perturbative arguments leading to asymptotic theorems for our and possibly other modifications of the Johansen test.
In addition, our exactly-solvable model is not isolated, but rather it is a representative of a whole class of similar cases. We expect that our approach works in several other situation in which various test statistics in the vector autoregression context can be understood through (other) instances of the Jacobi ensemble. Justifications of this point of view are contained in Sections \ref{Section_no_trend} ,\ref{Section_white}, and \ref{Section_white_proof}.

\subsection{Outline of the paper}
Section \ref{Section_setting} describes the setting and the main objects of interest. Section \ref{Section_asy_res} presents asymptotic results, while Section \ref{section_proofs} gives a sketch of their proofs. Section \ref{Section_simulations} shows supporting Monte Carlo simulations and Section \ref{Section_s&p} applies our test to S\&P$100$. Additional results and extensions are presented in Section \ref{sec_extensions}. Finally, Section \ref{Section_conclusion} concludes. All proofs, unless otherwise noted, are in Appendix.

\section{Setting}
\label{Section_setting}
\subsection{Building block}
We consider an $N$-dimensional vector autoregressive process of order $1$, VAR($1$), based on a sequence of i.i.d.~mean zero Gaussian errors $\{\eps_t\}$ with non-degenerate covariance matrix $\Lambda$. That is,
\begin{equation}\label{var_1}
\Delta X_t=\Pi X_{t-1}+\mu+\eps_t,\qquad t=1,\ldots,T.
\end{equation}
where $\Delta X_t:=X_t-X_{t-1}$ and $\Pi,\,\mu$ are unknown parameters. We do not impose any restrictions on $\Lambda$, thus, we allow for arbitrary correlations across coordinates of $X_t$. The process is initialized at fixed $X_0$. We outline possible extensions to autoregressions of higher order, VAR($k$), in Section \ref{section_vark}.

We are interested in analyzing whether $X_t$ is cointegrated. That is, whether there exists a non-zero $N\times r_0$ matrix $\beta$ such that $\beta'X_t$ is (trend) stationary.

If there exists a $N\times r_0$ matrix $\beta$ of rank $r_0$ such that $\beta'X_t$ is (trend) stationary, but there does not exist a $N\times(r_0+1)$ matrix $\tilde{\beta}$ of rank $r_0+1$ such that $\tilde{\beta}'X_t$ is (trend) stationary, then we say that $X_t$ is cointegrated of order $r_0$. As shown in \citet[Granger representation theorem]{engle_granger1987}, \citet[Theorem 4.5]{johansen_book} the necessary condition for $X_t$ to be cointegrated of order $r_0$ is ${\rm rank}(\Pi)=r_0$ and, thus, there exist two $N\times r_0$ matrices $\alpha$ and $\beta$ of rank $r_0$ such that $\Pi=\alpha\beta'$. For the sufficiency one also needs to require an extra non-degeneracy condition involving $\alpha$, $\beta$, and $\Gamma_i$. This additional requirement is used to rule out the $I(2)$ processes, i.e., such that their second differences are stationary, while the first differences are not.

\subsection{Cointegration test}\label{subsection_test}
We test the null hypothesis $H_0$ of no cointegration, which is equivalent to ${\rm rank}(\Pi)=0$ or $\Pi\equiv 0$. The complement to the null is  ${\rm rank}(\Pi)>0$. Yet, in order to design out test we set ${\rm rank}(\Pi)\in[1,r]$, where $r$ is a fixed finite number, to be our alternative hypothesis $H_1$. Thus, our alternative, in line with Granger representation theorem, can be interpreted as having at most $r$ cointegrating relationships. While the test has different power depending on $r$ and the true data-generating process, when working with real data for which the true value of ${\rm rank}(\Pi)\equiv r_0$ is unknown one can use any $r$ in order to reject $H_0$.\footnote{In practice we recommend using reasonably small values of $r$, since our asymptotic theorems are valid in the regime of fixed $r$ and large $N,T$.} Formally, we test
\begin{equation}
\label{hypothesis}
H_0:\:\Pi\equiv0\hskip0.5cm \text{vs.}\hskip0.5cm H_1:\:{\rm rank}(\Pi)\in[1, r],
\end{equation}

Our testing procedure is a modification of the Johansen likelihood ratio (LR) test \citep{johansen1988, johansen1991}.\footnote{Johansen LR test is based on the maximization of the Gaussian likelihood.} We focus on the small $r$ regime, which differs from the classical Johansen LR test, where $r=N$ is more commonly used as an alternative. In that sense, our approach is close to the maximum eigenvalue ($\lambda_{\max}$) test. That turns out to be important in the large $N,T$ setting, see Section \ref{Section_discussion}. Our approach consists of several steps:

\textbf{Step 1.} De-trend\footnote{We discuss the role of de-trending in Section \ref{Section_no_trend}.} the data and define
\begin{equation}
\label{eq_detrending}
 \tilde X_t = X_{t-1} - \frac{t-1}{T} (X_T-X_0).
\end{equation}
Note that we also do a time shift. This shift is in line with Johansen test, where lags and first differences are regressed on the observables.

\textbf{Step 2.} De-mean the data. That is, regress de-trended lags, $\tilde X_t$, and first differences $\Delta X_t$ on a constant. Define the residuals from those regressions as
\begin{equation}\begin{split}\label{residuals}
 R_{0t}=\Delta X_{t}-\frac{1}{T}\sum_{\tau=1}^T \Delta X_{\tau}, \quad
 R_{1t}=\tilde X_{t}-\frac{1}{T}\sum_{\tau=1}^T \tilde X_{\tau}.
\end{split}\end{equation}

\textbf{Step 3.} Calculate the squared sample canonical correlations between $R_0$ and $R_1$. That is, define
\begin{equation}\begin{split}\label{S_matrices}
S_{ij}=\sum\limits_{t=1}^{T} R_{it} R^{\ast}_{jt},\quad i,j=0,1,
\end{split}\end{equation}
where here and below the notation $X^{\ast}$ means transpose of the matrix $X$ (transpose-conjugate whenever complex matrices appear). Then calculate $N$ eigenvalues $\lambda_1\geq\ldots\geq\lambda_N$ of the matrix $S_{10}S^{-1}_{00}S_{01}S^{-1}_{11}$. The eigenvalues solve the equation
\begin{equation}
\label{eq_Johansen_equation}
 \det(S_{10} S_{00}^{-1} S_{01}-\lambda S_{11})=0.
\end{equation}

\textbf{Step 4.} Form the test statistic
\begin{equation}\label{LR_NT}
LR_{N,T}(r)=\sum\limits_{i=1}^{r}\ln(1-\lambda_i).
\end{equation}
The subscript $N,T$ in \eqref{LR_NT} indicates that we modify Johansen LR test to develop the large $N,T$ asymptotics. This statistic after centering and rescaling  will be compared with appropriate critical values to decide whether one can reject $H_0$ or not (see Theorem \ref{theorem_J_stat}).

\medskip

Throughout the proofs and extensions, we will be using an alternative way to write the residuals and matrices $S_{ij}$. For this let us define the \emph{demeaning} operator $\mathcal P$. It is a linear operator in a $T$--dimensional space defined by its matrix
 \begin{equation} \label{eq_demean}
  \mathcal P=\begin{pmatrix} 1-\frac1T & -\frac1T &\dots & -\frac1T\\ -\frac1T & 1-\frac1T& \dots &-\frac1T\\&&\ddots\\ -\frac1T& -\frac1T &\dots & 1-\frac1T\end{pmatrix}.
 \end{equation}
 $\mathcal P$ is an orthogonal projector on the space orthogonal to the vector $(1,1,\dots,1)$.
By definition
$$
 R_{0it}=[\Delta X\mathcal P]_{it}=\Delta X_{it}-\frac{1}{T}\sum_{\tau=1}^T \Delta X_{i\tau},\quad
 R_{1it}=[\tilde X\mathcal P]_{it}=\tilde X_{it}-\frac{1}{T}\sum_{\tau=1}^T \tilde X_{i\tau}.
$$
Using the fact that $\mathcal P^2=\mathcal P$,
\begin{equation}
\label{eq_modified_Joh_matrices}
 S_{00}=\Delta X \mathcal P \Delta X^{*},\quad S_{01}=\Delta X \mathcal P \tilde X^*, \quad S_{10}=S_{01}^*=\tilde X \mathcal P \Delta X^*,\quad S_{11}=\tilde X \mathcal P \tilde X^*.
\end{equation}

Let us emphasize that our test differs from the original Johansen test in the fact that we use $\tilde X_t$ instead of $X_{t-1}$. Note that $\tilde X_t$ can be viewed as a rank $1$ perturbation of $X_{t-1}$. Hence, our test statistic is a finite rank perturbation of the original Johansen procedure.

\section{Asymptotic results}\label{Section_asy_res}

In this section we formulate our main asymptotic results, explain and discuss the role of the precise setting we chose, and indicate directions for generalisations. We provide a sketch of the proof of our asymptotic theorem in Section \ref{section_proofs}.

\subsection{Large \boldmath{$(N,T)$} limit of the test}\label{subsection_limit}

Our results use the Airy$_1$ point process. Thus, let us introduce it before we formulate the main theorems. The Airy$_1$ point process is a random infinite sequence of reals
$$
\aa_1>\aa_2>\aa_3>\dots
$$
which can be defined through the following proposition.
\begin{proposition}[\citet{Forrest_spectr},\citet{Tracy_Widom}] \label{Proposition_Airy_Gauss} Let $Y_N$ be $N\times N$ matrix of i.i.d.~$\mathcal{N}(0,2)$ Gaussian random variables and let $\mu_{1;N}\ge \mu_{2;N}\ge \dots \ge\mu_{N;N}$ be eigenvalues of $\frac{1}{2}\left(Y_N+Y_N^*\right)$. Then in the sense of convergence of finite-dimensional distributions
	\begin{equation}
	\label{eq_GOE_to_Airy}
	\lim_{N\to\infty} \left\{N^{1/6}\left(\mu_{i;N}-2\sqrt{N}\right) \right\}_{i=1}^N = \{ \aa_i\}_{i=1}^\infty.
	\end{equation}
\end{proposition}
The law of the first coordinate $\aa_1$ is known as the Tracy-Widom $F_1$ distribution; its distribution function can be written in terms of a solution of the Painleve $II$ differential equation.

We remark that from the computational point of view \eqref{eq_GOE_to_Airy} gives an efficient way to access the distribution of various functions of $\{\aa_i\}_{i=1}^{\infty}$.\footnote{An even faster way uses tridiagonal matrix models \citep{dumitriu_edelman}.} From the theoretical perspective, one would like to have a more structural definition, which can be used for the analysis. Such definitions exist, yet, unfortunately, none of them is particularly simple.\footnote{
There are several equivalent ways to define Airy$_1$ point process: through Pfaffian formulas for the correlation functions \citep{Forrest_spectr, Tracy_Widom}, through combinatorial formulas for the Laplace transform \citep{Sodin}, through eigenvalues of the Stochastic Airy Operator \citep{RRV}.}

\begin{theorem} \label{theorem_J_stat} Suppose that $T,N\to\infty$ in such a way that the ratio $T/N$ belongs to $[2+\gamma_1, \gamma_2]$ for some $\gamma_1,\gamma_2>0$. Suppose that $H_0$ holds, that is, we have \eqref{var_1} with $\Pi=0$. Then for each finite $r=1,2,\dots$, we have convergence in distribution for the largest eigenvalues defined in Eq.~\eqref{eq_Johansen_equation}:
	\begin{equation}
	\label{eq_statistic_limit}
	 \frac{\sum_{i=1}^{r} \ln(1-\lambda_i)- r \cdot c_1(N,T)}{ N^{-2/3}  c_2(N,T)}  \, \xrightarrow[T,N\to\infty]{d} \sum_{i=1}^r \aa_i,
	\end{equation}
	where
	\begin{equation}
	c_1\left(N,T\right)=\ln\left(1-\lambda_+\right), \qquad
	c_2\left(N,T\right)=-\frac{2^{2/3} \lambda_+^{2/3}}{(1-\lambda_+)^{1/3} (\lambda_+-\lambda_-)^{1/3}} \left(\p+\q\right)^{-2/3}  <0,
	\end{equation}
	\begin{equation}\label{pq_def}
	\p=2-\frac{2}{N}, \qquad \q=\frac{T}{N}-1-\frac{2}{N},\qquad \lambda_\pm=\frac{1}{(\p+\q)^2}\left[\sqrt{\p(\p+\q-1)}\pm \sqrt{\q}  \right]^2.
	\end{equation}
\end{theorem}
\begin{remark}
 The condition $T/N>2$ is important: The procedure for constructing $\lambda_i$ involves computing squared sample canonical correlations between two $N$--dimensional subspaces in $T$--dimensional space. If $T/N<2$, then two subspaces necessary intersect, leading to $\lambda_1=1$. Hence, Eq.~\eqref{eq_statistic_limit} would need an adjustment in such case.
\end{remark}

Figure \ref{airy_density} shows densities for the random variables $\sum_{i=1}^r \aa_i$ for $r=1,2,3$.  As $r$ grows, the skewness of $\sum_{i=1}^r \aa_i$ decreases, the expectations of $\sum_{i=1}^r \aa_i$ tend to $-\infty$ and the variances tend to $+\infty$.

\begin{figure}[t]
	\centering
	{\scalebox{0.75}{\includegraphics{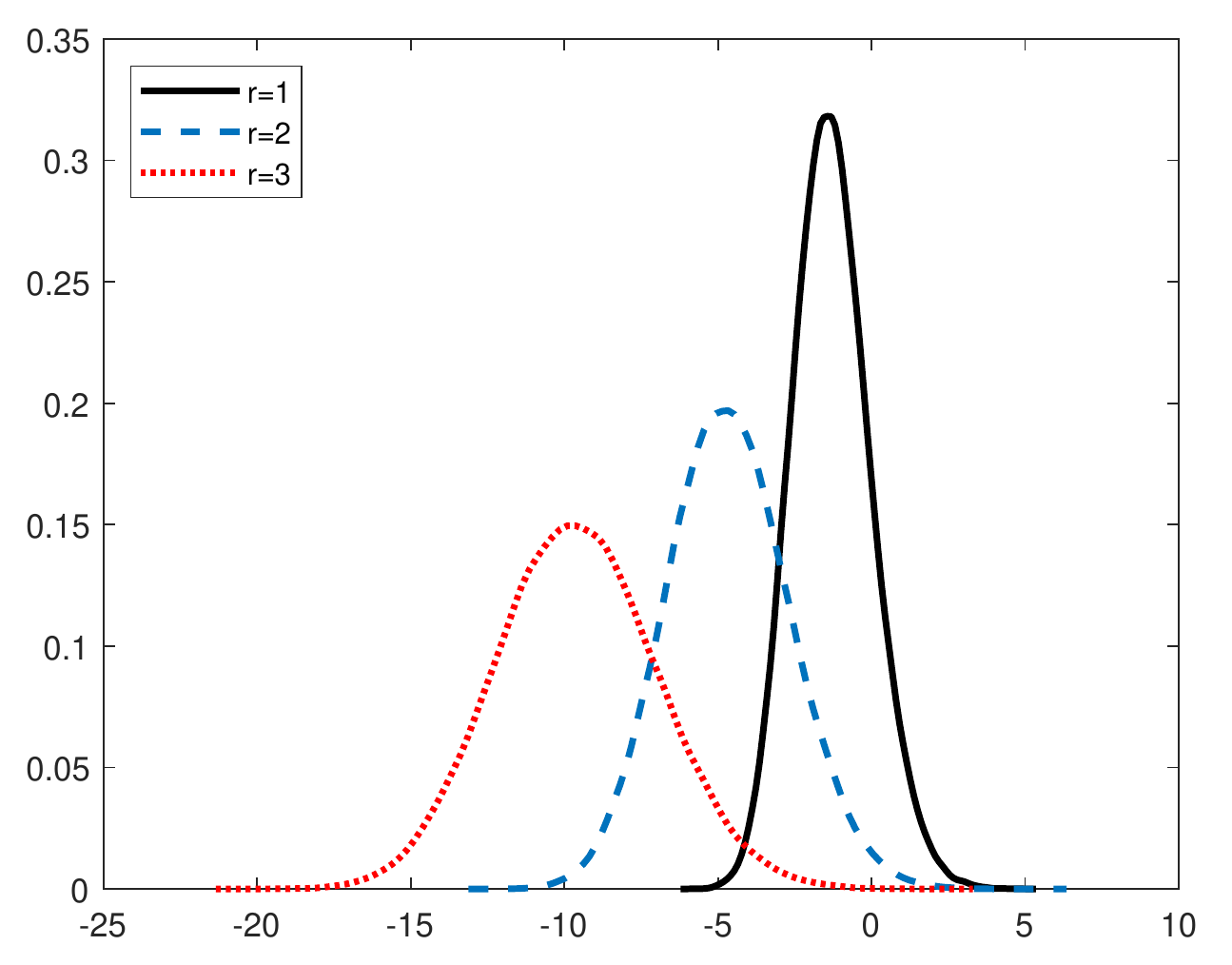}}}
	\caption{The probability density for the random variables $\sum\limits_{i=1}^r\aa_i$. The histogram plots are obtained from 100,000 samples of $200\times 200$ Gaussian matrices following Proposition \ref{Proposition_Airy_Gauss}. }
	\label{airy_density}
\end{figure}

Quantiles of the distribution of $\sum\limits_{i=1}^r \aa_i$ serve as critical values for testing the hypothesis $H_0$ of no cointegrations against the alternative of at most $r$ cointegrations. We report the quantiles for $r=1,2,3$ in Table \ref{airy_quantiles}.

\begin{table}[h]
	\begin{tabular}{c|c|c|c|c}
		\hline
		\diagbox[width=1.5cm, height=0.65cm]{$r$}{$\alpha$} & $0.9$ & $0.95$ & $0.975$ & $0.99$\\
		\hline
		\hline
		1 & 0.44  & 0.97  & 1.45  & 2.01 \\
		\hline
        2 & -1.88 &	-1.09 & -0.40& 	0.41\\
        \hline
        3 & -5.91 &	-4.91 & -4.03 & -2.99\\
        \hline
        \hline
        \multicolumn{5}{c}{}
	\end{tabular}
	\caption{Quantiles of $\sum\limits_{i=1}^r \aa_i$\, for $r=1,2,3$ (based on $10^6$ Monte Carlo simulations of $10^8\times 10^8$  tridiagonal matrices of \citet{dumitriu_edelman}).}	
  \label{airy_quantiles}
\end{table}

Note the shifts by $\frac{2}{N}$ in the definition \eqref{pq_def} of $\p$ and $\q$. Such finite sample correction is inspired by \citet[Discussion before Theorem 1]{Johnstone_Jacobi}. Theorem \ref{theorem_J_stat} remains valid without these shifts, i.e., for the simplified $\p=2$,  $\q=\frac{T}{N}-1$. However, our simulations show that for $T/N<6$ the shifts improve the speed of convergence and, thus, the finite sample behavior of the test, see Sections \ref{Section_simulations} and \ref{Section_Jacobi} for details.

\smallskip

We sketch the proof of Theorem \ref{theorem_J_stat} in Section \ref{section_proofs} and give full details in Appendix. The two main ideas that we use are:
\begin{enumerate}
 \item[(I)] We show that a small (vanishing as $N,T\to\infty$) perturbation of eigenvalues $\lambda_1\ge \lambda_2\ge\dots\ge\lambda_N$ leads to an explicit probability distribution known as the Jacobi ensemble (see Definition \ref{Definition_Jacobi} below). While this distribution appears in many problems of random matrix theory and multivariate statistics, its connection to the law of the Johansen test statistic remained previously unknown.
 \item[(II)]  Further, we rely on the universality phenomenon from random matrix theory, which says
that the particular choice ($\frac{1}{2}\left(Y_N+Y_N^*\right)$) for the law of random matrix made in Proposition \ref{Proposition_Airy_Gauss} is not of central importance.
Instead, the Airy$_1$ point process is a universal scaling limit for the largest eigenvalues in various ensembles of symmetric random matrices of growing sizes, see, e.g., \citet{ErdosYau} and \citet{Tao_Vu}. In particular,
an asymptotic statement similar to Proposition \ref{Proposition_Airy_Gauss} is known for the Jacobi ensemble (see Section \ref{Section_Jacobi}) and, by combining it with the first idea, we eventually arrive at Theorem \ref{theorem_J_stat}.
\end{enumerate}

\smallskip

It is natural to ask about an extension of Theorem \ref{theorem_J_stat} for the case of $r$ growing together with $N$. We distinguish two subcases here:

\begin{itemize}
	\item Slow growth: $r=\lfloor N^{\theta}\rfloor$ for some $0<\theta<1$.
	\item Linear growth: $r=\lfloor \rho N \rfloor $ for some $0<\rho\le 1$.
\end{itemize}

In both cases we expect that (after proper adjustment of $c_1$ and $c_2$) the limit in \eqref{eq_statistic_limit} becomes Gaussian. Although we are not going to pursue this direction here, for the slow growth case the proof of asymptotic normality
can probably  be obtained by the same methods as we use in Theorem \ref{theorem_J_stat}: we can again use the Jacobi ensemble for the computation. For the Jacobi ensemble individual eigenvalues $\lambda_i$ (in the regime of growing $i$ and $N-i$) become asymptotically Gaussian, hence, also their sums.\footnote{\citet{ORourke} proves asymptotic Gaussianity for eigenvalues of Gaussian random matrices and we expect the same proof to work for the Jacobi ensemble case, see also \cite{BPZ} for more details in the complex setting.}

For the linear growth case the situation is more complicated. Our present tools only allow us to prove an asymptotic upper-bound of the form
\begin{equation}\label{lin_growth_gauss}
\sum_{i=1}^{\lfloor \rho N\rfloor} \ln(1-\lambda_i)- N  \cdot c_3(T/N, \rho)=o(N^{\epsilon}), \quad N\to\infty,
\end{equation}
for any $\epsilon>0$ (for $\epsilon=1$ this also follows from the results of \citet{onatski_ecta, onatski_joe}),\footnote{The value of the constant $c_3$ can be computed through certain integrals involving the equilibrium measure of the Jacobi ensemble, see \eqref{eq_Jacobi_equilibrium} and \eqref{eq_Jacobi_LLN} for more details.} while we expect the expression to be $O(1)$. Yet, there exist very general statements on asymptotic Gaussianity of linear statistics of functions of random matrices (see, e.g., \citet{Guionnet_Novak}, \citet{Mingo_Popa}), and, therefore, there is little doubt in the fact that asymptotic distribution is Gaussian. Note, however, that these methods usually give very limited information about asymptotic variance and it is unclear at this moment how to find a reasonably simple explicit formula for it.

\subsection{Theorem \ref{theorem_J_stat} and the role of finite $\mathbf{\emph{r}}$: Discussion} \label{Section_discussion}

An important feature of our setting and our asymptotic results is that $r$ is kept finite as $N,T\to\infty$. (As far as the authors know, Eq.~\eqref{eq_statistic_limit} in Theorem \ref{theorem_J_stat} is the very first statistic for which the computation of the distributional limit became possible in the $N,T\to\infty$, $T/N\in [2+\gamma_1,\gamma_2]$ regime.) In practice, when the dimensions are given to us in the data, this means that $r$ should be much smaller than $N$ and $T$. Notice that for the purpose of rejecting $H_0$ any statistic with known asymptotic distribution can be useful. It could be tempting to instead choose $r=N$ and test $H_0$ of no co-integrations against the alternative of arbitrary number of cointegrations, yet, we believe that in many situations small $r$ allows to infer much more and is a good choice by the following reasons:

\subsubsection{What can be estimated and when?} \label{subsection_finite_r}
The ultimate goal of the cointegration analysis is to find the cointegrating relationships. Turns out, the feasibility of that task crucially depends on their true number. 
Imagine for a second that the true number of cointegrating relationships is large, say $N/2$, then they span an $N/2$--dimensional subspace of $\mathbb{R}^N$. Thus, the number of unknown parameters which need to be estimated is $N^2$ up to a multiplicative constant. On the other hand, the total number of observations that we have is $NT$. Hence, in our asymptotic regime $N,T\to\infty$, $T/N\in[2+\gamma_1,\gamma_2]$, we have only finitely many observations per estimated parameter and consistent estimation is unlikely. The conclusion from this dimension-based heuristics is that in the situation when the number of cointegration relationships is large, our data does not contain enough information for the consistent estimation of all cointegrating relationships and, in a sense, the problem is hopeless. On the other hand, the same heuristics shows that we are in a much better scenario, if we have some a priory reasons to expect that either the number of cointegrating relationships is small, or if this number is in a small neighbourhood of $N$ (in the latter case instead of the space of cointegrating relationships we can estimate its low-dimensional orthogonal complement). The former situation precisely corresponds to our choice of small $r$ in Theorem \ref{theorem_J_stat}; the first step towards the asymptotic analysis in the latter situation is discussed in Section \ref{Section_white}.

\subsubsection{Estimation of rank and cointegration relationships} In the small $N$ cointegration analysis described in detail in \citet{johansen1988, johansen1991, johansen_book}, rejection of the hypothesis of no cointegration is the first step. The next step is the estimation of the true rank $r_0$ of $\Pi$ in \eqref{var_1} by subsequent rejections of the hypotheses of the rank being smaller than $2,3,\dots, r_0$. Finally one wants to consistently estimate $r_0$ cointegrating relationships, closely related to $\Pi$ itself.

Switching to our limit regime $N,T\to\infty$, $T/N\in [2+\gamma_1,\gamma_2]$, there are two separate cases. If $r_0$ also grows linearly in $N$, then the full success of the above program is unlikely: as outlined in Section \ref{subsection_finite_r}, we are trying to consistently estimate too many parameters from too few data points.

On the other hand, if true $r_0$ is small, then the situation is different (another potentially feasible case is that of small $N-r_0$, which we do not discuss here). We believe that our present results --- Theorem \ref{theorem_J_stat} and the new approach underlying it --- open a path to the full realization of the cointegration analysis for growing $N$. Let us outline the reasons.

From the technical point of view a key ingredient which needs to be developed is the asymptotics of squared sample canonical correlations $\lambda_1,\dots,\lambda_N$ and corresponding eigenvectors for the model \eqref{var_1} with $\Pi$ of a fixed and not growing rank $r_0$. In particular Theorem \ref{theorem_J_stat} should be extended from $H_0:\:r_0=0$ to arbitrary finite $r_0$. A natural approach here is by further developing the link to random matrices established in our proof of Theorem \ref{theorem_J_stat} via the Jacobi ensemble. This leads to the theory of spiked random matrices. A basic question of this theory is to infer the information on a large-dimensional small rank matrix $B$ from observing $A+B$, where $A$ is a matrix of pure noise (e.g., one can use $A=\frac{1}{2}\left(Y_N+Y_N^*\right)$ in the notations of Proposition \ref{Proposition_Airy_Gauss}). This is a well-developed theory, see e.g., \citet{johnstone_spiked}, \citet{BBP} for seminal contributions and \citet{spiked_china}, \citet{johnstone_onatski2020} for the most recent progress. In order to obtain a connection with the VAR-framework \eqref{var_1}, we treat the $rank(\Pi)=r_0$ case as a finite rank deformation of the ``pure noise case'' $\Pi=0$, which we consider in this text. Hence, by combining our present results with the spiked random matrix theory, we expect to establish the complete asymptotic theory for any finite as $N\to\infty$ value of $r_0$. (This will require further technical and conceptual efforts and, therefore, is left for future work.)

\subsubsection{Power} Suppose that the true number of cointegrating relationships is small, say, $r_0=1$,  yet we are trying to reject the null of no cointegrations using the LR-statistic with $r=N$, i.e.,
\begin{equation}
\label{eq_full_statistic}
 \sum_{i=1}^N \ln(1-\lambda_i)
\end{equation}
in the notations of Theorem \ref{theorem_J_stat}. As we explain at the end of Section \ref{subsection_limit}, the standard deviation of \eqref{eq_full_statistic} is $O(1)$, i.e., after subtracting a constant as in Eq.~\eqref{lin_growth_gauss}, statistic \eqref{eq_full_statistic} stays finite as $N,T\to\infty$. Hence, in contrast to \eqref{eq_statistic_limit} there should be no $N$--dependent rescaling in a test based on \eqref{eq_full_statistic}. This should lead to very different powers of tests based on statistics \eqref{eq_statistic_limit} and \eqref{eq_full_statistic}.

Indeed, we expect the main difference in the joint law of $(\lambda_1,\dots,\lambda_N)$ under the $H_0$ and under the alternative of one cointegrating relationship to be in the behavior of $\lambda_1$ (which is in line with $\ln(1-\lambda_1)$ being the likelihood ratio test statistic in such setting). Change of $\lambda_1$ by $\nu\in\mathbb{R}$, i.e., $\lambda_1\to \lambda_1+\nu$, is on the same scale as the standard deviation of \eqref{eq_full_statistic}, and $\nu$ needs to be large in order for the test statistic to detect this change and to reject $H_0$. On the other hand, \eqref{eq_statistic_limit} is changing in the same situation by a number of order $N^{2/3} \nu$, which is much larger, and, hence, rejection of $H_0$ is more likely. Therefore, one could anticipate that \eqref{eq_statistic_limit} leads to a test of higher power than \eqref{eq_full_statistic} in this situation.\footnote{We do not analyze \eqref{eq_full_statistic} in this text and, thus, we do not provide any rigorous justification of the above. We remark that the existing in the literature similar comparisons, e.g., \citet{LST2001}, \citet{paruolo2001}, do not apply in our situation, since they only deal with the small $N$ case.}

\subsubsection{Checking model validity} The fact that we assume $r$ is finite and only use first $r$ largest eigenvalues allows us to use the remaining ones to verify the plausibility of our model specification. That is, we can check whether the statistical properties of $\lambda_{r+1},\dots,\lambda_N$ agree with the predictions of our model: for any finite $r$ we expect the empirical distribution function $\frac{1}{N-r} \sum_{i=r+1}^N \mathbf 1_{\lambda_i\le x}$ to converge to the CDF of the Wachter distribution (under $H_0$ this can be deduced from our Theorems \ref{Theorem_main_approximation} and \ref{Theorem_Jacobi_as}, while more general result under $H_1$ with finite $r$ follows from \cite[Theorem 1]{onatski_ecta}). Our empirical example in Section \ref{Section_s&p}, indeed, shows the match of the S$\&$P$100$ data with the Wachter distribution (see Figure \ref{Wachter_data}). 


\section{Outline of the proofs}\label{section_proofs}

The proof of Theorem \ref{theorem_J_stat} rests on the notion of the Jacobi ensemble. Let us first define it and then provide a sketch of the proof of Theorem \ref{theorem_J_stat}.

\subsection{Jacobi ensemble}\label{Section_Jacobi_prelim}

\begin{definition} \label{Definition_Jacobi}
 A (real) Jacobi ensemble $\J(N;p,q)$ is a distribution on $N\times N$ symmetric matrices $M$ of density proportional to
 \begin{equation}
  \label{eq_Jacobi_def}
  \det(M)^{p-1} \det(\1_N-M)^{q-1},\quad 0<M< \1_N,
 \end{equation}
 with respect to the Lebesgue measure, where $p,q>0$ are two parameters and $0< M < \1_N$ means that both $M$ and $\1_N-M$ are positive definite.
\end{definition}
When $N=1$, \eqref{eq_Jacobi_def} is the Beta distribution. For general $N$, the eigenvectors of  random Jacobi-distributed $M$ are uniformly distributed, while $N$ eigenvalues $x_1\ge \dots\ge x_N$ admit an explicit density with respect to the Lebesgue measure given by
\begin{equation}
 \label{eq_Jacobi_eig}
 \frac{1}{Z(N;p,q)}\prod_{1\le i<j \le N} (x_i-x_j) \prod_{i=1}^N x_i^{p-1} (1-x_i)^{q-1},
\end{equation}
where $Z(N;p,q)$ is an (explicitly known) normalization constant.

The Jacobi ensemble is widely used in statistics \citep{Muirhead_book}, theoretical physics \citep{Mehta}, and random matrix theory \citep{forrest}. There are numerous tools for studying it\footnote{Just to mention some: Pfaffian point processes \citep{Mehta}, combinatorics of moments \citep{Bai_Silverstein}, variational problems for log-gases \citep{BenArous}, Kadell integrals \citep{Kadell}, Schwinger-Dyson/Loop equations \citep{Johansson1998}, tridiagonal models \citep{Killip_Nenciu}, Macdonald processes \citep{BorG}.}, and as a result its asymptotic properties (usually as $N\to\infty$) are known in great detail.

Two particularly important appearances of the Jacobi ensemble in statistics are in the multivariate analysis of variance (MANOVA) and in the canonical-correlation analysis, see \citet[Sections 1 and 2]{Johnstone_Jacobi} for more detailed discussions and more examples. Let us explain the latter setting, as it has some resemblance with the Johansen test. We start with two rectangular matrices of data: $X$ of size $N\times T$ and $Y$ of size $K\times T$. One can interpret $X_{it}$ (and $Y_{it}$) as the $i$th variable at time $t$.

If $N=K=1$, then one can think about data sets $X$ and $Y$ as $T$ observations of two variables. Then one can measure the (linear) dependence between these variables by computing the (squared) sample correlation coefficient:
\begin{equation}
\label{eq_cor_1}
 \frac{ \left(\sum_{t=1}^T X_t Y_t\right)^2}{ \sum_{t=1}^T (X_t)^2  \sum_{t=1}^T (Y_t)^2}.
\end{equation}
A direct computation shows that if the elements of $X$ and $Y$ are i.i.d.~mean zero Gaussians, then the law of \eqref{eq_cor_1} is given by Beta distribution for any $T=1,2,\dots$.

A generalization of \eqref{eq_cor_1} to $N,K\ge 1$ defines the  \emph{squared sample canonical correlations} $\{r_i\}$ as solutions to the equation
\begin{equation}
\label{eq_cor_general}
 \det\left( S_{XY} S_{YY}^{-1} S_{YX}- r S_{XX}\right)=0,
\end{equation}
where $S_{XY}= X Y^*$, $S_{YY}=Y Y^*$, $S_{YX}= Y X^*$, $S_{XX}= X X^*$ (by $X^*$ we mean transpose of $X$ when $X$ is a real matrix and conjugate-transpose of $X$ when $X$ is complex). Generically the equation \eqref{eq_cor_general} has $\min(N,K)$ non-zero solutions. They have a variational meaning. For instance, the maximal solution of Eq.~\eqref{eq_cor_general} is the maximal squared correlation coefficient between $a^* X$ and $b^*Y$, where the maximization goes over all $N$--dimensional vectors $a$ and $K$--dimensional vectors $b$.

Now suppose that the columns of $X$ are i.i.d.~$N$--dimensional mean $0$ Gaussians with (non-degenerate) covariance matrix $\Lambda$ and the columns of $Y$ are i.i.d.~$K$--dimensional mean $0$ Gaussians with (non-degenerate) covariance matrix $\Lambda'$. Assume for simplicity that $N\le K$ and $N+K \le T$. If $X$ and $Y$ are independent, then it can be shown that the (squared) sample canonical correlations $r_1>r_2>\dots>r_N$ have the law of the eigenvalues of the Jacobi ensemble $\J(N; \frac{K-N+1}{2},\frac{T-N-K+1}{2})$.

At this point we observe both a similarity and a difference between the above instance of the Jacobi ensemble and the matrix appearing in the Johansen test. On one hand, the latter also deals with sample canonical correlations of two data sets. On the other hand, the data sets are no longer independent, instead one is obtained from another by a deterministic linear transformation. So are the roots of Eq.~\eqref{eq_Johansen_equation} related to the Jacobi ensemble? There is evidence in both directions. First, computer simulations quickly reveal that in one-dimensional case the distribution of the single eigenvalue in the Johansen test is \emph{not} the Beta distribution. Yet, second, recent results of \citet{onatski_ecta,onatski_joe} show that the Law of Large Numbers for the empirical distribution of the eigenvalues appearing in the Johansen test (with roots of Eq.~\eqref{eq_Johansen_equation} being a particular case) matches the one for the Jacobi ensemble (with shifted dimension parameters) in the limit as $N,T\to\infty$. Those articles were asking for an explanation.

\subsection{Sketch of the proof of Theorem \ref{theorem_J_stat}}

To prove Theorem \ref{theorem_J_stat}, we first need to establish the following central statement: a small perturbation of the Johansen test matrix, obtained by replacing the deterministic summation matrix in its definition by a random analogue, exactly matches the Jacobi ensemble. We further show that the perturbation vanishes in the limit as $N,T\to\infty$, thus, allowing us to obtain the asymptotics of the variants of the Johansen test from the known asymptotic results for the Jacobi ensemble.

\begin{theorem}
\label{Theorem_main_approximation}
Suppose that $T,N\to\infty$ in such a way that $T>2N$ and the ratio $T/N$ remains bounded. Under the hypothesis $\Pi=0$ for \eqref{var_1}, one can couple (i.e.~define on the same probability space) the eigenvalues $\lambda_1\ge \lambda_2\ge \dots\ge \lambda_N$ of the matrix $S_{10} S_{00}^{-1} S_{01}S_{11}^{-1}$ and eigenvalues $x_1\ge \dots\ge x_N$ of the Jacobi ensemble $\J(N;\frac{N}{2}, \frac{T-2N}{2})$ in such a way that  for each $\epsilon>0$ we have
 $$
   \lim_{T,N\to\infty} \mathrm{Prob}\left( \max_{1\le i \le N} |\lambda_i-x_i|< \frac{1}{N^{1-\epsilon}}\right)=1.
 $$
\end{theorem}
The proof of Theorem \ref{Theorem_main_approximation} is based on the following idea. Looking at Eq.~\eqref{var_1} when $\Pi=0$,  one can notice that matrices entering into the test and given by Eq.~\eqref{eq_modified_Joh_matrices} can be expressed in terms of the matrix of data $X$ and the lag operator mapping $X_t\to X_{t-1}$. A computation shows that since we deal with the de-trended and de-meaned data, we can replace the lag operator with its \emph{cyclic version},  which maps $X_1$ to $X_{T}$ rather than $X_0$. The latter is an orthogonal operator whose eigenvalues are roots of unity of order $T$. Then, the idea is to replace this operator by \emph{uniformly random} orthogonal operator. From that we proceed in two steps:
\begin{enumerate}
\item[(I)] We show that when the lag operator is replaced by its random counterpart, the eigenvalues of $S_{10} S_{00}^{-1} S_{01}S_{11}^{-1}$ have distribution $\J(N;\frac{N}{2}, \frac{T-2N}{2})$. We remark that this is a new appearance of the Jacobi ensemble, which was not present in the literature before.
\item[(II)] We show that replacement of the lag operator by its random counterpart introduces an error, which can be upper-bounded by $N^{\epsilon-1}$ for any $\epsilon>0$. This part is based on the rigidity results from random matrix theory, which say that eigenvalues of a uniformly random orthogonal matrix can be very closely approximated by equally spaced roots of unity.
\end{enumerate}
The full proof of Theorem \ref{Theorem_main_approximation} is given in Appendix. In addition to the real case, we also simultaneously prove a similar statement for complex matrices, encountering Jacobi ensemble of Hermitian matrices. Generally complex settings are rare guests in economics. Yet, they play a major role in the spectral analysis of time series data and in other areas, such as high energy physics.

By combining Theorem \ref{Theorem_main_approximation} with known asymptotic results for the Jacobi ensemble we can obtain the asymptotics of test statistic \eqref{LR_NT} in various regimes.

\section{Monte Carlo simulations}

\label{Section_simulations}

In this section we illustrate the performance of our test via Monte Carlo simulations. We consider both size (rejection rate) and power.

\subsection{Rejection rate}\label{Section_rej_rate}
First, we compare the finite sample performance of our approach versus Johansen's LR test and one of its corrected versions for (reasonably small) $T=30$ and $N=5,\dots,10$. The finite sample correction takes the form $\frac{T-N}{T}$ and was suggested in \citet{Reinsel_Ahn}. (Let us note that there are several more advanced empirical finite sample correction procedures for Johansen's LR test and its variations, we refer to \citet[Table 2]{onatski_joe} for a recent comparison of some of those for a wide range of values of $N$ and $T$.) Table \ref{rej_rate} summarizes the results (numbers closer to $5$ mean better performance). The $LR_{N,T}$ column is our test and the last two columns are from \citet{gonzalo_pitarakis}. They correspond to the same null of no cointegration, but use a different from ours alternative hypothesis. Our alternative is at most $r$ cointegrating relationships (in Table \ref{rej_rate} we use $r=1$, which means we look at the largest eigenvalue). The LR and RALR tests consider ``at most $N$'' cointegrations as $H_1$, which means that they use the sum of all eigenvalues. We can see that in finite samples when both $N$ and $T$ are of a similar magnitude, our approach significantly outperforms the alternatives based on small $N$, large $T$ asymptotics.

\begin{table}[h]
\begin{tabular}{c|c|c|c}
  \hline
  $N$ & $LR_{N,T},\,r=1$ & LR & RALR\\
  \hline
  \hline
  5 & 6.60 & 20.75 & 3.59\\
  6 & 5.45 & 31.66 & 2.68\\
  7 & 4.52 & 47.44 & 1.98\\
  8 & 3.80 & 67.42 & 2.00\\
  9 & 3.16 & 85.00 & 1.32\\
  10 & 2.60 & 96.69 & 0.96\\
  \hline
  \hline
  \multicolumn{4}{c}{}
\end{tabular}
\caption{Empirical size under no cointegration hypothesis ($5\%$ nominal level).
Data generating process: $\Delta X_{it}=\eps_{it}$, $\eps_{it}\thicksim$ i.i.d.~$\mathcal{N}(0,1)$, $T=30$, $MC=1,000,000$ replications for $LR_{N,T}$ and $MC=10,000$ for LR and RALR.}
\label{rej_rate}
\end{table}

To illustrate the performance of our test as the sample size increases, we fix the ratio $T/N\equiv c$ and plot the empirical size as a function of $N$. This is shown in Figure \ref{rej_rate_c}, where the target is $5\%$ rejection rate. The picture suggests that the test has rejection rate close to $5\%$. The green solid line corresponding to $c=4$ achieves $5\%$ rejection rate very fast. Other three curves overshoot $5\%$ by couple percents (e.g., both $c=5$ and $c=6$ are always below $7\%$). Moreover, the larger is $c$, the higher is the corresponding rejection curve. E.g.,~$c=10$ curve (blue, short dash) is strictly above other curves.


\begin{figure}[h]
{\scalebox{0.5}{\includegraphics{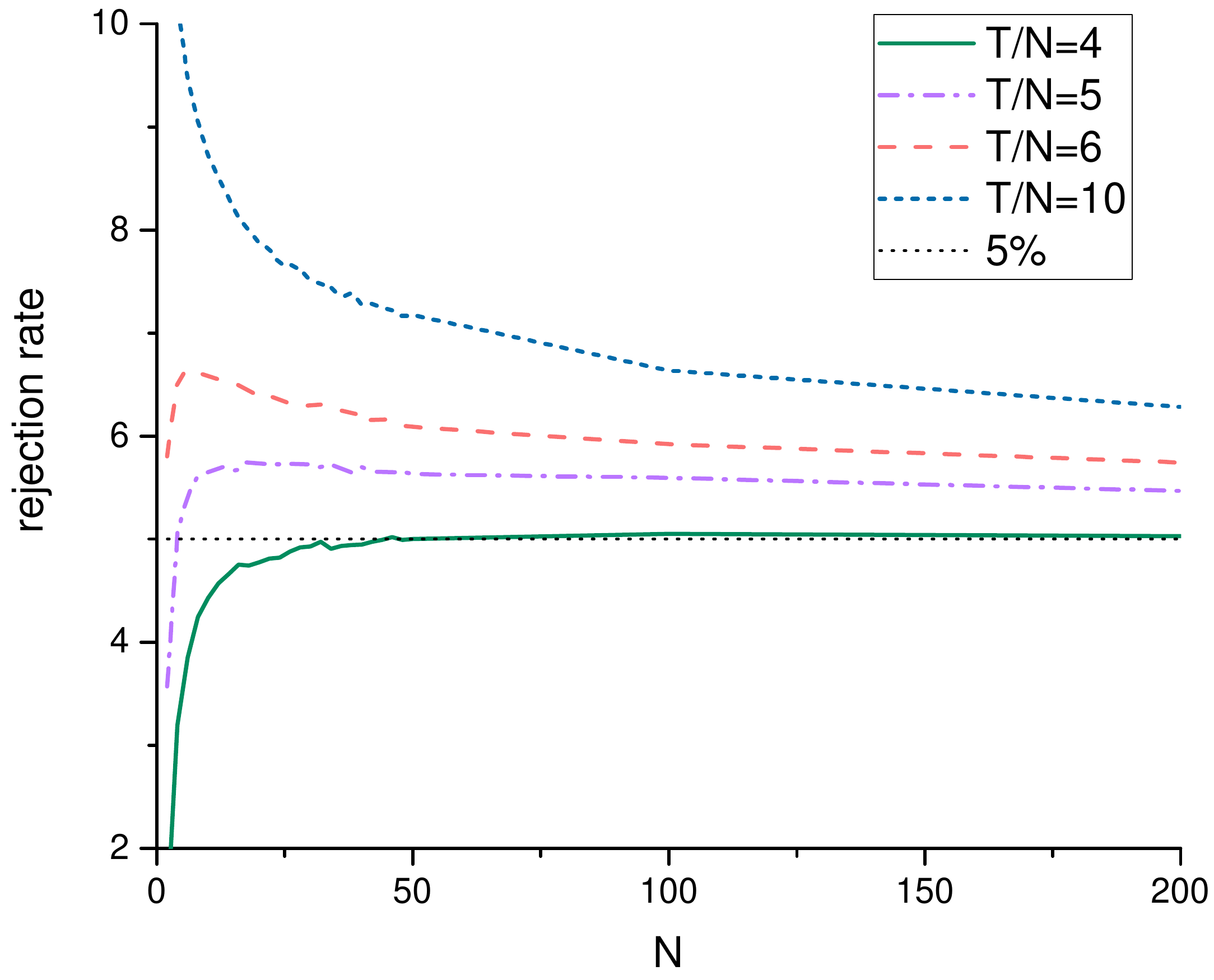}}}
\caption{Empirical size when $T/N$ is fixed and $N$ increases.}
\label{rej_rate_c}
\end{figure}

To improve the finite-sample behavior for large $c$ we suggest to ignore the $\frac{2}{N}$ correction in $\p$ and $\q$ in Theorem \ref{theorem_J_stat}, Eq.~\eqref{pq_def}, when $T/N$ is large. That is, to use simplified formulas instead: $\p=2,\,\q=T/N-1$. We recalculate the empirical rejection rate under the simplified formulas for $\p,\q$ in Figure \ref{rej_rate_c_simpl}. Under the simplified formulas, the larger is $c$, the smaller is over-rejection and the closer the curve is to $5\%$ line. As can be seen by comparing Figures \ref{rej_rate_c} and \ref{rej_rate_c_simpl}, for small $c$ the formula with the correction leads to better rejection rates, while for large $c$ it is the opposite, and we do not gain from finite sample correction. The conclusion is to use the simplified formula when $c=T/N$ is at least $6$.

\begin{figure}[h]
{\scalebox{0.5}{\includegraphics{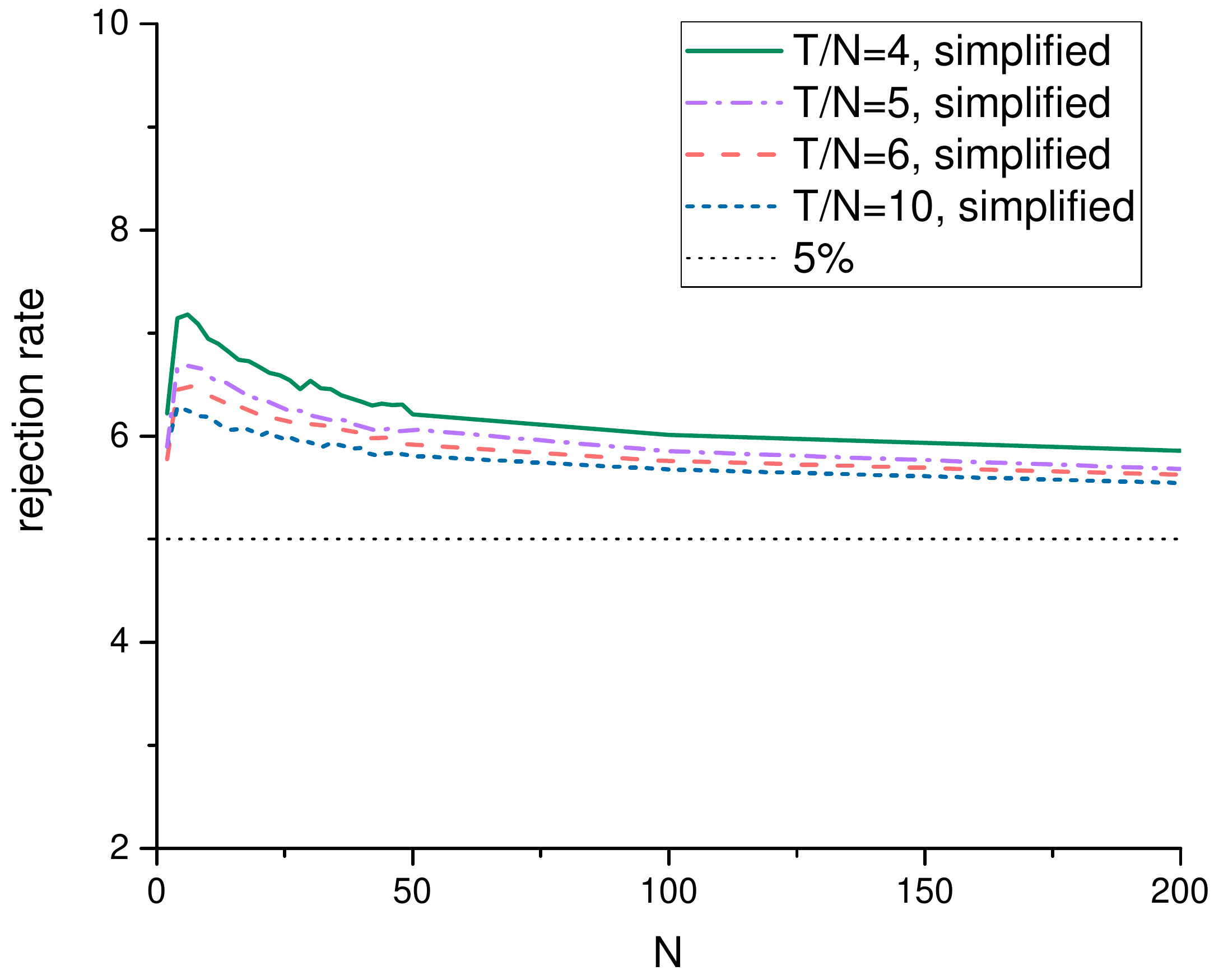}}}
\caption{Empirical size when $T/N$ is fixed and $N$ increases. Test based on $\p=2$, $\q=T/N-1$.}
\label{rej_rate_c_simpl}
\end{figure}

An important feature of Figures \ref{rej_rate_c} and \ref{rej_rate_c_simpl} is that as $N,T\to\infty$ with fixed $T/N$ the results improve towards the perfect match to $5\%$. This is in contrast to various finite sample corrections of Johansen's LR test used earlier. In particular, examination of Table $2$ in \citet{onatski_joe} reveals that each procedure has its own pairs $(N,T)$ where the results are close to perfect (empirical size of the test is very close to the desired $5\%$) and some others where the results are much worse, yet rejection rates in general do not improve as the sample size increases keeping $T/N$ fixed.


\subsection{Power} \label{Section_power}
\begin{figure}[h!]
{\scalebox{0.5}{\includegraphics{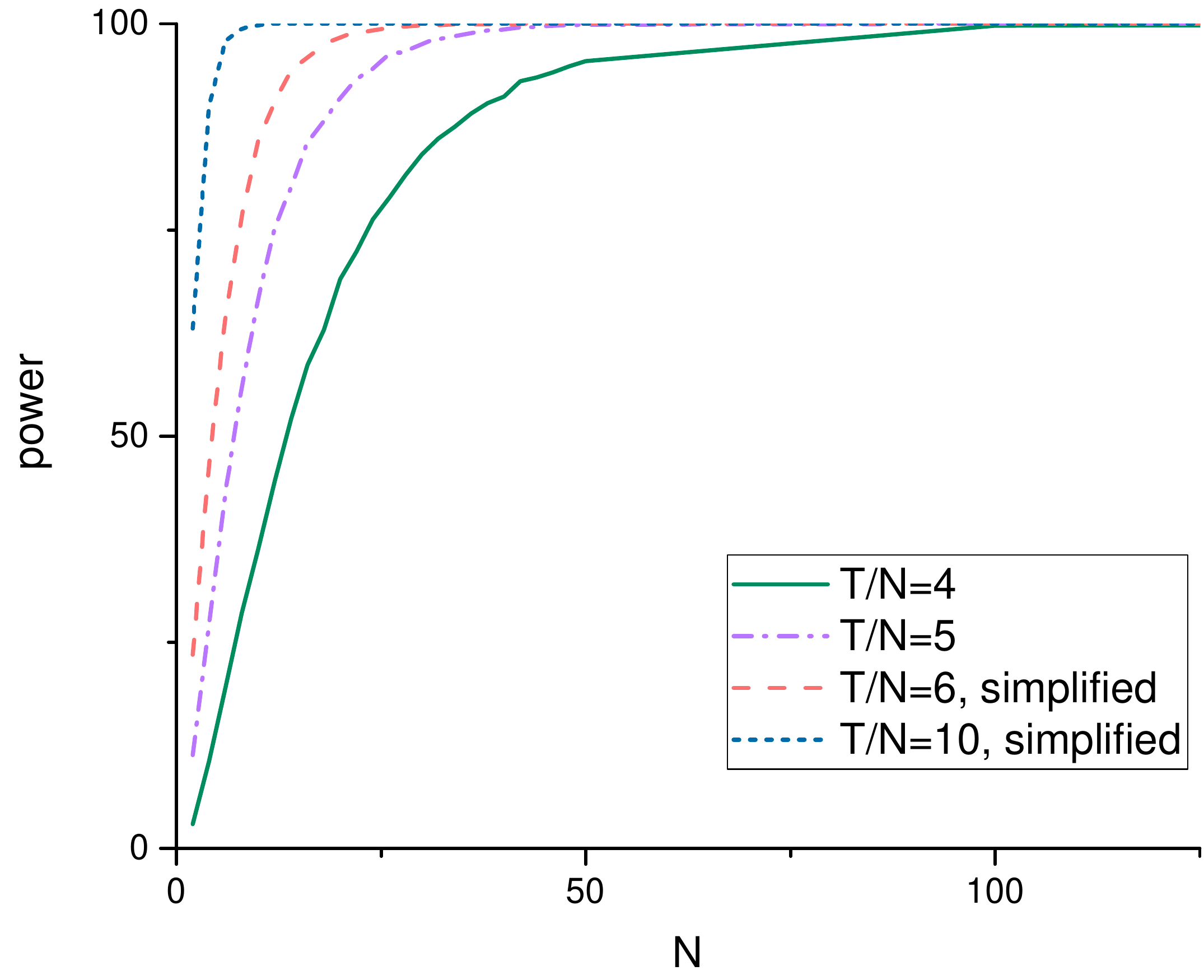}}}
\caption{Power against random alternative of $1$ cointegrating relationship when $T/N$ is fixed and $N$ increases.}
\label{power_asym}
\end{figure}
\begin{figure}[h!]
{\scalebox{0.5}{\includegraphics{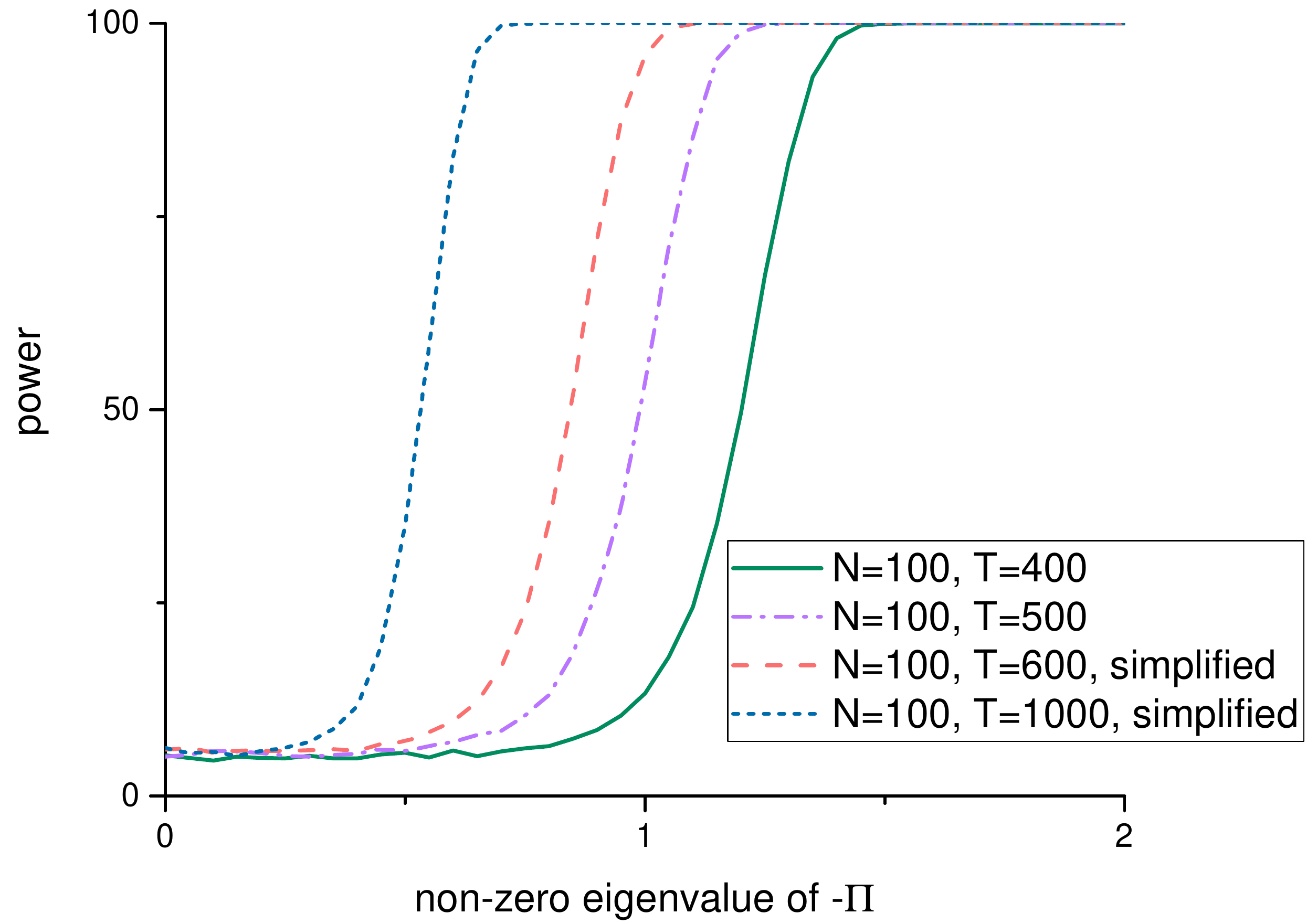}}}
\caption{Power against random alternative of $1$ cointegrating relationship when $\Pi$ is symmetric. To ensure stationarity of $\Delta X$ a non-zero eigenvalue of $-\Pi$ lies between $0$ and $2$.}
\label{power_sym}
\end{figure}

When one analyzes the power of a test, the question is which data generating process (dgp) to use under an alternative $H_1$. In our setting the space of alternative dgps has growing with $N$ dimension. Thus, there is no clear choice of the proxy alternative hypothesis to use in simulations. Hence, instead we proceed with various random alternatives.\footnote{Another approach would be to stick to some ad hoc matrix $\Pi$. Yet, we do not expect that to affect our simulations in a qualitatively significant way.}  Therefore, the corresponding power is also random. For illustrational convenience, we only report averages of those powers.

To analyze the power of our test we conduct several experiments. In all of them the errors $\eps_{it}$ are generated as i.i.d.~$\mathcal{N}(0,1)$. In the first experiment we randomly sample a matrix $\Pi$ of rank $1$. We do this by generating two uniformly random $N$--dimensional unit vectors, $u,v$, such that the non-zero eigenvalue of $uv^{\ast}$ is negative. Then we set $\Pi=uv^{\ast}$. By construction, $\Pi$ has rank $1$, singular value $1$, and one eigenvalue between $-1$ and $0$ (others are zero),\footnote{As $N$ goes to infinity the non-zero eigenvalue is of order $N^{-1/2}$.} so that $X_t$ is non-stationary, while $\Delta X_t$ is stationary. The average power constructed from such random alternative is shown in Figure \ref{power_asym}. As in the previous subsection we fix the ratio $T/N$ and plot the average power as a function of $N$. Following our recommendation, we use simplified formulas for $\p$ and $\q$ when $T/N\geq6$. We can see that the average power quickly reaches $100\%$ for all ratios of $T/N$. The larger is that ratio, the faster we reach $100\%$. This is due to the fact that higher ratio means higher time span. This gives the process more chances to accumulate the effects of the presence of cointegration.

In the second experiment, we randomly generate a symmetric matrix $\Pi$ of rank $1$. We do this by generating a uniformly random $N$--dimensional unit vector $v$. Then we set $\Pi=-\lambda vv^{\ast}$, where $\lambda$ goes from $0$ to $2$. The coefficient $\lambda$ equals to the non-zero eigenvalue of $-\Pi$. The fact that it lies between $0$ and $2$ guarantees that $X_t$ is non-stationary, while $\Delta X_t$ is stationary. Figure \ref{power_sym} shows the power as a function of $\lambda$ for $N=100$ and various values of $T$ corresponding to $T/N=\{4,5,6,10\}$ as in the previous experiments. The larger is $-\lambda$, the larger is power, which eventually reaches $100\%$. When $\lambda=0$, the dgp corresponds to the null $H_0$. Thus, all curves start from $\approx5\%$, which corresponds to the empirical size of the test. We can see that again the larger is $T/N$, the faster we reach $100\%$. The reason is the same as in the previous simulation.

\begin{figure}[t]
{\scalebox{0.5}{\includegraphics{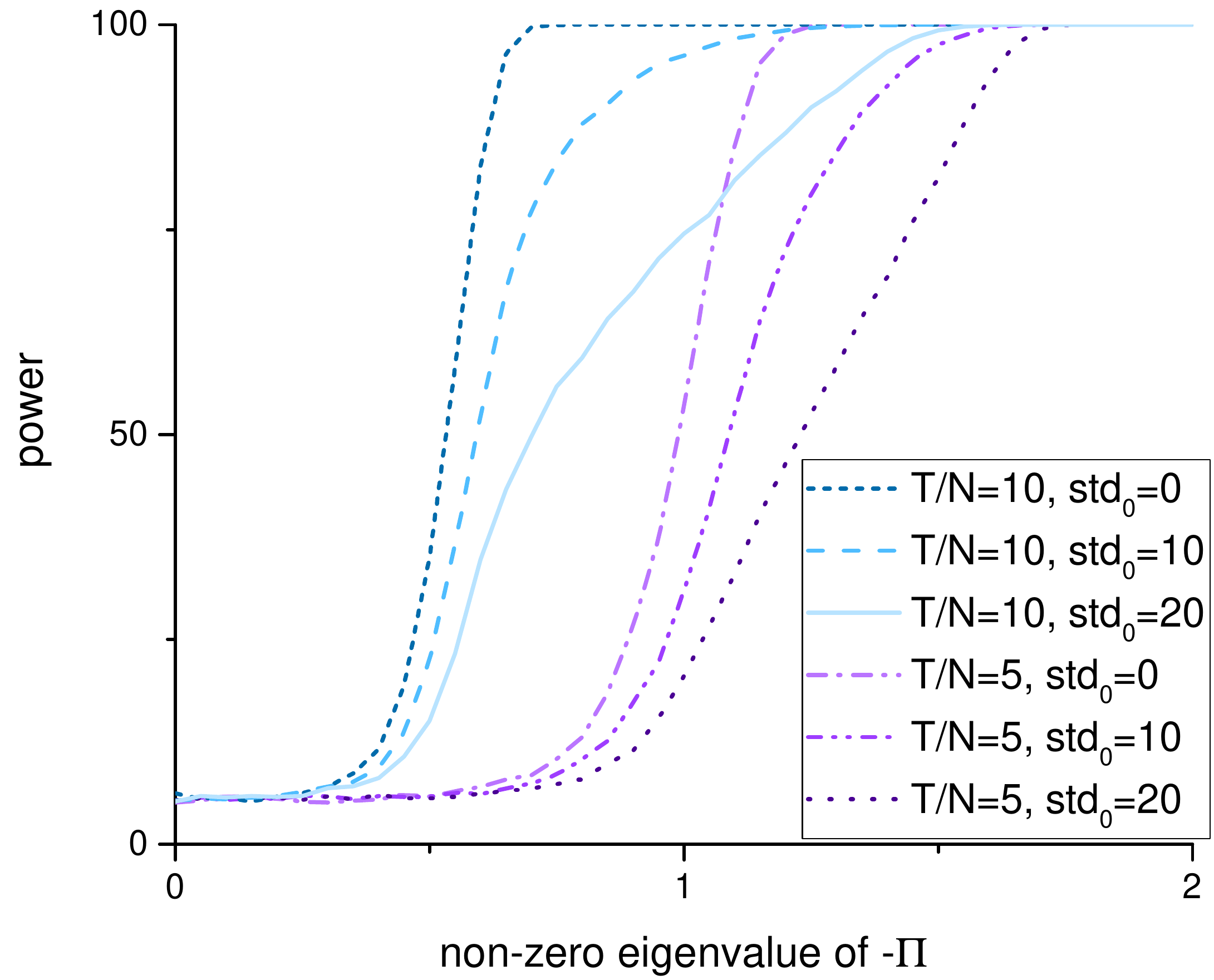}}}
\caption{Power against random alternative of $1$ cointegrating relationship when $\Pi$ is symmetric. Initial condition is chosen as $X_0=\text{std}_0\cdot\mathcal{N}(0,1)$, $N=100$.}
\label{power_X0}
\end{figure}

Drawing the intuition from the unit root testing literature (e.g., \citet{unitroot_X0_ME}, \citet{unitroot_X0_HLT}), we also analyze the sensitivity of power to the choice of initial condition, $X_0$.\footnote{Note that $X_0$ does not affect rejection rates in Section \ref{Section_rej_rate}, because under $H_0$ it cancels out.} In two previous experiments we set $X_0=0$. Figure \ref{power_X0} shows how the power against random alternative of $1$ cointegrating relationship for symmetric $\Pi$ changes for various magnitudes of $X_0$. To be more specific, we redo the same simulations as in the previous paragraph, but start with $X_{i0}\thicksim$ i.i.d.~$\text{std}_0\cdot\mathcal{N}(0,1)$ for each Monte Carlo draw. We consider $T/N=5$ and $T/N=10$. Curves with std$_0=0$ are the same as on Figure \ref{power_sym} (they are also represented with the same style and color on both pictures). We can see that the larger is the magnitude of $X_0$, as measured by its standard deviation std$_0$, the slower the power reaches $100\%$.

In contrast to the previous paragraph, the power against random alternative of $1$ cointegrating relationship when $\Pi$ is asymmetric (as in Figure \ref{power_asym}) does not exhibit any substantial changes with respect to the magnitude of $X_0$. Hence, we do not redraw the analogue of Figure \ref{power_asym} for non-zero $X_0$.

Overall, the simulations suggests the good asymptotic performance of our test procedure both under $H_0$ and $H_1$. The theoretical analysis of the power remains an open and challenging question.

\section{Empirical illustration}\label{Section_s&p}


In this section we illustrate our testing strategy by analyzing cointegration in log prices of various stocks. The search for cointegrations is a part of a stock market strategy called pairs trading, see, e.g., \citet{pairs_trading} and references therein. For us this is a convenient testing ground, as both $N$ and $T$ are large.

We use logarithms of weekly S$\&$P100 prices over ten years: $01.01.2010-01.01.2020$, which gives us $522$ observations across time. The time range is chosen so that we do not need to worry about potential structural breaks due to financial crisis of $2007-2008$ and due to COVID-19. S$\&$P100 includes $101$ (because one of its component companies, Google, has 2 classes of stock) leading U.S.~stocks with exchange-listed options. We use 92 of those stocks (those which were available for the whole ten years span and only one of two Google's stocks). More details on those stocks are in Section \ref{data_subsection} in Appendix. Therefore, $N=92$, $T=521$ and $T/N\approx 5.66$.

Figure \ref{Wachter_data} shows the histogram of eigenvalues which solve Eq.~\ref{eq_Johansen_equation} for the S$\&$P100 data. The key feature of this histogram is that it closely resembles the Wachter distribution, which density is shown by thick orange line in Figure \ref{Wachter_data}. This distribution governs the asymptotics of the eigenvalues of the Jacobi ensemble (see Section \ref{Section_Jacobi} in Appendix for the details). In particular, we see a precise match between supports of the histogram and of the Wachter distribution. The resemblance is in line with our Theorem \ref{Theorem_main_approximation}. Indeed, if we believe that the true data-generating process \eqref{var_1} has $\Pi=0$, then the theorem combined with the asymptotics of the Jacobi ensemble of Theorem \ref{Theorem_Jacobi_as} predicts convergence to the Wachter distribution. The conclusion is that Figure \ref{Wachter_data} shows a close match between theoretical predictions and real data. Simultaneously, this figure is consistent with $H_0$ of no cointegration.

\begin{figure}[h]
{\scalebox{0.7}{\includegraphics{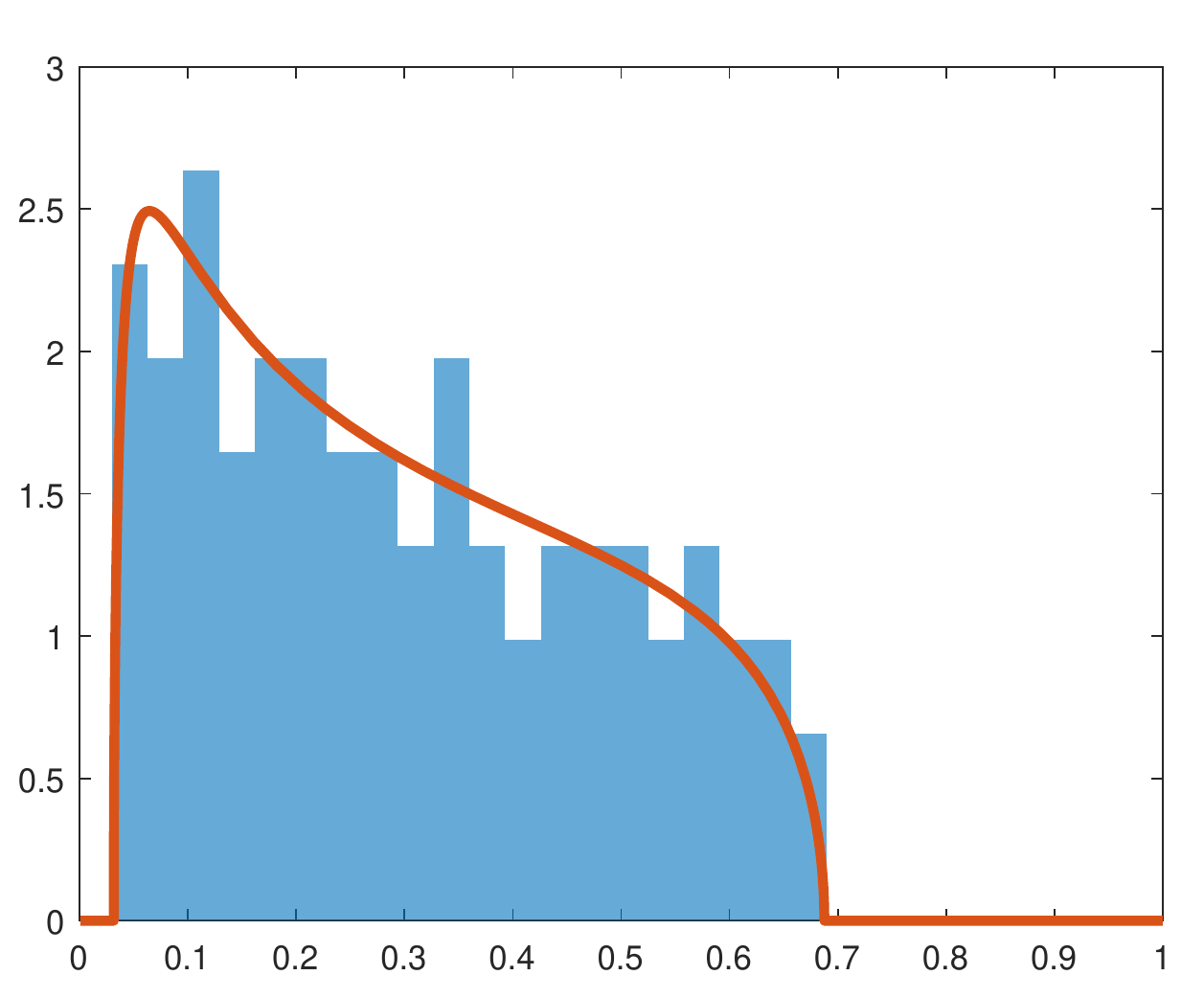}}}
\caption{Eigenvalues from S$\&$P data and Wachter distribution.}
\label{Wachter_data}
\end{figure}

We compute our test statistic for $r=1,2,3$ using the S$\&$P100 data. In neither of the cases the value of the statistic is large enough for a statistically significant rejection of $H_0$. If $r=1$, then the value is $-0.27$, which corresponds to approximately the $0.78$ quantile of the asymptotic distribution shown in Figure \ref{airy_density}. For $r=2,3$ the values are closer to the right tails of the distribution, but still below the (one-sided) $0.95$ quantiles. Hence, we do not see evidence towards the presence of cointegration in S$\&$P100 stock prices for the last 10 years.

\section{Extensions}\label{sec_extensions}

Let us describe possible extensions and modifications of our results. In Subsection \ref{Section_nonGauss} we consider non-Gaussian errors $\eps_t$. Subsection \ref{Section_no_trend} looks at the model without an intercept $\mu$ (linear trend in $X_t$) and considers the effect of de-trending vs.~no de-trending in such setting. Subsection \ref{section_vark} investigates the performance of our test when the true process follows higher order of autoregression. Finally, in Subsection \ref{Section_white} we discuss testing the hypothesis $\Pi=-\1_N$ using the same approach as for $\Pi=0$. 

\subsection{Non-gaussian errors}\label{Section_nonGauss}

The result of Theorem \ref{theorem_J_stat} is obtained under the assumption that the errors $\eps_t$ in Eq.~\eqref{var_1} are Gaussian. However, we believe that it should be possible to remove this restriction and it is reasonable to expect that the very same statement should hold for any (independent and identical across time $t$) distribution of $\eps_t$, as long as it has sufficiently many moments. The underlying reason for this belief is the so-called universality phenomenon in the random matrix theory: asymptotic local spectral characteristics of a random matrix are almost independent of the distributions of the matrix elements, see, e.g., \citet{ErdosYau}, \citet{Tao_Vu} for general reviews and \citet{HanPanZhang_2016}, \citet{HanPanYang_2018} for the recent work in contexts of multivariate analysis of variance and canonical correlations. In particular, for the Wigner matrices ($Y+Y^*$, where $Y$ is a square matrix with real i.i.d.~entries), it is known that the asymptotic behavior of the largest eigenvalues depends only on the first two moments (expectation and variance) of the distribution of an individual matrix element. Note, however, that our asymptotic result in Theorem \ref{theorem_J_stat} holds for any choice of the first two moment of Gaussian noise $\eps_t$: in Eq.~\eqref{var_1} the covariance matrix $\Lambda$ is arbitrary and any shift in expectation can be absorbed into the parameter $\mu$. Hence, we conjecture that the result of Theorem \ref{theorem_J_stat} would hold for any distribution of $\eps_t$, as long as it is sufficiently well-behaved.

In order to test this conjecture we made simulations for the case when elements of $\eps_t$ are non-gaussian, but i.i.d.~across both $i$ and $t$ (corresponding to a diagonal covariance matrix $\Lambda$). We ran Monte-Carlo simulations for three different distributions for $\eps_{it}$: uniform on the interval $[0,1]$, uniform on $3$ points $\{1,2,3\}$, and the product of two independent $\mathcal{N}(0,1)$ random variables. In each case for $T=900$, $N=300$, and small values of $r$, we do not see any significant changes in the distribution of $\sum_{i=1}^r \ln(1-\lambda_i)$ from the limit in Theorem \ref{theorem_J_stat}. However, things go differently when the distribution has heavy tails. In the forth experiment we chose $\eps_{it}$ to be Cauchy-distributed, and then the distribution of $\sum_{i=1}^r \ln(1-\lambda_i)$ dramatically changed from what we saw in the Gaussian case. Hence, we conclude, that the existence of at least some number of moments of the errors $\{\eps_t\}$ should be necessary for the validity of the conjecture. Rigorous proof of the conjecture remains an important problem for the future research.

\subsection{Model without trend}\label{Section_no_trend}

One of the important steps in our testing procedure is de-trending the data. Moreover, the exact form of the de-trending that we use (we subtract the slope of the line which connects $X_0$ at $t=0$ and $X_T$ at $t=T$) is an ingredient substantially used in the proof of Theorem \ref{theorem_J_stat}. Yet, we expect that this is just a technical artifact and statements similar to Theorem \ref{theorem_J_stat} should also be true with other forms of de-trending or in models where this step is not needed at all. Let us provide some evidence.

Our model allows for any linear trend and works even if the true value of $\mu$ is zero. Yet, when one has such prior knowledge it may seem natural to ignore the de-trending and de-meaning steps (Steps $1$ and $2$ in Section \ref{subsection_test}) and proceed without them. In this section we compare tests with and without de-trending for the model without $\mu$:
\begin{equation}\label{var1_no_trend}
\Delta X_t=\Pi X_{t-1}+\eps_t,\qquad t=1,\ldots,T.
\end{equation}
If both de-trending and de-meaning steps are omitted, we get (the small $r$ version of) the classical Johansen test statistic for the model \eqref{var1_no_trend}. If only de-trending is omitted, we get the Johansen test statistic corresponding to our original model \eqref{var_1}. When both de-trending and de-meaning are implemented, we get our procedure described in Section \ref{subsection_test}. To compare the asymptotic properties of those procedures, we perform a Monte Carlo simulation. Results are reported in Figure \ref{MC_trend}.

\begin{figure}[h]
\begin{subfigure}{.32\textwidth}
  \centering
  \includegraphics[width=1.0\linewidth]{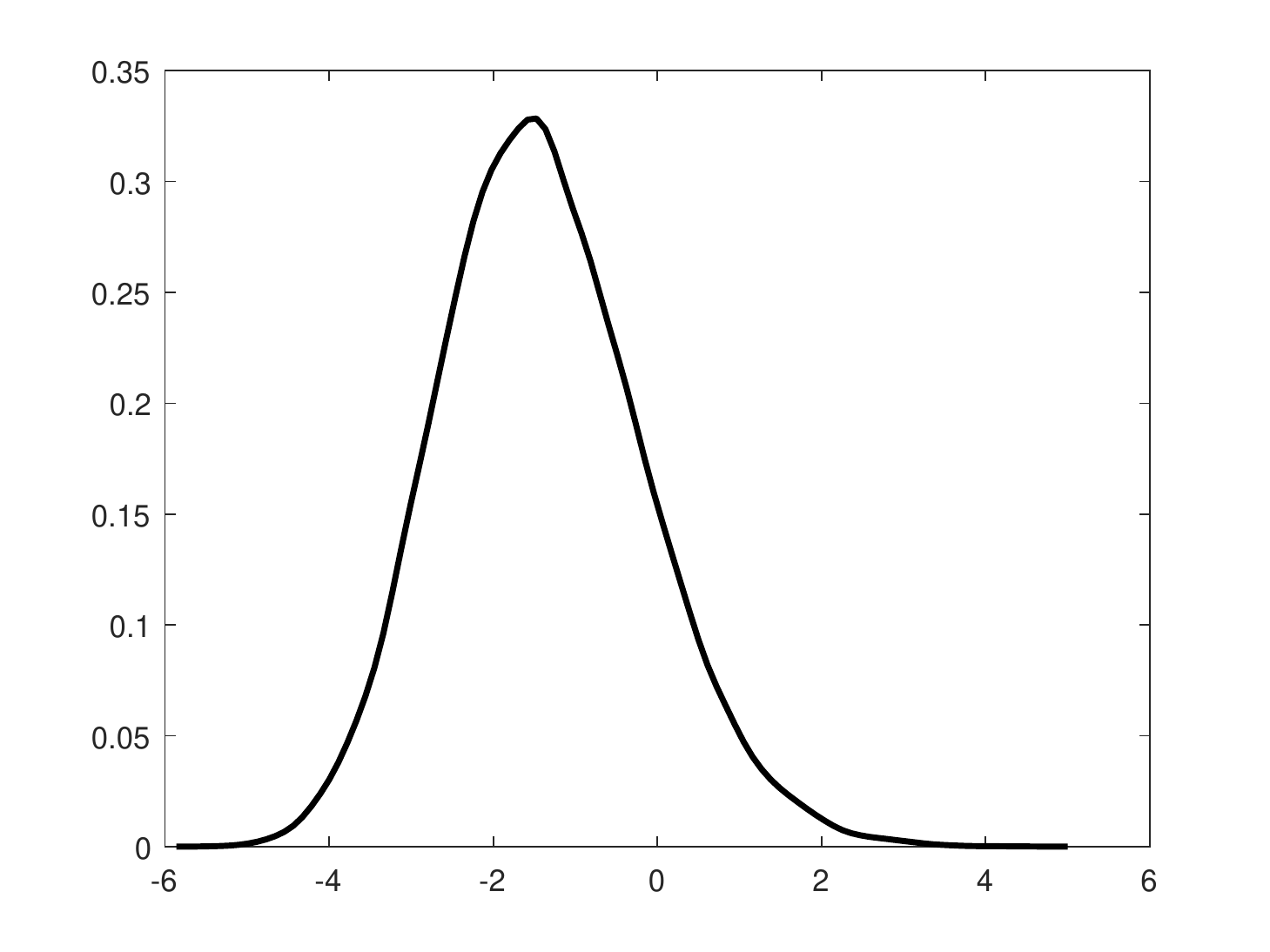}
  \caption{\scriptsize{No de-trending and de-meaning.}}
  \label{trend_sub1}
\end{subfigure}%
\begin{subfigure}{.32\textwidth}
  \centering
  \includegraphics[width=1.0\linewidth]{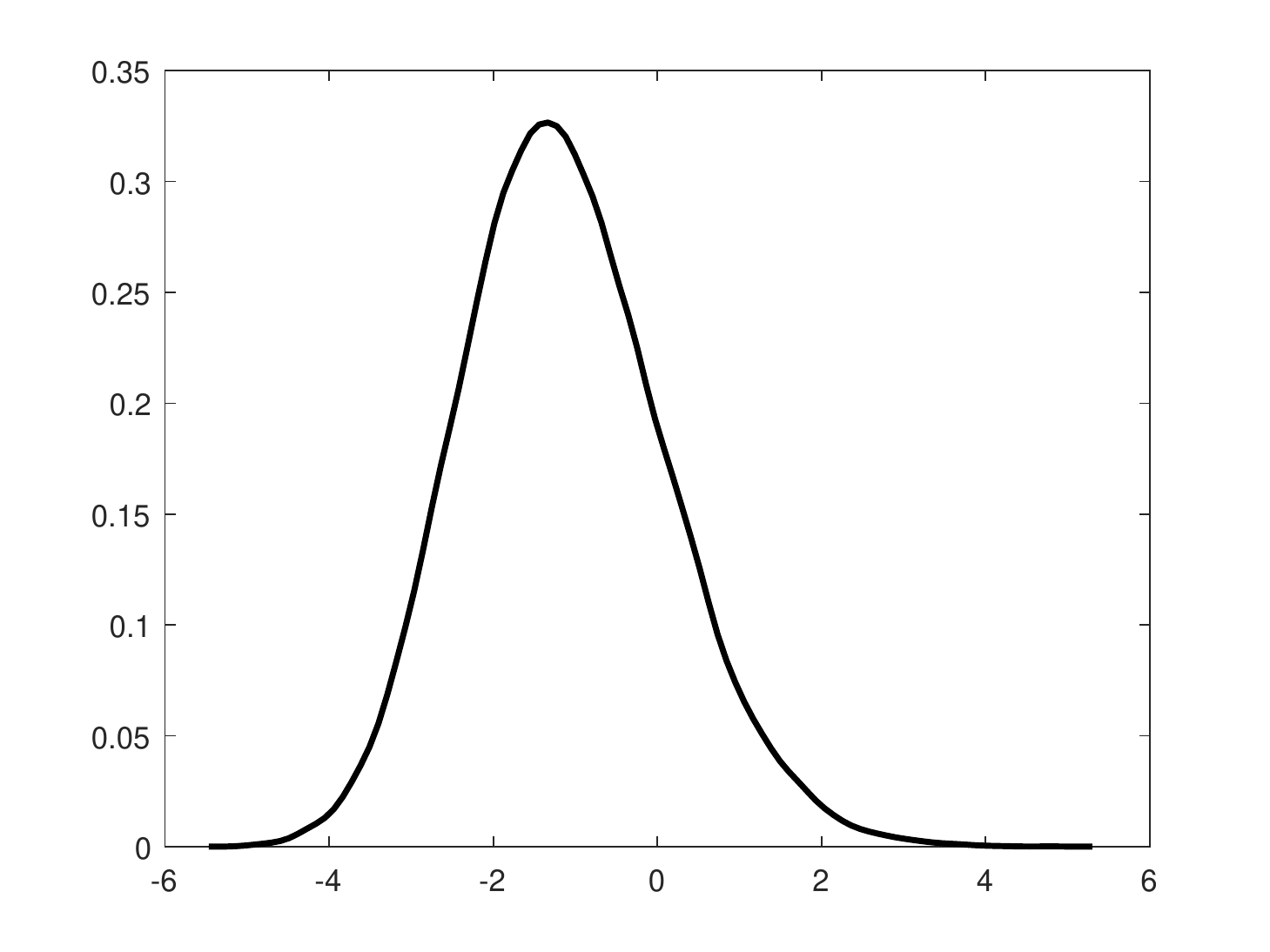}
  \caption{\scriptsize{De-meaning only.}}
  \label{trend_sub2}
\end{subfigure}
\begin{subfigure}{.32\textwidth}
  \centering
  \includegraphics[width=1.0\linewidth]{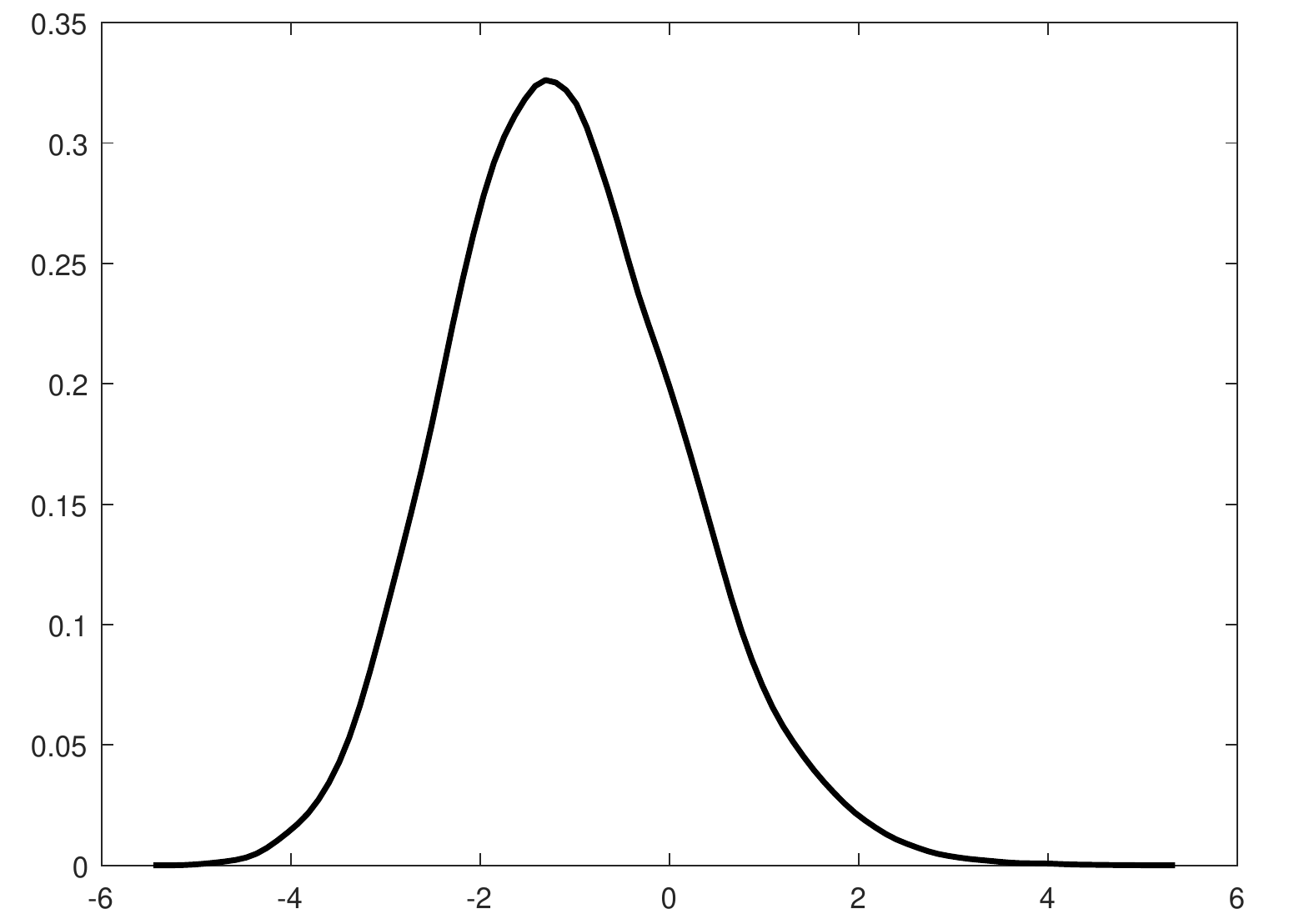}
  \caption{\scriptsize{De-trending and de-meaning.}}
  \label{trend_sub3}
\end{subfigure}
\caption{Asymptotic distribution of the rescaled $\ln(1-\lambda_1)$ under various testing procedures (Eq.~\eqref{eq_statistic_limit} with $r=1$).
Data generating process: $\Delta X_{it}=\eps_{it}$, $\eps_{it}\thicksim$ i.i.d.~$\mathcal{N}(0,1)$, $T=500$, $N=100$, $MC=50,000$ replications.}
\label{MC_trend}
\end{figure}

\begin{figure}[h]
\begin{subfigure}{.45\textwidth}
  \centering
  \includegraphics[width=1.0\linewidth]{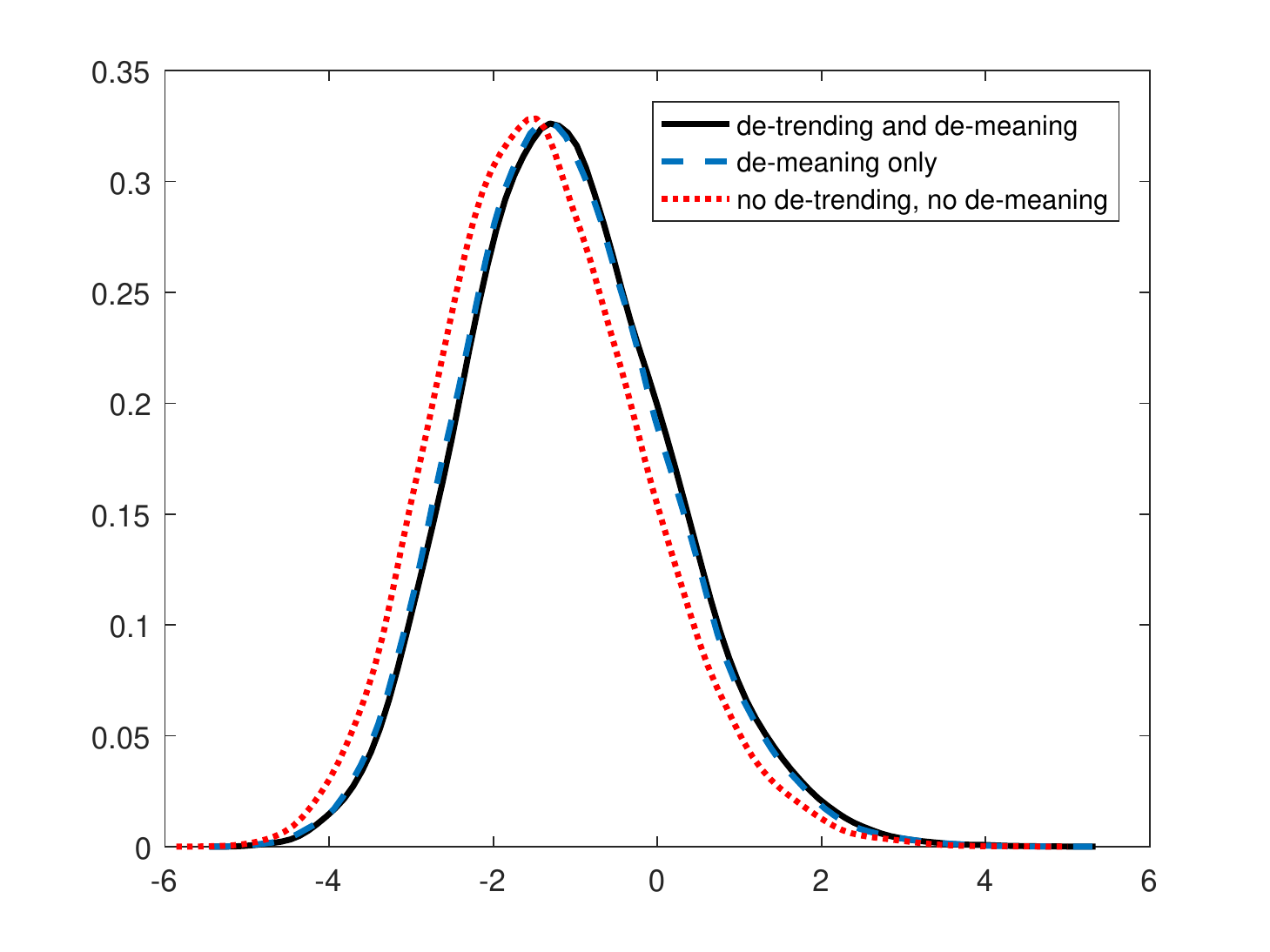}
  \caption{Three curves together.}
  \label{trend_sub_all}
\end{subfigure}%
\begin{subfigure}{.45\textwidth}
  \centering
  \includegraphics[width=1.0\linewidth]{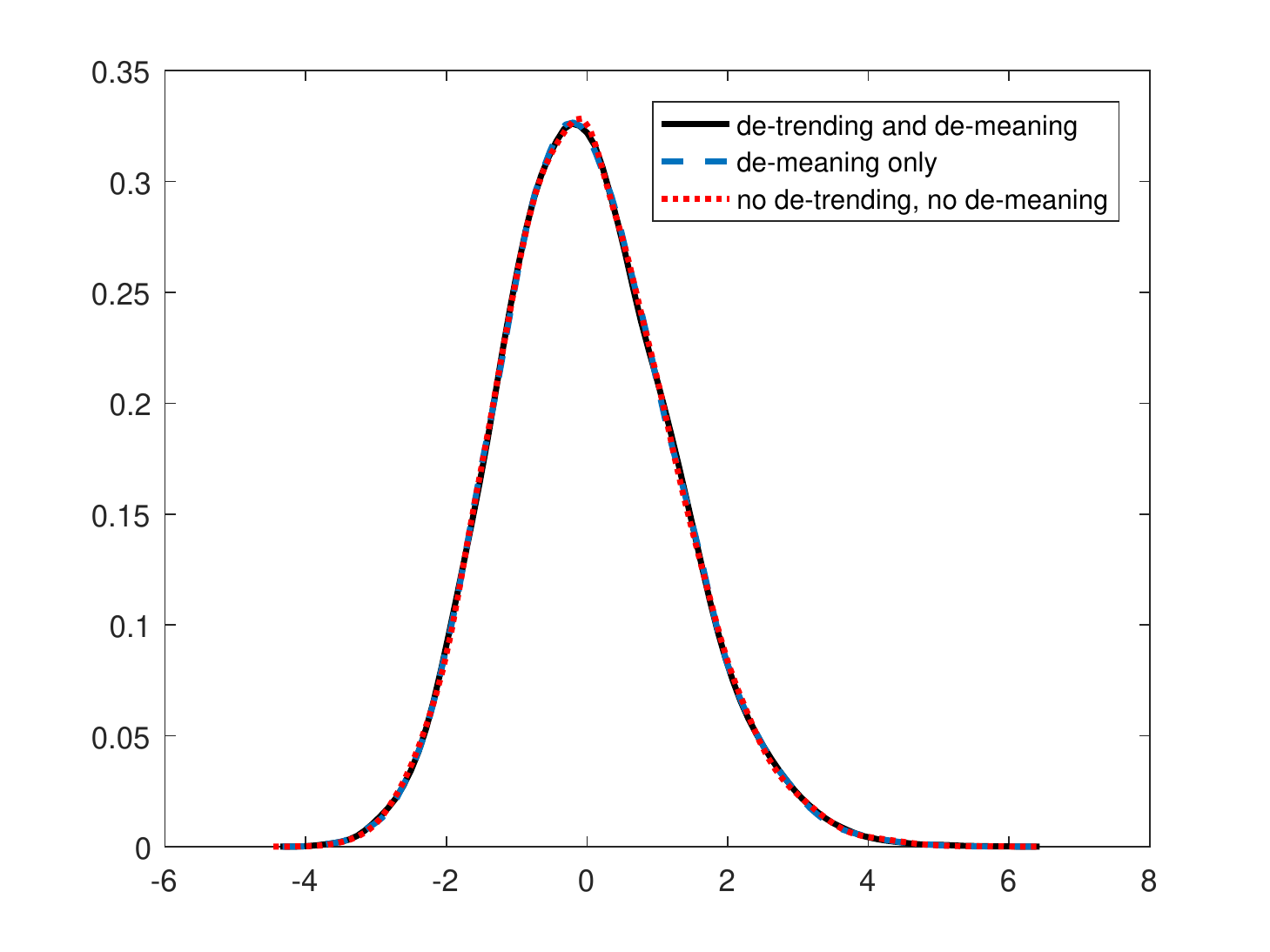}
  \caption{Mean is subtracted from each of curve.}
  \label{trend_sub_all_mean}
\end{subfigure}
\caption{Asymptotic distribution of the rescaled $\ln(1-\lambda_1)$ under various testing procedures with means normalized to zero (Eq.~\eqref{eq_statistic_limit} with $r=1$).
Data generating process: $\Delta X_{it}=\eps_{it}$, $\eps_{it}\thicksim$ i.i.d.~$\mathcal{N}(0,1)$, $T=500$, $N=100$, $MC=50,000$ replications.}
\label{MC_trend_mean}
\end{figure}

As Figure \ref{MC_trend} suggests, the densities have almost identical shape. Figure \ref{MC_trend_mean} plots all three of them together as well as their versions with subtracted mean. As we can see, after subtracting the mean, all three densities are identical. This suggests robustness of the Airy$_{1}$ point process in our asymptotic results of Theorem \ref{theorem_J_stat} (which correspond to Figure \ref{trend_sub3})  and predicts that some modifications, which preserve the limiting distribution, are possible.

Let us emphasize that our limit theorems currently only apply to Figure \ref{trend_sub3}, but not to the settings of Figures \ref{trend_sub1} and \ref{trend_sub2}. The mismatch of the means in Figure \ref{trend_sub_all} makes one suspect that some modifications of the constants in the asymptotic theorems are needed as soon as we start slightly adjusting the setting.

\subsection{Higher order of VAR}\label{section_vark}

In this subsection we discuss the performance of our test when the true data generating process (dgp) is VAR($k$), $k>1$. That is
\begin{equation}\label{var_k_extension}
\Delta X_t=\sum\limits_{i=1}^{k-1}\Gamma_i\Delta X_{t-i}+\Pi X_{t-k}+\mu+\eps_t,\qquad t=1,\ldots,T,
\end{equation}
and the no cointegration situation corresponds to $\Pi\equiv 0$.

Even if $k>1$, we expect to see the Tracy-Widom distribution and marginals of Airy$_1$ process (as in Theorem \ref{theorem_J_stat}) in the asymptotic behavior of the squared sample canonical correlations from Section \ref{subsection_test} under mild restrictions on $\Gamma_i$. The belief is based on the universality intuition of the random matrix theory. However, we do not expect the scalings (such as the coefficients $c_1$ and $c_2$ in Theorem \ref{theorem_J_stat}) to remain the same. The most plain analogy is the dependence of centering and scaling in the classical Central Limit Theorem on the underlying process. Closer to our context is the asymptotic behavior of sample covariance matrices: when the data is i.i.d., the empirical distribution of the eigenvalues of the sample covariance matrix converges to the Marchenko-Pastur law (and the largest eigenvalues concentrate near the right edge of this distribution), while data with general covariance structure leads to much richer limits, see \citet[Chapter 4]{Bai_Silverstein} and references therein. Figuring out (even heuristically) any formulas for the scaling coefficients $c_1$ and $c_2$ as functions of $\Gamma_i$ is a challenging open problem for the future research.

There is an important family of cases where one can hope that the formulas for $c_1$ and $c_2$ from Theorem~\ref{theorem_J_stat} remain valid (perhaps, with minor modifications). This is when $\Gamma_i$ are small and evolution of $X_t$ given by \eqref{var_k_extension} can be treated as a small perturbation of the VAR($1$) process. One way to formalize the ``smallness' of $\Gamma_i$ is by requiring them to be of small rank (cf.\ discussion of rank in Section \ref{Section_discussion}) and of small norm.

In order to investigate the above conjecture, we run Monte Carlo simulations for $k=2$ and $\Gamma_1$ of rank $1$. We consider the null $\Pi\equiv0$ and calculate our test statistic based on $r=1$ under VAR($2$) and VAR($1$) data generating processes and compare their asymptotic distribution. Under VAR($1$) $\Gamma_1\equiv0$, while for VAR($2$) we use $\Gamma_1=0.5E_{11}$ and $\Gamma_1=0.5E_{12}$, where $E_{ij}$ is a matrix which has $1$ on the intersection of row $i$ and column $j$ and $0$ everywhere else; the components of the noise $\eps_{it}$ are  i.i.d.~$\mathcal{N}(0,1)$. The asymptotic distribution of our test statistic with $r=1$ for various dgps is shown in Figure \ref{test_for_VAR2}. We see that the distributions on each panel of Figure \ref{test_for_VAR2} are close to each other. Hence, testing based on Theorem \ref{theorem_J_stat} in such a VAR($2$) setting remains valid. Yet, if we consider more general situation of $\Gamma_1=\theta E_{11}$ or $\Gamma_{1}=\theta E_{12}$, then we observe in simulations (not shown) that the quality of approximations significantly deteriorates as $\theta$ grows to $1$. 

\begin{figure}[h]
\begin{subfigure}{.45\textwidth}
  \centering
  \includegraphics[width=1.0\linewidth]{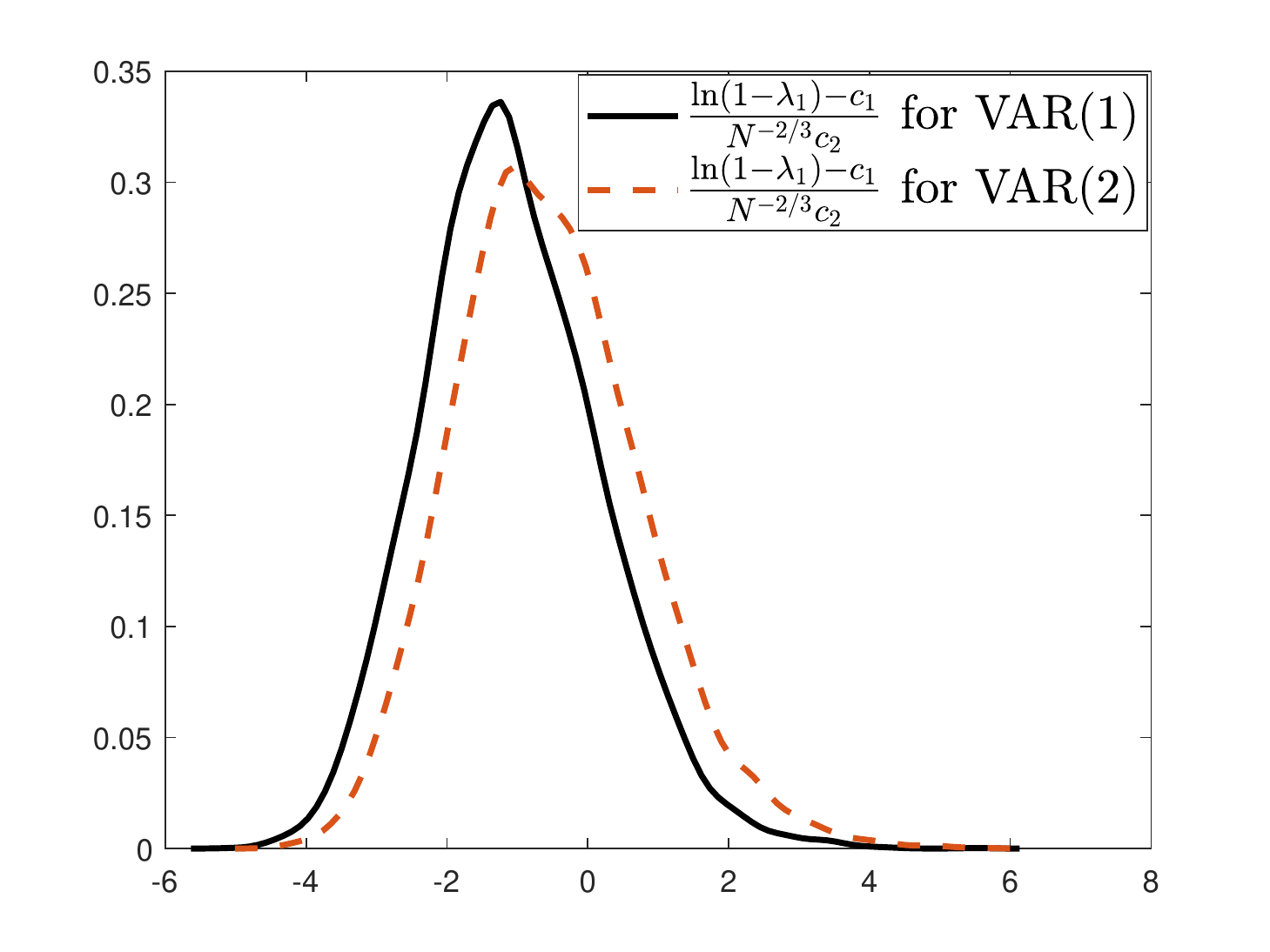}
  \caption{$\Gamma_1=0.5E_{11}$.}
  \label{VAR2E11}
\end{subfigure}%
\begin{subfigure}{.45\textwidth}
  \centering
  \includegraphics[width=1.0\linewidth]{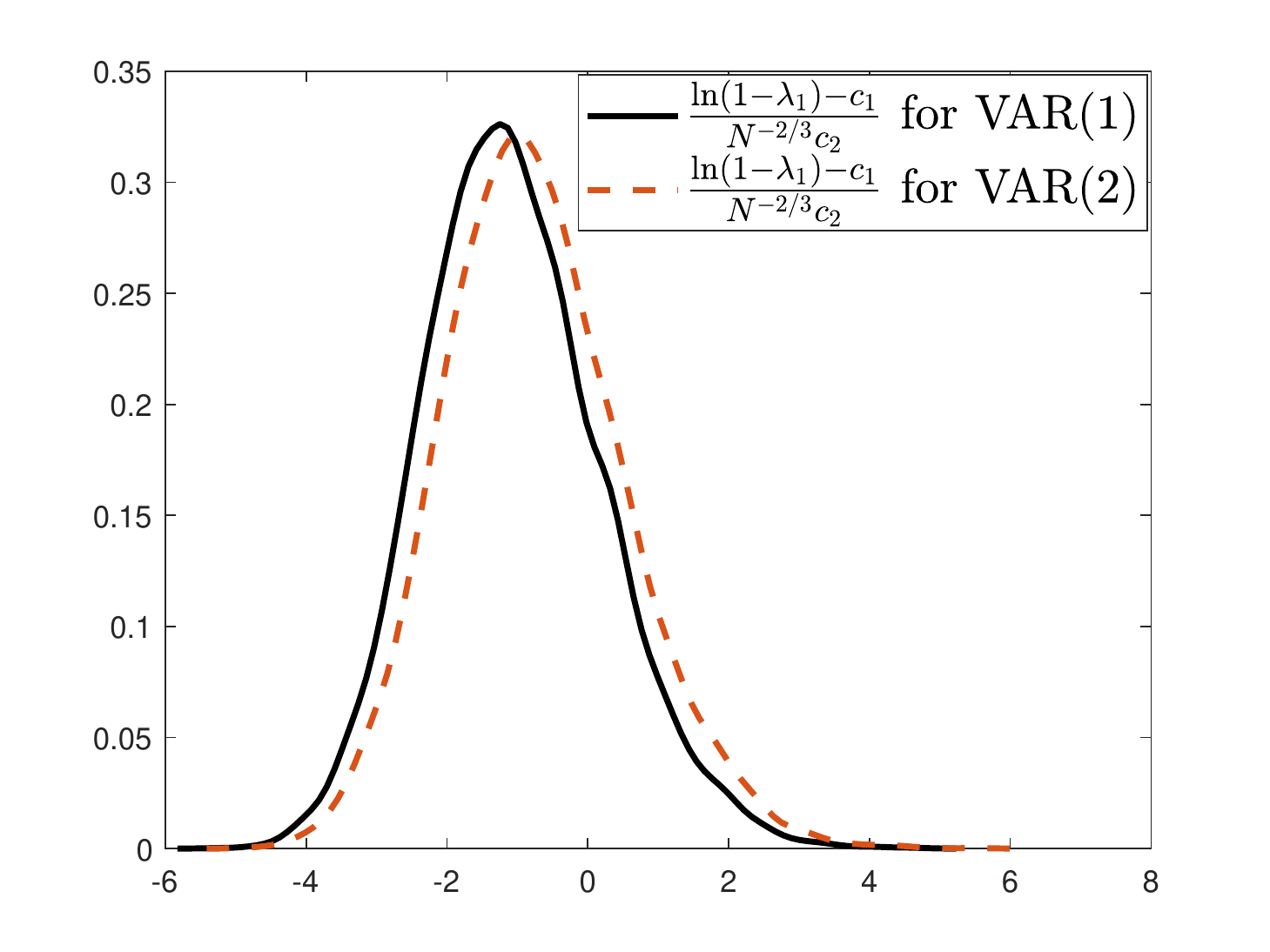}
  \caption{$\Gamma_1=0.5E_{12}$.}
  \label{VAR2E12}
\end{subfigure}
\caption{Asymptotic distribution of the rescaled $\ln(1-\lambda_1)$ (Eq.~\eqref{eq_statistic_limit} with $r=1$) under VAR($1$) with $\Pi=0$ and VAR($2$) with $\Pi=0$ and $\Gamma_1$ of rank $1$.\\ Data generating process: $\eps_{it}\thicksim$ i.i.d.~$\mathcal{N}(0,1)$, $T=500$, $N=100$, $MC=10,000$ replications.}
\label{test_for_VAR2}
\end{figure}

\subsection{Testing for white noise hypothesis in VAR($1$) setting}\label{Section_white}

The main result of Theorem \ref{theorem_J_stat} is a development of a test for the hypothesis $\Pi=0$ in the VAR($1$) model \eqref{var_1}. One could also
try to understand for which other $\Pi$ can the testing be possible. The asymptotic distribution depends on the choice
of $\Pi$, and it is not possible to estimate $\Pi$ consistently in our regime of $T$ and $N$ growing to infinity proportionally. Thus, for general $\Pi$ the problem seems infeasible at this point. However, there is another particular choice of $\Pi$ for
which an approach very similar to Theorem \ref{Theorem_main_approximation} still works:
$\Pi=-\1_N$. Denoting this hypothesis $H_0^{w.n.}$, where $w.n.$ stays for the white noise, the data generating process (Eq.~\eqref{var_1}) becomes:
\begin{equation}
\label{eq_wn_hypothesis}
 H_0^{w.n.}:\qquad X_t=\mu+\eps_t,\qquad t=1,\ldots,T.
\end{equation}
In other words, we are now testing the hypothesis that the time series $X_t$ is independent across time $t$ against various VAR($1$) alternatives. Here is one setup where such testing can be relevant. Suppose that we want to forecast some variable $Y$ and we chose some model for it. After estimating parameters of the model we obtain residuals $X_t$. If we know that $X_t$ are independent, then they are unforecastable and we cannot further improve our forecasting model.
To check the above we can take the residuals $X_t$ and then apply our white noise hypothesis testing procedure.

Let us introduce an adaptation of the Johansen statistic to $H_0^{w.n.}$.
As in Eq.~\eqref{eq_demean}, we use the notation $\mathcal P$ for the de-meaning operator projecting on the hyperplane orthogonal to $(1,1,\dots,1)$.

Following \citet{johansen1988,johansen1991} we are going to use the de-meaned data $X_t \mathcal P$. For the increments $\Delta X_t$ in addition to the conventional de-meaning we use an extra modification: we deal with \emph{cyclic increments} $\Delta^c X_t$ defined as:
$$
 \Delta^c X_t=\begin{cases} X_{t+1}-X_{t}, & t=1,2,\dots,T-1,\\ X_{1}-X_T, & t=T.\end{cases}
$$
While it might seem bizarre to subtract the last observation from the first one, if we recall that our current hypothesis of interest \eqref{eq_wn_hypothesis}  ignores the time ordering, then this becomes less controversial. Note also a shift of index by $1$, as compared to the conventional $\Delta X_t$, which is compensated by the lack of shift $t\to t-1$ in $X_t$, as compared to Eq.~\eqref{eq_detrending}.  Our choice of definition of $\Delta^c$ is important for the following precise asymptotic results. As in Section \ref{Section_setting}, the conventional Johansen statistic should be thought of as a finite rank perturbation of the modified version that we now introduce.
\begin{equation}
\label{eq_modified_Joh_matrices_wn}
 S_{00}^{w.n.}=\Delta^c X \mathcal P (\Delta^c X)^{*},\quad S_{01}^{w.n.}=\Delta^c X \mathcal P X^*, \quad S_{10}^{w.n.}=X \mathcal P (\Delta^c X)^*,\quad S_{11}^{w.n.}=X \mathcal P X^*,
\end{equation}
We further define $N$ numbers $\lambda_1\ge \lambda_2\ge \dots\ge \lambda_N$ as $N$ roots to the equation
\begin{equation}
\label{eq_Johansen_equation_wn}
 \det\bigl(S_{10}^{w.n.} (S_{00}^{w.n.})^{-1} S_{01}^{w.n.}-\lambda S_{11}^{w.n.}\bigr)=0.
\end{equation}
Equivalently, $\{\lambda_i\}$ are eigenvalues of $S_{10}^{w.n.} (S_{00}^{w.n.})^{-1} S_{01}^{w.n.}(S_{11}^{w.n.})^{-1}$.
\begin{theorem}
\label{Theorem_white_noise_approximation}
Suppose that $T,N\to\infty$ in such a way that $T>2N$ and the ratio $T/N$ remains bounded. Under the hypothesis  $H_0^{w.n.}$ one can couple (i.e.~define on the same probability space) the eigenvalues $\lambda_1\ge \lambda_2\ge \dots\ge \lambda_N$ of the matrix $S_{10}^{w.n.} (S_{00}^{w.n.})^{-1} S_{01}^{w.n.}(S_{11}^{w.n.})^{-1}$ and eigenvalues $x_1\ge \dots\ge x_N$ of Jacobi ensemble $\J(N;\frac{T-N-1}{2}, \frac{T-2N}{2})$ in such a way that  for each $\epsilon>0$ we have
 $$
   \lim_{T,N\to\infty} \mathrm{Prob}\left( \max_{1\le i \le N} |\lambda_i-x_i|< \frac{1}{N^{1-\epsilon}}\right)=1.
 $$
\end{theorem}
The proof of Theorem \ref{Theorem_white_noise_approximation} follows a similar strategy as Theorem \ref{Theorem_main_approximation} and we refer to Section \ref{Section_white_proof} in Appendix for details; in particular, the proofs rely on yet another novel appearance of the Jacobi ensemble.

The remaining straightforward step to obtain the asymptotics of various statistics built on  the eigenvalues $\{\lambda_i\}$ is to combine Theorem \ref{Theorem_white_noise_approximation} with asymptotic results for the Jacobi ensemble presented in Section \ref{Section_Jacobi}.  This is in the spirit of Theorem \ref{theorem_J_stat}.

Note that the hypothesis $H_0^{w.n.}$ implies the maximal amount of cointegrating relationships: each of the $N$ components of $X_t$ is already stationary. A reasonable alternative hypothesis $H_1$ is the presence of $N-r$ cointegrating relationships. For simplicity, let us concentrate on the case $r=1$. Then the alternative can be also interpreted as a presence of a single growing factor. In this situation we expect the \emph{smallest} eigenvalue $\lambda_N$ to be a good test statistic. We see in numerical simulations that $\lambda_N$ is bounded away from $0$ under $H_0^{w.n.}$. It can also be formally proved by combining Theorem \ref{Theorem_white_noise_approximation} with asymptotics of the Jacobi ensemble from Section \ref{Section_Jacobi}.  Thus, if $\lambda_N$ is close to $0$, then we are able to reject $H_0^{w.n.}$. The same simulations indicate that $\lambda_N$ starts to be close to $0$ when there are at most $N-1$ cointegrating relationships. Hence, the test based on $\lambda_N$ should have a good asymptotic power. We leave rigorous justifications of this observation till future research, and for now only mention the following heuristics: the stationary linear combinations of $X_t$ are strongly correlated with the same linear combinations of $\Delta^c X_t$; on the other hand, the growing linear combinations of $X_t$ have very weak correlation with the same linear combinations of $\Delta^c X_t$ (cf.~correlations of a one dimensional random walk with its increments). Hence, if the latter are present, the smallest canonical correlations of $X \mathcal P$ and $\Delta^c X\mathcal P$, which coincide with the eigenvalues of the matrix $S_{10}^{w.n.} (S_{00}^{w.n.})^{-1} S_{01}^{w.n.} (S_{11}^{w.n.})^{-1}$, should become small leading to close to $0$ value of $\lambda_N$.

\section{Conclusion}
\label{Section_conclusion}

The paper presents a cointegration test which has desirable empirical size when $N$ and $T$ are of the same magnitude. To our knowledge, this is a first paper which constructs and analyzes asymptotic properties of a test that does not suffer from significant distortions (such as over-rejection) for comparable $N$ and $T$. The test is based on the Johansen LR test and incorporates some additional steps. First, our procedure reinforces the importance of de-trending in cointegration testing. It turns out, that de-trending is crucial for deriving desirable asymptotic properties. (E.g., only after de-trending one can rewrite the lagged process as a linear function of its first differences.) Second, our asymptotic results reveal and explain an unexpected connection between the Johansen cointegration test and the Jacobi ensemble --- a classical ensemble of the random matrix theory whose previous appearances in statistics include multivariave analysis of variance (MANOVA) and sample canonical correlations for independent sets of data.

On the theoretical side the next step would be to go from null hypothesis of zero cointegration to analyzing the behaviour of our test under $r$ cointegrations. This will allow us to calculate the power of the test, reinforcing our simulational findings in Section \ref{Section_power},  as well as to perform tests of $r$ versus $r+1$ cointegrations.

On the empirical side it would be interesting to apply our test to other data sets beyond what is presented in Section \ref{Section_s&p}. Annual cross-country data provides a natural example of our setting where the number of years and countries is comparable.  Another example arises if one considers network-type settings which evolve over time (e.g., as in \citet{bykh}). Data on trade or on foreign direct investment can potentially be non-stationary, especially if we focus on largest and the most active countries. Moreover, although such monthly data is available, for many countries it only covers $~20$ years. Thus, we have $T\approx200$. If we look at directed pairs across $10$ largest countries, this gives us $N=90$ cross-section units, which fits ideally in our setting.
Classical cointegration tests are known to perform poorly in the above settings. However, the asymptotic results of our paper open up a possibility of detecting the presence of cointegration in such time-series data.


\section{Appendix}

\subsection{New matrix models for the Jacobi ensemble}

Recall that our testing procedure relies on the squared sample canonical correlations between two correlated data sets. As we later show in the proofs of Propositions \ref{Prop_gaussian_rotation} and \ref{Prop_gaussian_rotation_white}, an equivalent point of view is that we deal with eigenvalues of a product of two orthogonal projectors $P_1$ and $P_2$, where $P_1$ is a projection on a random $N$-dimensional subspace $\V$ of a $T$--dimensional space and $P_2$ is a projection on $U \V$, where $U$ is a certain deterministic linear operator.

We randomize this problem by replacing $U$ with a \emph{random} operator, whose spectrum is close to $U$. The randomized problem turns out to be exactly solvable --- the eigenvalues of new $P_1 P_2$ coincide with the classical Jacobi ensemble. The goal of this section is to prove this fact. The choices of $U$ and $\V$ depend on the hypothesis that we are testing and, hence, we need several theorems. Throughout this section we are going to deal both with real and complex matrices. According to the customary random matrix theory notation, they are referred to as $\beta=1$ and $\beta=2$ cases, respectively.

In what follows $\1_n$ means the identity matrix in $n$--dimensional space. Sometimes we omit $n$ and write simply $\1$ when the dimension is clear from the context. For a matrix $X$ we let $[X]_{NN}$ to be its top-left $N\times N$ corner. The parameter $\T$ used in this section is related to $T$ of the main text through $\T=T-1$.

Throughout this section we repeatedly change the coordinates in various measures, which produces a factor given by the absolute value of the Jacobian of the transformation. We rely on the computation of the Jacobian of the multiplicative change of variables in the space of matrices. We need three forms of it, where in each of them $Q$ stays for a $n\times n$ matrix:
\begin{itemize}
 \item The map $Z\mapsto Q Z$ on $n\times m$ matrices has the Jacobian
 \begin{equation}
 \label{eq_mult_Jacobian}
 \left| \frac{\partial (Q Z)}{\partial Z}\right| = |\det Q|^{\beta m},
 \end{equation}
  \item The map $Z\mapsto Q Z Q^*$ from the space of $n\times n$ symmetric (Hermitian if $\beta=2$) matrices to itself has the Jacobian
 \begin{equation}
\label{eq_sym_Jacobian}
\left|\frac{\partial (Q Z Q^*)}{\partial Z}\right|=|\det Q|^{\beta n+2-\beta}=\begin{cases} |\det Q|^{2 n},& \beta=2,\\ |\det Q|^{n+1},& \beta=1.\end{cases}
\end{equation}
 \item The map $Z\mapsto Q Z Q^*$ from the space of $n\times n$ skew-symmetric (skew-Hermitian if $\beta=2$) matrices to itself has the Jacobian
 \begin{equation}
\label{eq_skew_Jacobian}
\left|\frac{\partial (Q Z Q^*)}{\partial Z}\right|=|\det Q|^{\beta n+ \beta-2}=\begin{cases} |\det Q|^{2 n},& \beta=2,\\ |\det Q|^{n-1},& \beta=1.\end{cases}
\end{equation}
\end{itemize}
The first identity \eqref{eq_mult_Jacobian} follows from the observation that each column of $Z$ is transformed by linear map $Q$ and there are $m$ such columns. The second and third identities are similar and we refer to   \citet[(1.35)]{forrest} for details.

\smallskip

The first theorem of this section is relevant for the setting of Theorem \ref{Theorem_main_approximation} for VAR($1$).

\begin{theorem}
\label{Theorem_Sum_Jacobi}
Assume $\T\ge 2N$. Let $O$ be a random $\T\times \T$ matrix chosen from the uniform measure on the group of real orthogonal  matrices of determinant $1$ if $\beta=1$ or of complex unitary matrices if $\beta=2$. Define $U=(\1_\T+O)^{-1}$. Then the matrix
\begin{equation}
\label{eq_corner_Var1}
 M=[U]_{NN} ([U^* U]_{NN})^{-1} [U^*]_{NN}
\end{equation}
is distributed as the  $N\times N$ real symmetric (if $\beta=1$) or complex Hermitian (if $\beta=2$)   matrix of density proportional to
\begin{equation}
\label{eq_Var1}
  \det (M)^{\frac{\beta}{2}N+\beta-2} \det(\1_N-M)^{\frac{\beta}{2}(\T-2N+1)-1},\quad 0\le M\le \1_N, \qquad \beta=1,2.
 \end{equation}
 with respect to the Lebesgue measure on symmetric/Hermitian matrices.
\end{theorem}
\begin{remark}\label{Remark_N_to_T_transform} Define a $\T\times\T$--dimensional matrix $\mathcal M$ by putting $M$ in $N\times N$ corner of $\mathcal M$ and filling the rest with zeros. Non-zero eigenvalues of $M$ and $\mathcal M$ are the same. Simultaneously, $\mathcal M$ can be identified with $P_1 P_2 P_1$, where $P_1$ is the orthogonal projector on space $\V$ spanned by the first $N$ coordinate vectors and $P_2$ is the projector on $U \V$.
\end{remark}
\begin{remark}
 If $\T<2N$, then the spaces $\V$ and $U \V$ necessarily intersect, and therefore $M$ has deterministic eigenvalues which equal $1$. This should be taken into account when extending \eqref{eq_Var1} to this case and we will not pursue this direction here.
\end{remark}
\begin{proof}[Proof of Theorem \ref{Theorem_Sum_Jacobi}]
 We parameterize the orthogonal (or unitary if $\beta=2$) group  by the means of the Cayley transform. Namely, we set
 $$
  \R=(\1_\T-O)(\1_\T+O)^{-1}=\frac{\1_\T-O}{\1_\T+O},\quad \text{ so that } \quad O=\frac{\1_\T-\R}{\1_\T+\R}.
 $$
 Since $O^{-1}=O^*$, $\R$ is a skew-symmetric $\T\times \T$ matrix, i.e., $\R^*=-\R$. The uniform (Haar) measure on $O$ leads to the following distribution on $\R$, in which we omitted the irrelevant for us normalization constant:
 \begin{equation}
 \label{eq_Cayley_distribution}
  \det(\1_\T-\R^2)^{-\frac{\beta}{2}\T+\frac{2-\beta}{2}}d\R, \quad \beta=1,2,
 \end{equation}
 see, e.g., \citet[(2.55)]{forrest}.
 Further, we have
 $$
  U=\frac{\1_\T+\R}{2},
 $$
 so that
 \begin{equation}
   M=[\1_\T+\R]_{NN} ( [\1_\T-\R^2]_{NN})^{-1} [\1_\T-\R]_{NN}.
 \end{equation}
 We partition $\T=N+(\T-N)$ and write $\1_\T+\R$ in a block form according to this split:
 $$
   \1_\T+\R=\begin{pmatrix} \1_N+A & B \\ -B^* & \1_{\T-N}+C\end{pmatrix},
 $$
 where $A$ is a $N\times N$ skew-symmetric matrix, $C$ is $(\T-N)\times(\T-N)$ skew-symmetric matrix and $B$ is an arbitrary $N\times (\T-N)$ matrix.

 We make a change of variables by introducing $\tilde B$ so that
 $$
  B=(\1_N-A) \tilde B, \quad \text{ and } \quad  B^*= \tilde B^* (\1_N+A).
 $$
 Using \eqref{eq_mult_Jacobian} we compute the Jacobian of the transformation:
 \begin{equation}
 \label{eq_x11}
   \left|\frac{ \partial B}{\partial \tilde B}\right| = \left|\det(\1_N-A)\right|^{\beta(\T-N)}.
 \end{equation}
 Note that
 \begin{multline}
 \label{eq_x17}
  M=(\1_N+A) (\1_{N}-A^2 + B B^*)^{-1} (\1_N-A)\\=(\1_N+A) \bigl( (\1_N-A) (\1_{N}+ \tilde B \tilde B^*) (\1_N+A) \bigr)^{-1} (\1_N-A)=\bigl(\1_{N}+\tilde B \tilde B^*)^{-1}
 \end{multline}

 Using the formula for the determinant of a block matrix
 $$
  \det \begin{pmatrix} \mathbf{A}& \mathbf{B}\\ \mathbf{C}& \mathbf{D}\end{pmatrix}=\det \mathbf{A}\,\cdot\, \det (\mathbf{D} - \mathbf{C} \mathbf{A}^{-1}\mathbf{B}),
 $$
 we also have
 \begin{multline}
 \label{eq_x13}
  \det(\1_\T+\R)=\det\begin{pmatrix} \1_N+A & (\1_N-A)\tilde B \\ -\tilde B^* (\1_N+A) & \1_{\T-N}+C\end{pmatrix}\\=\det(\1_N+A) \det\begin{pmatrix} \1_{N} & (\1_N-A) \tilde B \\ -\tilde B^*  & \1_{\T-N}+C\end{pmatrix}=\det(\1_N+A) \det(\1_{\T-N}+C+\tilde B^*(\1_N-A) \tilde B)
  \\=\det(\1_N+A)\det(\1_{\T-N}+\tilde B^* \tilde B)\det\biggl(\1_{\T-N}+(\1_{\T-N} + \tilde B^* \tilde B)^{-1/2} (C-\tilde B^* A \tilde B) (\1_{\T-N}+ \tilde B^* \tilde B)^{-1/2}\biggr) ,
 \end{multline}
 Next, we introduce
 $$
  \tilde C=  (\1_{\T-N} + \tilde B^* \tilde B)^{-1/2} (C-\tilde B^* A \tilde B) (\1_{\T-N}+ \tilde B^* \tilde B)^{-1/2}
 $$
 and notice that the map $C\mapsto \tilde C$ preserves the space of all skew-symmetric $(\T-N)\times (\T-N)$ matrices. Using \eqref{eq_skew_Jacobian} the map has the Jacobian
 \begin{equation}
 \label{eq_x12}
 \frac{\partial C}{\partial \tilde C}= \det(\1_{\T-N} + \tilde B^* \tilde B)^{\frac{\beta}{2} (\T-N)+ \frac{\beta-2}{2}},
 \end{equation}

 Combining \eqref{eq_x11}, \eqref{eq_x13}, and \eqref{eq_x12}, we rewrite the measure \eqref{eq_Cayley_distribution} as follows:
\begin{multline}
\label{eq_measure_factor}
 \det(\1_\T-\R^2)^{-\frac{\beta}{2} \T+\frac{2-\beta}{2}}\,dA\, dB\, dC =
 \left| \det(\1_\T+\R)^{-\beta \T+2-\beta}\right|\,dA\, dB\, dC \\=
   \left| \det(\1_N+A)\det(\1_{\T-N}+\tilde B^* \tilde B)\det(\1_{\T-N}+\tilde C)  \right|^{-\beta \T+2-\beta}
   \\ \times |\det(\1_N-A)|^{\beta(\T-N)} \cdot  \det(\1_{\T-N} + \tilde B^* \tilde B)^{\frac{\beta}{2} (\T-N)+ \frac{\beta-2}{2}} \, dA\, d\tilde B\, d\tilde C.
\end{multline}
 The key property of \eqref{eq_measure_factor} is that the measure has factorized and projecting onto the $\tilde B$--component is straightforward. We conclude that $\tilde B$ is distributed according to the measure
 \begin{equation}
 \label{eq_x14}
  \det(\1_{\T-N} + \tilde B^* \tilde B)^{-\frac{\beta}{2}(\T+ N)+ \frac{2-\beta}{2}} \, d\tilde B=  \det(\1_{N} + \tilde B \tilde B^*)^{-\frac{\beta}{2}(\T+ N)+ \frac{2-\beta}{2}} \, d\tilde B,
 \end{equation}
 where we used $\det(\1+UW)=\det(\1+WU)$ in the last equality.
 According to \citet[Exercise 3.2.q6]{forrest}, \eqref{eq_x14} implies that the symmetric (or Hermitian if $\beta=2$) non-negative definite $N\times N$ matrix $F= \tilde B \tilde B^*$ has the law
 \begin{equation}
 \label{eq_x15}
     \det(F)^{\frac{\beta}{2}(\T-2N+1)-1}\det(\1_N + F)^{-\frac{\beta}{2}(\T+ N)+ \frac{2-\beta}{2}} dF.
 \end{equation}
We have by \eqref{eq_x17}
 $$
  M=\frac{1}{\1_N+F},\quad F= \frac{1-M}{M}, \quad dM = -\frac{1}{\1_N+F} dF \frac{1}{\1_N+F}.
 $$
 Using \eqref{eq_sym_Jacobian} we have
 \begin{equation}
 \label{eq_x16}
  \left| \frac{\partial M}{\partial F}\right|= \left|\det\left(\frac{1}{\1_N+F}\right)\right|^{\beta N -\beta+2 } =\left| \det M\right|^{\beta N-\beta+2}.
 \end{equation}
 Formulas
 \eqref{eq_x15} and \eqref{eq_x16} imply that the matrix $M$ has the distribution
 \begin{multline*}
   \det(\1_N-M)^{\frac{\beta}{2}(\T-2N+1)-1} \det(M)^{-\frac{\beta}{2}(\T-2N+1)+1}\det(M)^{\frac{\beta}{2} (\T+N)- \frac{2-\beta}{2}}  \det(M)^{-\beta N+\beta-2} dM\\
   =\det(\1_N-M)^{\frac{\beta}{2}(\T-2N+1)-1} \det(M)^{\frac{\beta}{2} N+\beta-2} dM,
    \quad 0\le M \le \1_N. \qedhere
 \end{multline*}
\end{proof}

Our next theorem gives another realisation of the Jacobi ensemble, which is relevant to testing $\Pi=-\1_N$ in VAR($1$) setting.

\begin{theorem} \label{Theorem_Var_0} Assume $\T\ge 2N$. Let $O$ be a random $\T\times \T$ real matrix chosen from the uniform measure on the group of orthogonal matrices with determinant $1$ if $\beta=1$ or of complex unitary matrices if $\beta=2$.  Let $A$ be $N\times N$ top-left corner of $O$. Then the matrix
$$
 M=(\1_N+A) (2\1_N+A+A^*)^{-1} (\1_N+A^*)
$$
is distributed as $N\times N$ real symmetric (if $\beta=1$) or complex Hermitian (if $\beta=2$) matrix of density proportional to
\begin{equation}
\label{eq_Jacobi_ensemble}
  \det (M)^{\frac{\beta}{2}(\T-N)+\beta-2} \det(\1_N-M)^{\frac{\beta}{2}(\T-2N+1)-1},\quad 0\le M\le \1_N, \qquad \beta=1,2.
 \end{equation}
 with respect to the Lebesgue measure on symmetric/Hermitian matrices.
\end{theorem}
\begin{proof}
 The computation of the law of $A$ is well-known and \citet[(3.113)]{forrest} provides a formula for the density with respect to the Lebesgue measure on all real/complex $N\times N$ matrices. It is
 \begin{equation}
 \label{eq_corner_measure}
  \det(\1_N-A^* A)^{\frac{\beta}{2}(\T-2N+1)-1}\, dA, \quad 0\le  A^* A\le \1_N,
 \end{equation}
 where we omit here and below the normalization constant, which makes the total mass of measure equal to $1$. The matrix $M$ is a function of $A$ and in the rest of the proof we transform the measure \eqref{eq_corner_measure} to make the distribution of this function explicit.

 Our first step is to rewrite \eqref{eq_corner_measure} in terms of $M$. We claim that
 \begin{equation}
 \label{eq_density_factor}
  \det(\1_N-A^* A)=\det(\1_N-M) \det (2\1_N+ A + A^*)
 \end{equation}
 Indeed, we first define
 $$
  M_1=(\1_N+A^*) (\1_N+A) (2\1_N+A+A^*)^{-1}
 $$
and notice that
$$
 \1_N-A^* A= 2\1_N+A+A^*- M_1 (2\1_N+A+A^*) = (\1_N-M_1)(2\1_N+A+A^*).
$$
Hence,
\begin{equation}
\label{eq_x1}
 \det(\1_N-A^* A)=\det(\1_N-M_1)\det(2\1_N+A+A^*)
\end{equation}
and \eqref{eq_density_factor} is obtained from \eqref{eq_x1} by using the  identity $\det(\1-CD)=\det(\1-DC)$.

We conclude that the density \eqref{eq_corner_measure} has the form proportional to
\begin{equation}
\label{eq_corner_measure_form}
\det(\1_N-M)^{\frac{\beta}{2}(\T-2N+1)-1} \det (2\1_N+ A + A^*)^{\frac{\beta}{2}(\T-2N+1)-1}\, dA.
\end{equation}
Note that the first factor has the desired form from the statement of the theorem.

We now split the Euclidean space of all $N\times N$ real (or complex) matrices into the symmetric (or Hermitian) and skew-symmetric (or skew-Hermitian) parts. For that we define
$$
 X=\1_N+\frac{A+A^*}{2},\quad Y= \frac{A-A^*}{2}.
$$

Since $A$ belongs to a unit matrix ball (because $A A^*+B B^*=\1_N$, where $B$ is the top-right $N\times (\T-N)$ corner of $O$), $X$ is a positive-definite symmetric (or Hermitian) matrix, i.e., $X\ge 0$. The relation
$$
 (\1_N+A) (2\1_N+A+A^*)^{-1} (\1_N+A^*)=M
$$
is now rewritten as
\begin{equation}
\label{eq_x2}
(X+Y) (2X)^{-1} (X-Y)=M \quad \text{ or } \quad
 X=2M + Y X^{-1} Y.
\end{equation}
Note that $X^{-1}$ is positive definite, while $Y$ is skew-Hermitian. This implies that $Y X^{-1} Y\le 0$, i.e., $Y X^{-1} Y$ is a negative semi-definite Hermitian matrix. Hence, \eqref{eq_x2} implies
$$
 X\le 2 M,
$$
which means that $2M-X$ is a non-negative definite matrix.

Using \eqref{eq_x2} we make a change of coordinates in the space of matrices
$
 (X,Y)\to (M,Y).
$
The Jacobian of this change of coordinates is
$$
 \det\begin{pmatrix} \frac{\partial M}{\partial X}& \frac{\partial Y}{\partial X}\\ \frac{\partial M}{\partial Y} & \frac{\partial Y}{\partial Y} \end{pmatrix}= \det\begin{pmatrix}  \frac{\partial M}{\partial X}  & 0 \\  \frac{\partial M}{\partial Y}  & \1\end{pmatrix}= \det\left(\frac{\partial M}{\partial X}\right)
$$
Therefore, we need to compute the Jacobian of the map
$$
 X\mapsto M=\tfrac{1}{2}\left(X- Y X^{-1} Y\right),
$$
which maps the Eucledian space of symmetric (Hermitian if $\beta=2$) matrices to itself. Differentiating, and using $d X^{-1}=X^{-1} dX X^{-1}$, we see the matrix identity
\begin{equation}
\label{eq_dM}
 2 dM= dX - Y X^{-1} dX X^{-1} Y= dX + (Y X^{-1}) dX (Y X^{-1})^*,
 \end{equation}
where we used $X=X^*$, $Y=-Y^*$ in the last identity. Hence, the Jacobian of the map $X\mapsto M$ is a function of $Y X^{-1}$ and we denote this function $J( Y X^{-1})$. Therefore, \eqref{eq_corner_measure_form} becomes
\begin{equation}
\label{eq_x5}
 \sim \det(\1_N-M)^{\frac{\beta}{2}(\T-2N+1)-1} \det (X)^{\frac{\beta}{2}(\T-2N+1)-1}\, J^{-1}(Y X^{-1}) \, dM\, dY.
\end{equation}
We now make another change of variables by setting
$$
 \tilde Y=M^{-1/2} Y M^{-1/2}, \qquad Y=M^{1/2} \tilde Y M^{1/2}.
$$
The Jacobian of the map $(M,Y)\to (M,\tilde Y)$ is the same as the Jacobian of the map $Y \to \tilde Y$, which is computed by \eqref{eq_skew_Jacobian} as
\begin{equation}
\left|\frac{\partial \tilde Y}{\partial Y}\right|=|\det M|^{-\frac{\beta}{2} N+ \frac{2-\beta}{2}}=\begin{cases} (\det M)^{-N},& \beta=2,\\ (\det M)^{-(N-1)/2},& \beta=1.\end{cases}
\end{equation}
This converts \eqref{eq_x5} into
\begin{equation}
\label{eq_x6}
  \sim \det(\1_N-M)^{\frac{\beta}{2}(\T-2N+1)-1}  (\det M)^{\frac{\beta}{2} N+ \frac{\beta-2}{2}} \det (X)^{\frac{\beta}{2}(\T-2N+1)-1}\, J^{-1}(Y X^{-1}) \, dM\, d\tilde Y.
\end{equation}
It remains to figure out the factors involving $\det X$ and $ Y X^{-1}$ in the last formula. For that let us introduce the notation
$$
\tilde X= M^{-1/2} X M^{-1/2},
$$
and notice that \eqref{eq_x2} implies
\begin{equation}
\label{eq_x7}
 \tilde X=2\1_N+ \tilde Y \tilde X^{-1} \tilde Y, \qquad 0\le \tilde X\le 2\1_N.
\end{equation}
Solution to \eqref{eq_x7} gives us a function $\tilde X=\tilde X(\tilde Y)$, such that $X=M^{1/2} \tilde X M^{1/2}$.
Hence, we can transform
\begin{equation}
\label{eq_x8}
 \det (X)^{\frac{\beta}{2}(\T-2N+1)-1}= \det (M)^{\frac{\beta}{2}(\T-2N+1)-1} \det (\tilde X)^{\frac{\beta}{2}(\T-2N+1)-1}.
\end{equation}
For $J(Y X^{-1})$ we notice two properties. First,
$$
 Y X^{-1}= M^{1/2} \tilde Y \tilde X^{-1} M^{-1/2}.
$$
In addition, we have the following statement, which we prove later:

{\bf Claim.} For the above Jacobian $J(\cdot)$, we have
\begin{equation}
\label{eq_x9}
 J(A)=J(B^{-1} A B),\qquad\text{whenever } B=B^*.
\end{equation}
Thus, $J(YX^{-1})=J(\tilde Y \tilde X^{-1})=J(f(\tilde Y))$ for a certain function $f$, and
using \eqref{eq_x8} and \eqref{eq_x9} we rewrite \eqref{eq_x6} as
\begin{equation}
  \det(\1_N-M)^{\frac{\beta}{2}(\T-2N+1)-1}  (\det M)^{\frac{\beta}{2} N+ \frac{\beta-2}{2}} \det (M)^{\frac{\beta}{2}(\T-2N+1)-1}\, g(\tilde Y)\, dM\, d\tilde Y
\end{equation}
for a certain function $g(\tilde Y)$. Since the parts involving $\tilde Y$ and $M$ are now decoupled, we can integrate out $\tilde Y$ arriving at the desired expression for the distribution of $M$:
\begin{equation}
  \det(\1_N-M)^{\frac{\beta}{2}(\T-2N+1)-1}   \det (M)^{\frac{\beta}{2}(\T-N)+ \beta-2}\, dM.
\end{equation}
It remains to prove the claim \eqref{eq_x9}. By definition \eqref{eq_dM}, $J(B^{-1} A B)$ is the Jacobian of the linear map
$$
 dX \mapsto \tfrac12dX+ \tfrac12(B^{-1} A B) dX (B A^* B^{-1})
$$
We split this map into a composition of three:
$$
 dX \mapsto B\, dX\, B \mapsto \tfrac12(B\, dX\, B) + \tfrac12A (B\, dX\, B) A^* \mapsto B^{-1} \left[\tfrac12(B\, dX\, B) + \tfrac12A (B\, dX\, B) A^*\right] B^{-1}.
$$
$J(B^{-1} A B)$ is the product of Jacobians of these three maps.
The first and the last maps are inverse to each other (above we used the explicit Jacobian of these maps, but this is not even needed) and their Jacobians cancel. The Jacobian of the middle map is the desired $J(A)$.
\end{proof}

\subsection{Proof of Theorem \ref{Theorem_main_approximation}}
\label{Section_proof_of_main_approx}
The proof of Theorem \ref{Theorem_main_approximation} is split into two steps. First, Proposition \ref{Prop_gaussian_rotation} uses rotation invariance of the Gaussian law to rewrite the  matrix
 $S_{10} S_{00}^{-1} S_{01} S_{11}^{-1}$ as a product of corners of certain matrices, reminiscent of \eqref{eq_corner_Var1}. Next, we show in Proposition \ref{Proposition_H1_error} that replacement of a certain deterministic matrix (in that product) by its random perturbation leads precisely to \eqref{eq_corner_Var1} and simultaneously has controlled effect on the change of eigenvalues.

\medskip

We need to introduce some deterministic $T\times T$ matrices. The summation matrix  $\Phi$ has $1$'s below the diagonal and $0$'s on the diagonal and everywhere above the diagonal:
$$
 \Phi=\begin{pmatrix} 0&0&0&\dots&0\\ 1& 0 &0&\dots &0\\ 1& 1& 0 &\dots & 0 \\ && \ddots \\ 1&1 &\dots & 1 &0\end{pmatrix}.
$$

\begin{definition} $V$ is the $(T-1)$--dimensional hyperplane orthogonal to $(1,1,\dots,1)$.
\end{definition}
   We let $\mathcal P$ be the orthogonal projector on $V$ (see Eq.~\eqref{eq_demean}).  Finally, we set
$$
 \tilde \Phi =\mathcal P \Phi \mathcal P.
$$
We also need the cyclic version of the lead operator $F$ mapping $(x_1,x_2,\dots,x_{T})$ to $(x_{2},x_3, x_4,\dots,x_{T},x_1)$. $V$ is an invariant subspace for $F$ and we denote through $F_V$ the restriction of $F$ onto $V$.

\begin{lemma} \label{Lemma_shift_inversion}
 The operator $\tilde \Phi$ preserves the space $V$. Its restriction onto $V$ coincides with $-(\1_V-F_V)^{-1}$,
where $\1_V$ is the identical operator acting in $V$.
\end{lemma}

\begin{proof}[Proof of Lemma \ref{Lemma_shift_inversion}]
 Let $e_i$, $i=1,\dots,T$ denote the $i$th coordinate vector. Then $(e_i-e_{i+1})_{i=1}^{T-1}$ gives a linear basis of $V$. Since $\Phi e_i=e_{i+1}+e_{i+2}+\dots+e_{T}$, we have
 $$
  \tilde \Phi (e_i-e_{i+1})= \mathcal P \Phi (e_i-e_{i+1})= \mathcal P e_{i+1}=e_{i+1}-\frac{1}{T} \sum_{j=1}^{T} e_j.
 $$
 Applying $\1-F$ to the last vector and noticing that $F e_{i+1}=e_i$, we conclude that
 $$
  (\1-F) \tilde \Phi (e_i-e_{i+1})=e_{i+1}-e_i. \qedhere
 $$
\end{proof}

Let us introduce one more notation. Choose some orthonormal basis $\tilde e_1,\tilde e_2,\dots,\tilde e_{T-1}$ of $V$, e.g., this can be an orthogonalization of $e_1-e_2$, $e_2-e_3$, \dots, $e_{T-1}-e_{T}$ (but the exact choice is irrelevant for the following theorems).
\begin{definition}
For an operator $A$ acting in $V$ we set $[ A ]_{NN}^V$ to be the $N\times N$ corner of the operator $A$, taken in the basis $\{\tilde e_i\}$.
\end{definition}

Next, we take a uniformly-random orthogonal (or unitary if $\beta=2$) operator $\tilde O$ acting in $(T-1)$--dimensional space $V$ and define an operator $\tilde U$ acting in $V$:
$$
 \tilde U= (\1_V- \tilde O F_V \tilde O^*)^{-1}.
$$
 Let us introduce a symmetric (or Hermitian if $\beta=2$) $N\times N$ matrix $\tilde M$,
\begin{equation}
\label{eq_tilde_M}
 \tilde M= [\tilde U]_{NN}^V ([\tilde U^* \tilde U]_{NN}^V)^{-1} [\tilde U^*]_{NN}^V,
\end{equation}
\begin{proposition}
\label{Prop_gaussian_rotation}
Choose an arbitrary positive definite covariance matrix $\Lambda$. Let $\eps$ be $N\times T$ matrix of random variables (real if $\beta=1$ and complex if $\beta=2$), such that $T$ columns of $\eps$ are i.i.d., and each of them is an $N$-dimensional mean zero Gaussian vector with covariance $\Lambda$. Fix arbitrary $N$--dimensional vectors $X_0$ and $\mu$. Define the $N\times T$ matrix $X=(X_1,X_2,\dots,X_T)$ as in Eq.~\eqref{var_1} via recurrence
$$
 X_t=X_{t-1}+\mu+\eps_t, \quad t=1,\dots,T.
$$
Further set $\Delta X_t=X_t-X_{t-1}$ and $\tilde X_t= X_{t-1}-\frac{t-1}{T} (X_T-X_0)$, $t=1,\dots,T$.
 Define $N\times N$ matrices:
 $$
   S_{00}=   \Delta X \mathcal P \Delta X^*,\quad S_{01}=\Delta X \mathcal P \tilde X^* ,\quad S_{10}=S_{01}^*=\tilde X \mathcal P \Delta X^*, \quad S_{11} = \tilde X\mathcal P \tilde X^*.
 $$
 Then the eigenvalues of the matrix  $S_{10} S_{00}^{-1} S_{01} S_{11}^{-1}$ have the same distribution as those of $\tilde M$ in \eqref{eq_tilde_M}.
\end{proposition}
\begin{proof}
We start by expressing the matrices $S_{ij}$ via $\eps$. Clearly, $\Delta X_t=\mu+\eps_t$. Further,
$$
  \left(\Delta X \mathcal P\right)_t= \mu+\eps_t-\frac{1}{T}\sum_{i=1}^T (\mu+\eps_i), \qquad \Delta X \mathcal P=\eps \mathcal P .
$$
Hence, $S_{00}=\eps \mathcal P \eps^*$. Next,
$$
 \tilde X_t= X_0 + (t-1)\mu+ \sum_{i=1}^{t-1} \eps_i - \frac{t-1}{T} \left(T\mu +\sum_{i=1}^{T} \eps_i\right)= X_0 +  \sum_{i=1}^{t-1} \eps_i - \frac{t-1}{T} \sum_{i=1}^{T} \eps_i
$$
We claim that $\tilde X \mathcal  P= \eps \tilde \Phi^*$. Indeed, $ \tilde X \mathcal P$ coincides with $\dbtilde X \mathcal P $, where
$$
\dbtilde X_t=\tilde X_t- X_0= \sum_{i=1}^{t-1} \eps_i - \frac{t-1}{T} \sum_{i=1}^{T} \eps_i=\sum_{i=1}^{t-1} \left( \eps_i - \frac{1}{T} \sum_{j=1}^{T} \eps_j\right).
$$
Since $\Phi$ is the summation operator, we have $\dbtilde X = ( \Phi ( \eps \mathcal P)^* )^* = \eps \mathcal P \Phi^*$ and the claim is proven as $\tilde \Phi^*=\mathcal P\Phi^*\mathcal P$.

We conclude that
 \begin{equation}
 \label{eq_x27}
   S_{00}=  \eps \mathcal P \eps^*,\quad S_{01}=\eps \tilde\Phi \eps^* ,\quad S_{10}=S_{01}^*=\eps \tilde \Phi^* \eps^*, \quad S_{11} = \eps \tilde \Phi^* \tilde \Phi \eps^*.
 \end{equation}
Next, note that we can (and will) assume without loss of generality that the covariance matrix $\Lambda$ is identical, i.e., all matrix elements of $\eps$ are i.i.d.~random variables $\mathcal{N}(0,1)$ if $\beta=1$ and $\mathcal{N}(0,1)+\ii \mathcal{N}(0,1)$ if $\beta=2$. Indeed, using \eqref{eq_x27}, if we take any non-degenerate $N\times N$ matrix $A$ and replace $\eps$ by $A\eps$, then the  matrix $S_{10} S_{00}^{-1} S_{01} S_{11}^{-1} $ is conjugated by $A$, which keeps eigenvalues unchanged. By choosing appropriate $A$, we can guarantee that the covariance of columns of $A\eps$ is identical, and then rename $A\eps$ as $\eps$.

In the remaining proof we explicitly couple $\eps$ with $\tilde O$ (arising in the definition of $\tilde M$) by constructing $\tilde O$ using randomness coming from $\eps$.

Since for any two rectangular matrices $A$ and $B$ of the same sizes, $A B^*$ and $B^* A$ have the same non-zero eigenvalues,  the eigenvalues of $S_{10} S_{00}^{-1} S_{01} S_{11}^{-1} $ can be identified (using also $\mathcal P \tilde \Phi=\tilde \Phi$) with those of
\begin{equation}
\label{eq_product_proj}
  \bigl( \mathcal P \eps^* \, S_{00}^{-1}\, \eps \mathcal P \bigr) \bigl( \tilde \Phi \eps^* \,  S_{11}^{-1} \, \eps \tilde \Phi^*  \bigr)= P_1 P_2,
\end{equation}
where $P_1$ is the orthogonal projector onto the space spanned by $N$ columns of $\mathcal P \eps^* $ and $P_2$ is the orthogonal projector onto the $N$ columns of $\tilde \Phi \eps^* $. Since, the eigenvalues of $P_1 P_2$ are the same as those of $P_1 P_2 P_1$, and the latter operator acts as $0$ on the orthogonal complement of $V$, we can restrict all the operators to $V$. Due to invariance of the Gaussian law with respect to rotations by orthogonal (or unitary if $\beta=2$) matrices, the columns of $\mathcal P \eps^*$ span a rotationally-invariant $N$--dimensional subspace in $V$. The columns of $\tilde \Phi \eps^*=\tilde \Phi \mathcal P\eps^* $  then span the  $\tilde \Phi$--image of this subspace.

The eigenvalues of $P_1 P_2 P_1$ are preserved when we conjugate each of the projectors with an orthogonal transformation $O$, i.e.~replace $P_1 P_2 P_1$ with $ O P_1 O^* O P_2 O^* O P_1 O^*=OP_1 P_2 P_1 O^*$. In this transformation the spaces to where $P_1$ and $P_2$ are projecting are replaced by their $O$--images. We take $O$ to be a random operator satisfying two conditions:
\begin{itemize}
\item $O$ maps the span of the columns of $\mathcal P \eps^*$ to the subspace spanned by the first $N$ basis vectors, $\tilde e_1,\dots,\tilde e_N$ in $V$;
\item $O$ is distributed as uniformly random orthogonal operator acting in $V$, i.e., $O\stackrel{d}{=}\tilde O$.
\end{itemize}
The existence of such $O$ follows from the rotational invariance of the Gaussian law.\footnote{One way to see it is by noticing that the span of the columns of $\mathcal P \eps^*$ and the span of $N$ columns of $O^{-1}$ (for uniformly random $O$) have the same distribution given by the uniform measure on all $N$--dimensional subspace of $T$--dimensional space.} Hence, $O P_1 O^*$ is a projector onto $\langle\tilde e_1,\dots,\tilde e_N\rangle$, and we conclude that the eigenvalues of $P_1 P_2 P_1$ coincide with those of
$$
 [ O P_2 O^*]_{NN}^V.
$$
Since $P_2$ is the projector on the span of columns of $\tilde \Phi \eps^*= O^* (O \tilde \Phi O^*) O \eps^*$, $O P_2 O^*$ is the projector on the $O$--image of this span, i.e., $O P_2 O^*$ is the projector onto the columns of $ (O\tilde \Phi O^*) O \eps^*$. Since columns of $O\eps^*$ span $\langle\tilde e_1,\dots,\tilde e_N\rangle$ and $O \tilde \Phi O^*=O(-(\1_V-F_V)^{-1})O^*=-\tilde U$ on $V$ by Lemma \ref{Lemma_shift_inversion}, we reach the formula for $\tilde M$ as the $N\times N$ corner of the projector on the first $N$ columns of $\tilde U$.
\end{proof}

Set $\T = T-1$ and let $O$ be uniformly random real orthogonal of determinant $1$ (if $\beta=1$) or complex unitary (if $\beta=2$) $\T\times \T$ matrix. Define $U=(\1_\T+ O)^{-1}$. Then we set, as in \eqref{eq_corner_Var1},
\begin{equation}
\label{eq_M_def}
 M=[U]_{NN} ([U^* U]_{NN})^{-1} [U^*]_{NN}
\end{equation}

For a matrix $A$ with real spectrum, we set $\lambda_i(A)$ to be the $i$th largest eigenvalue of $A$.
\begin{proposition}
\label{Proposition_H1_error}
Assume $\T = T-1 \ge 2N$.
 One can couple $M$ from \eqref{eq_M_def} with $\tilde M$ from \eqref{eq_tilde_M} in such a way that  for each $\epsilon>0$ we have
 $$
   \mathrm{Prob}\left( \max_{1\le i \le N} |\lambda_i(M)-\lambda_i(\tilde M)|< \frac{1}{N^{1-\epsilon}}\right)\to 1,
 $$
 as $T,N\to\infty$ in such a way that the ratio $T/N$ remains bounded.
\end{proposition}
\begin{remark}
 We believe that the condition $\T \ge 2N$ can be weakened and replaced by $T/N$ bounded away from $0$, but we stick to it, since that's the situation when Theorem \ref{Theorem_Sum_Jacobi} applies.
\end{remark}
Let us compare the definitions of $M$ and $\tilde M$. There is only one difference between them: the eigenvalues of $-\tilde O F_V \tilde O^*$ in the definition of $\tilde M$ are deterministic, while the eigenvalues of $O$ are random. Note that both matrices $-\tilde O F_V \tilde O^*$ and $O$  have uniformly-random eigenvectors. Then the idea of the proof is to rely on the phenomenon of \emph{rigidity} for random matrix eigenvalues, which says that as $T\to\infty$, the eigenvalues of $O$ are very close to their deterministic expected positions, which, in turn, essentially coincide with eigenvalues of $-F_V$. One technical difficulty is that we need to deal with $(\1+O)^{-1}$, whose eigenvalues can be large (the absolute value of the largest eigenvalue grows linearly in $T$), however, as we will see, one can study $(\1+O)$ instead of $(\1+O)^{-1}$, and the problem of unbounded operators disappears. We proceed to the detailed proof.

\begin{proof}[Proof of Proposition \ref{Proposition_H1_error}]

 We start by explicitly constructing the desired coupling. The eigenvalues of $F_V$ are all roots of unity of order $T$ different from $1$. If $\beta=2$, then we can diagonalize $F_V$ to turn it into $\T\times \T = (T-1)\times (T-1)$ diagonal matrix with the roots of unity on the diagonal. If $\beta=1$, the matrix $F_V$ should be  block-diagonalized (with blocks of size $2$ and one additional block of size $1$ corresponding to eigenvalue $-1$ if $\T$ is even): the pair of complex conjugate roots of unity $\omega $ and $\bar \omega$ gives rise to the $2\times 2$ matrix of rotation by the angle $|\arg(\omega)|$. Let us denote by $D$ the resulting (block) diagonal matrix multiplied by $-1$. In order to avoid ambiguity about the order of eigenvalues, we assume that the blocks correspond to the increasing order of $|\arg(-\omega)|$, i.e.~the top-left $2\times 2$ corner of $D$ corresponds to the pair of the closest to $1$ eigenvalues of $D$.

 The eigenvalues of $O$ also lie on the unit circle and if $\beta=1$ they come in complex-conjugate pairs. Hence, $O$ can be similarly block-diagonalized (we do not need to multiply by $-1$ this time) and we denote through $D^{\text{rand}}$ the result. The distinction with $F_V$ is that the eigenvalues are \emph{random} and so is $D^{\text{rand}}$. The law of the eigenvalues of $O$ is explicitly known in the random-matrix literature. Both for $\beta=1$ and $\beta=2$ they form a determinantal point process on the unit circle with explicit kernel. The \emph{repulsion} between the eigenvalues leads to them being very close to evenly spaced as $T\to\infty$. We summarize this property in the following statement (which is a manifestation of much more general rigidity of eigenvalues, see, e.g., \citet{ErdosYau}), whose proof can be found in \citet[Lemma 10, $m=1$, $u=T^{\delta}$ case, and Section 5]{MeckesMeckes}.

 {\bf Claim.} There exist constants $c_1(\beta),c_2(\beta)>0$, such that for $\beta=1,2$, every $\delta>0$, there exists $\T_0(\delta)$ and for every  $\T>\T_0(\delta)$ we have \footnote{All the constants can be made explicit, following \citet{MeckesMeckes}.}
 \begin{equation}
 \label{eq_MM_bound}
  \mathrm{Prob}\left(\max_{1\le i,j<\T} \bigl| D-D^{\text{rand}}\bigr|_{ij} > \frac{1}{\T^{1-\delta}}\right) < c_1(\beta) \cdot \T \cdot \exp\left( - c_2(\beta) \frac{\T^{2\delta}}{\log \T} \right).
 \end{equation}

\medskip

 We remark that since $D$ and $D^{\text{rand}}$ are block-diagonal, the bound on the maximum matrix element of their difference is equivalent to a similar bound for any other norm, e.g., for the maximum absolute value among the eigenvalues of $D-D^{\text{rand}}$.

 We now choose another $\T\times \T$ uniformly-random orthogonal (or unitary if $\beta=2$) matrix $O_2$ (independent from the rest), replace $-\tilde O F_V \tilde O^*$ with $O_2 D O_2^*$ and replace $O$ with $O_2 D^{\text{rand}} O_2^*$. The invariance of the uniform measure on the orthogonal group $SO(N)$ (or on the unitary group $U(N)$ if $\beta=2$) with respect to right/left multiplications, implies the distributional identities:
 $$
  -\tilde O F_V \tilde O^* \stackrel{d}{=} O_2 D O_2^*,\qquad O\stackrel{d}{=} O_2 D^{\text{rand}} O_2^*.
 $$
 The right-hand sides of the identities provide the desired coupling and \eqref{eq_MM_bound} implies that these two random matrices are close to each other as $\T\to\infty$. Both matrices $M$ and $\tilde M$ are obtained from the above two matrices by the following mechanism:

 {\bf Map $M(Z)$:} Given an orthogonal (or unitary if $\beta=2$) matrix $Z$ in $\T$--dimensional space, we define $U(Z)=(\1_\T+Z)^{-1}$ and an $N\times N$ symmetric (or hermitian) matrix $M(Z)$ through
 \begin{equation}
 \label{eq_MZ_def}
  M(Z)=[U(Z)]_{NN} ([U^*(Z)U(Z)]_{NN})^{-1} [U^*(Z)]_{NN}.
 \end{equation}

 \medskip

 Clearly, under our coupling
 $$
   M=M(O_2 D^{\text{rand}} O_2^*),\qquad \tilde M=M(O_2 D O_2^*),
 $$
 and it remains to prove a form of the uniform continuity of the map $Z\mapsto M(Z)$ as $N,\T\to\infty$.

 We note that
 $$
 2U(Z)-\1_\T= \frac{\1_\T-Z}{\1_\T+Z}
$$
is a skew-symmetric matrix and we use it to write $2 U(Z)$ in the block form according to the splitting $\T=N+(\T-N)$.
$$
  2 U(Z)= \begin{pmatrix} \1_N+ A(Z)& B(Z)\\ -B^*(Z) & \1_{\T-N}+ C(Z)\end{pmatrix}, \qquad A^*(Z)=-A(Z),\quad C^*(Z)=-C(Z).
$$
Then, as in the proof of Theorem \ref{Theorem_Sum_Jacobi}, we have
\begin{equation}
\label{eq_x18}
 M(Z)=\bigl(\1_{N}+\tilde B (Z) \tilde B^*(Z)\bigl)^{-1},\qquad \tilde B(Z)=\bigl(\1_N-A(Z)\bigr)^{-1} B(Z).
\end{equation}
and our statement boils down to the uniform continuity of the map $Z\mapsto \tilde B(Z)$. We remark that in our limit regime the matrix $B(Z) B^*(Z)$ might have large eigenvalues (in fact, of order $\T$) and therefore, $B(Z)$ is potentially an exploding factor in the definition of $\tilde B(Z)$. However, the exploding parts would precisely cancel out with similarly growing parts in $(\1_N-A(Z))$. In order to see that, we use the formula for the inverse of the block matrix, which reads
$$
 \begin{pmatrix} \mathbf{A}& \mathbf{B}\\ \mathbf{C}& \mathbf{D}\end{pmatrix}^{-1} = \begin{pmatrix} \mathbf{A}^{-1}+ \mathbf{A}^{-1} \mathbf{B} \mathbf{Q} \mathbf{C} \mathbf{A}^{-1} & - \mathbf{A}^{-1}\mathbf{B} \mathbf{Q} \\ - \mathbf{Q} \mathbf{C}\mathbf{A}^{-1}& \mathbf{Q} \end{pmatrix} ,\qquad \mathbf{Q} =  (\mathbf{D}- \mathbf{C} \mathbf{A}^{-1}\mathbf{B} )^{-1}
$$
The advantage of this formula is that when used for $2U(Z)$, the ratio $-\mathbf{C} \mathbf{A}^{-1}=B^*(z)(1+A(z))^{-1}=\tilde B^*(Z)$ gets expressed as $\mathbf{Q}^{-1} (-\mathbf{Q} \mathbf{C} \mathbf{A}^{-1})$, which is the ratio of two blocks in $U(Z)^{-1}=1+Z$.\footnote{The same observation can be used to link Theorems \ref{Theorem_Sum_Jacobi} and \ref{Theorem_Var_0} to each other.} Thus, if we write $1+Z$ itself in the block form according to $\T=N+(\T-N)$ splitting:
$$
 \1_\T+Z=\begin{pmatrix} \1_N + Z_{11} & Z_{12}\\ Z_{21} & \1_{\T-N}+Z_{22}\end{pmatrix},
$$
then
$$
 \tilde B^*(Z)=\mathbf{Q}^{-1} (-\mathbf{Q} \mathbf{C} \mathbf{A}^{-1})= (\1_{\T-N}+ Z_{22})^{-1} Z_{21},
 $$
 $$
 M(Z)=\biggl(\1_{N}+Z_{21}^*  \bigl( (\1_{\T-N}+ Z_{22})( \1_{\T-N}+ Z_{22}^*)\bigr)^{-1} Z_{21} \biggr)^{-1}
$$
Hence, the matrix $M(Z)$ has the same (other from $1$) eigenvalues as
$$
 M'(Z)=\biggl(\1_{\T-N}+ \bigl( (\1_{\T-N}+ Z_{22})( \1_{\T-N}+ Z_{22}^*)\bigr)^{-1} Z_{21} Z_{21}^*  \biggr)^{-1}
$$
The latter can be simplified, since $Z$ is orthogonal, implying $Z_{21} Z_{21}^*+Z_{22} Z_{22}^*=\1_{\T-N}$:
\begin{multline}
\label{eq_x19}
 M'(Z)=\biggl(\1_{\T-N}+ \bigl( (\1_{\T-N}+ Z_{22}+ Z_{22}^*+Z_{22}Z_{22}^* )\bigr)^{-1} (\1_{\T-N}-Z_{22} Z_{22}^*)  \biggr)^{-1}
 \\= \biggl(\bigl( (\1_{\T-N}+ Z_{22}+ Z_{22}^*+Z_{22}Z_{22}^* )\bigr)^{-1} (2\cdot \1_{\T-N}+ Z_{22}+ Z_{22}^*)  \biggr)^{-1}
 \\= (2\cdot \1_{\T-N}+ Z_{22}+ Z_{22}^*)^{-1} \bigl( \1_{\T-N}+ Z_{22}+ Z_{22}^*+Z_{22}Z_{22}^* \bigr)
\end{multline}
At this point we rely on the lemma, which is proven below:
\begin{lemma}
\label{Lemma_inv_bounded} Fix any small $\upsilon>0$.
 If $Z$ is a uniformly random element of the group $SO(\T)$ of orthogonal determinant $1$ matrices (or $U(\T)$ if $\beta=2$), then there exists $r>0$ such that with probability tending to $1$ (uniformly) as $\T,N\to\infty$ in such a way that $1+\upsilon<\T/N<\upsilon^{-1}$, the smallest eigenvalue of\, $2\, \1_{\T-N}+ Z_{22}+ Z_{22}^*$ is larger than $r$.
\end{lemma}

Now we are ready to compare $M'(Z)$ with $M'(\tilde Z)$, where $Z=O_2 D^{\text{rand}} O_2^*$, $\tilde Z=O_2 D O_2^*$. Fix $\delta>0$. By \eqref{eq_MM_bound}, with probability tending to $1$ as $N,\T\to\infty$, $\|Z-\tilde Z\|<\T^{\delta-1}$, where $\|\cdot\|$ is the spectral norm of the matrix (i.e.~its largest singular value). Hence, since cutting corners of the matrix only decreases the spectral norm, with probability tending to $1$,
\begin{equation}
\left\|\bigl( \1_{\T-N}+ Z_{22}+ Z_{22}^*+Z_{22}Z_{22}^* \bigr)- \bigl( \1_{\T-N}+ \tilde Z_{22}+ \tilde Z_{22}^*+\tilde Z_{22}\tilde Z_{22}^* \bigr)\right\|<4 \T^{\delta-1}.
\end{equation}
Simultaneously, by Lemma \ref{Lemma_inv_bounded}, the spectral norm of  $(2\, \1_{\T-N}+ Z_{22}+ Z_{22}^*)^{-1}$ is bounded from above by $r^{-1}$ (and, hence, a similar bound holds for $Z$ replaced with $\tilde Z$). We conclude that the part of the increment $M'(Z)-M'(\tilde Z)$ due to the change of the numerator in the right-hand side of \eqref{eq_x19} is bounded from above (in spectral norm) by
\begin{equation}
\label{eq_x21}
 \text{const} \cdot  \T^{\delta-1}.
\end{equation}
For the denominator in \eqref{eq_x19}, we use the matrix identity
\begin{equation}
\label{eq_x20}
 (G+\Delta)^{-1}-G^{-1}=- (G+\Delta)^{-1}  \Delta\, G^{-1}
\end{equation}
with $G=2\, \1_{\T-N}+ Z_{22}+ Z_{22}^*$, $\Delta = \tilde Z_{22}+\tilde Z_{22}^*-Z_{22}- Z_{22}^{*}$. By \eqref{eq_MM_bound}, $\|\Delta\|<2 \T^{\delta-1}$ with probability tending to $1$. Hence, using Lemma \ref{Lemma_inv_bounded} we upper bound the spectral norm of \eqref{eq_x20} by
$\mathrm{const}\cdot r^{-2}\, \T^{\delta-1}$. Since the spectral norm of the numerator in the right-hand side of \eqref{eq_x19} is at most $4$, we conclude that the part of the increment $M'(Z)-M'(\tilde Z)$ due to the change of the denominator also admits a bound \eqref{eq_x21}. Summing up, we have shown that with probability tending to $1$,
$$
 \| M'(O_2 D^{\text{rand}} O_2^*)-M'(O_2 D O_2^*)\|< \text{const} \cdot  \T^{\delta-1}.
$$
This finishes the proof of Proposition \ref{Proposition_H1_error}.
\end{proof}

\begin{proof}[Proof of Lemma \ref{Lemma_inv_bounded}] We would like to show that for some deterministic $r>0$ the following inequality on the norm holds with probability tending to $1$ as $N,\T\to\infty$:
\begin{equation}
\label{eq_x22}
\left\| \frac{Z_{22}+Z_{22}^*}{2}\right\|<1-\tfrac{r}{2}.
\end{equation}
For that we denote $H=\frac{Z+Z^*}{2}$ and notice that $H$ is a Hermitian matrix with spectrum between $-1$ and $1$. In fact, \citet[Section 3.7]{forrest} explains that $\frac{1+H}{2}$ is again a Jacobi ensemble. The law of large numbers for the latter (reviewed in Theorem \ref{Theorem_Jacobi_as}) implies that as $\T\to\infty$, the empirical measure of the eigenvalues $\frac{1}{\T}\sum_{i=1}^\T \delta_{\lambda_i(H)}$ converges (weakly in probability\footnote{A sequence of \emph{random} probability measures $\mu_n$ converges weakly in probability to a measure $\mu$ if for each bounded continuous function $f(x)$, the random variables $\int f(x) \mu_n(dx)$ converge in probability to $\int f(x) \mu(dx)$.}) to a certain explicit deterministic measure on $[-1,1]$ with a density $p(x)$.

The matrix $\frac{Z_{22}+Z_{22}^*}{2}$ is a $(\T-N)\times(\T-N)$ corner of $H$. The largest eigenvalue of this matrix has the following representation:
\begin{equation}
\label{eq_x23}
 \lambda_1\left(\frac{Z_{22}+Z_{22}^*}{2}\right)=\max_{\begin{smallmatrix} u\in \langle e_{N+1},\dots,e_{\T}\rangle\\ |u|=1\end{smallmatrix}}\left( \sum_{i=1}^\T \lambda_i(H)\, \bigl|(u, v_i(H))\bigr|^2\right),
\end{equation}
where $\lambda_i(H)$ is the $i$th largest eigenvalue of $H$, $v_i(H)$ is the corresponding eigenvector and $(u,v_i(H))$ is the scalar product with this eigenvector. Since the vectors $v_i$ are orthonormal, $\sum_{i=1}^N \bigl|(u, v_i(H))\bigr|^2=1$. On the other hand, all $\lambda_i(H)$ are not greater than $1$ and only few of them are close to $1$. Hence,  \eqref{eq_x22} would follow, if we manage to show that the sum in \eqref{eq_x23} is not concentrated on few largest $\lambda_i(H)$. To show that, it suffices to upper-bound $\bigl|(u, v_i(H))\bigr|^2$, which we now do.

Let us fix $k$ with $k<\min(N,\T-N)$ and upper bound \eqref{eq_x23} by
\begin{multline}
\label{eq_x25}
 \max_{\begin{smallmatrix} u\in \langle e_{N+1},\dots,e_{\T}\rangle \\ |u|=1\end{smallmatrix}}\left(\sum_{i=1}^k  \bigl|(u, v_i(H))\bigr|^2 + \lambda_{k+1}(H)\sum_{i=k+1}^\T  \bigl|(u, v_i(H))\bigr|^2\right)\\=
  \lambda_{k+1}(H)+ (1-\lambda_{k+1}(H)) \max_{\begin{smallmatrix} u\in \langle e_{N+1},\dots,e_{\T}\rangle \\ |u|=1\end{smallmatrix}}\left(\sum_{i=1}^k  \bigl|(u, v_i(H))\bigr|^2 \right).
\end{multline}
Since the law of the matrix $H$ is invariant under conjugations with orthogonal (or unitary) matrices, the eigenvectors $v_1(H),v_2(H),\dots,v_\T(H)$ are uniformly-distributed, i.e.~the $\T\times \T$ matrix formed by them is a uniformly-random orthogonal (or unitary) matrix. Hence,
$$
\max_{\begin{smallmatrix} u\in \langle e_{N+1},\dots,e_{\T-N}\rangle \\ |u|=1\end{smallmatrix}}\left(\sum_{i=1}^k  \bigl|(u, v_i(H))\bigr|^2 \right).
$$
is a maximum eigenvalue of the $(\T-N)\times (\T-N)$ corner of the projector in $\T$--dimensional space on random $k$--dimensional subspace chosen uniformly at random. If we denote through $P_k$ and $P_{\T-N}$ the projectors on the first $k$ and $\T-N$ basis vectors, respectively, and take a uniformly random orthogonal (or unitary if $\beta=2$) $\T\times \T$ matrix $O$, then we deal with the maximal eigenvalue of $P_{\T-N} O P_{k} O^* P_{\T-N}$. The maximal eigenvalue of this product is the same as the one of  $O P_{k} O^* P_{\T-N}$, or of $P_k O^* P_{\T-N} O$, or of $P_k O^* P_{\T-N} O P_k$.
Equivalently, this is the maximal eigenvalue of $k\times k$ corner $Y$ of a projector on uniformly-random $(\T-N)$--dimensional space. Such $k\times k$ matrix $Y$ is distributed as Jacobi ensemble with density proportional to
\begin{equation}
\label{eq_x24}
 \det(Y)^{\frac{\beta}{2}(\T-N-k+1)-1} \det(1-Y)^{\frac{\beta}{2}(N-k+1)-1} dY,
\end{equation}
see  \citet[(3.113) and the following formula]{forrest}.
We now choose $k=N/2$. Then the largest eigenvalue of the ensemble \eqref{eq_x24} converges as $N,\T\to\infty$ in probability (see e.g., \citet{Johnstone_Jacobi}, or \citet[Section 2.6.2]{AGZ}, or \citet{Holcomb_MorenoFlores}) to a deterministic number $0<c<1$, depending on the ratio $\T/N$.

Simultaneously, $\lambda_{N/2+1}(H)$ converges as $\T,N\to\infty$ to another deterministic number $0<c'<1$, due to the aforementioned convergence of the empirical measure of $H$.  Hence,  \eqref{eq_x25} implies that the largest eigenvalue of $\frac{Z_{22}+Z_{22}^*}{2}$ is bounded away from 1.
\end{proof}
\begin{remark} Although we are not pursuing this direction, it is possible to upgrade the above argument to the case when $N,\T\to\infty$ in such a way that $\T/N\to\infty$. Then $r$ is  tending to $0$, but we can find the speed of convergence (and then propagate it to the precise bound in Proposition \ref{Proposition_H1_error}). For that we would need to analyze the largest eigenvalue of \eqref{eq_x24} more precisely (when $\T/N\to\infty$, the Jacobi ensemble concentrates near $1$ and degenerates into the Wishart ensemble after proper rescaling), and also use exact asymptotics of $\lambda_{k+1}(H)$ (which can be obtained from the local law of the Jacobi ensemble for $H$ near $1$).
\end{remark}

We now have all the ingredients for Theorem \ref{Theorem_main_approximation} (as well as for its $\beta=2$ version).

\begin{proof}[Proof of Theorem \ref{Theorem_main_approximation}]
 By Theorem \ref{Theorem_Sum_Jacobi}, the eigenvalues $x_1\ge \dots\ge x_N$ of Jacobi ensemble $\J(N; \frac{N}{2},\frac{T-2N}{2})$ have the same distribution as those of $M$ defined in \eqref{eq_M_def}. On the other hand, by Proposition \ref{Prop_gaussian_rotation} the eigenvalues $\lambda_1\ge \dots\ge \lambda_N$ of the matrix $S_{10} S_{00}^{-1} S_{01} S_{11}^{-1}$ have the same distribution as those of $\tilde M$ defined in \eqref{eq_tilde_M}. Hence, Proposition \ref{Proposition_H1_error} gives the desired statement.
\end{proof}

\subsection{Asymptotic of Jacobi ensemble}
\label{Section_Jacobi}

In this section we review the asymptotic results for the Jacobi ensemble $\J(N; p,q)$ introduced in the Definition \ref{Definition_Jacobi} as $N\to\infty$.

We assume that as $N\to\infty$, also $p,q\to\infty$, in such a way that
\begin{equation}
\label{eq_asymptotic_pars}
 \frac{p-1}{N}=\frac{1}{2} (\p-1), \quad \p\ge 1,\qquad \frac{q-1}{N}=\frac{1}{2} (\q-1), \quad \q\ge 1,
\end{equation}
where $\p$ and $\q$ are two parameters, which stay bounded away from $1$ and from $\infty$ as $N\to\infty$.\footnote{The use of $p-1$ and $q-1$ instead of $p$ and $q$ is related to the choice of notations in \eqref{eq_Jacobi_def}. The following theorems would work without this shift as well, however, the shift leads to a faster speed of convergence, cf.~\citet[Discussion before Theorem 1]{Johnstone_Jacobi}.}  We further define the \emph{equilibrium measure} $\mu_{\p,\q}$ of the Jacobi ensemble through:
\begin{equation}
\label{eq_Jacobi_equilibrium}
 \mu_{\p,\q}(x)\, d x = \frac{\p+\q}{2\pi} \cdot \frac{\sqrt{(x-\lambda_-)(\lambda_+-x)}}{x (1-x)} \mathbf 1_{[\lambda_-,\lambda_+]}\, d x,
\end{equation}
where the support $[\lambda_-,\lambda_+]$ of the measure is defined via
\begin{equation}
 \lambda_\pm=\frac{1}{(\p+\q)^2}\left(\sqrt{\p(\p+\q-1)}\pm \sqrt{\q}  \right)^2.
\end{equation}
One can check that $0<\lambda_-<\lambda_+<1$ for every $\p,\q>1$. The law \eqref{eq_Jacobi_equilibrium} is known as Wachter distribution.
Further, define
\begin{equation}
 c_\pm=\frac{(\p+\q)}{2} \frac{\sqrt{\lambda_+-\lambda_-}}{\lambda_\pm (1-\lambda_\pm)},
\end{equation}
and note that
$$
 \mu_{\p,\q}(x)\approx \frac{c_\pm}{\pi} \sqrt{|x-\lambda_{\pm}|}, \text{ as } x\to\lambda_{\pm}\quad \text{ inside }\quad  [\lambda_-,\lambda_+],
$$
where the normalization $\frac{1}{\pi} \sqrt{|x-\lambda_{\pm}|}$ was chosen to match the behavior of the Wigner semicircle law $\frac{1}{2\pi} \sqrt{4-x^2}$ near edges $\pm 2$.
\begin{theorem}
\label{Theorem_Jacobi_as}
 Suppose that $N,p,q\to\infty$ in such a way that $\p\ge 1 $ and $\q\ge 1 $ in \eqref{eq_asymptotic_pars} stay bounded.  For the second conclusions we additionally assume that $\q$ is bounded away from $1$ and for the third conclusion we additionally require $\p$ to be bounded away from $1$. Let $x_1\ge x_2\ge \dots\ge x_N$ be $N$ random eigenvalues of Jacobi ensemble $\J(N;p,q)$. Then
 \begin{enumerate}
  \item[(I)] $\displaystyle \lim_{N\to\infty} \left|\frac{1}{N} \sum_{i=1}^N \delta_{x_i}- \mu_{\p,\q} \right|=0,$ weakly in probability.

      This means that for any continuous function $f(x)$ we have convergence in probability:
      \begin{equation}
      \label{eq_Jacobi_LLN}
        \lim_{N\to\infty} \left|\frac{1}{N}\sum_{i=1}^N f(x_i)-\int_{0}^1 f(x)\mu_{\p,\q}(x)dx\right|=0.
      \end{equation}
  \item[(II)] For $\{\aa_i\}_{i=1}^{\infty}$ as in Proposition \ref{Proposition_Airy_Gauss}, we have convergence in finite-dimensional distributions for the largest eigenvalues:
      \begin{equation}
      \label{eq_Jacobi_up_edge}
        \lim_{N\to\infty} \left\{ N^{2/3} c_+^{2/3} \left(x_i- \lambda_+\right)  \right\}_{i=1}^{N}\to \{\aa_i\}_{i=1}^{\infty}.
      \end{equation}
      In particular, $N^{2/3} c_+^{2/3} \left(x_1- \lambda_+\right)$ converges to the Tracy-Widom distribution $F_1$.
  \item[(III)] We also have convergence in distribution for the smallest eigenvalues\footnote{The limiting processes $\{\aa_i\}_{i=1}^{\infty}$ arising for the largest and smallest eigenvalues are independent.}
    \begin{equation}
    \label{eq_Jacobi_low_edge}
        \lim_{N\to\infty} \left\{ N^{2/3} c_-^{2/3} \left(\lambda_--x_{N+1-i}\right)  \right\}_{i=1}^{N}\to \{\aa_i\}_{i=1}^{\infty}.
      \end{equation}
 \end{enumerate}
\end{theorem}
 The first proof of \eqref{eq_Jacobi_LLN} appeared in \citet{Wachter}, for other proofs see \citet{Bai_Silverstein} and \citet{Dumitriu_Paquette}; we follow the notations of the last reference. The asymptotics \eqref{eq_Jacobi_up_edge}, \eqref{eq_Jacobi_low_edge} can be found in \citet{Johnstone_Jacobi}, and it is a manifestation of the universality for the distributions of largest/smallest eigenvalues of random matrices holding in much wider generality, see, e.g., \citet{Deift_Gioev}, \citet{ErdosYau}, \citet{Tao_Vu}. We remark that the original article \citet{Johnstone_Jacobi} stated the convergence under an additional parity constraint on the parameters, yet this technical restriction can be removed, see the discussion at the end of Section 2 in \cite{Johnstone_Jacobi} and \cite{HanPanZhang_2016}.

\subsection{Proof of Theorem \ref{theorem_J_stat}} As an intermediate step we prove the following theorem.

\begin{theorem} \label{Theorem_Joh_largest}
 Suppose that $T,N\to\infty$ in such a way that $T/N>2$ remains bounded away from $2$ and from $\infty$. Let $\lambda_1\ge \dots\ge \lambda_N$ be eigenvalues of the matrix $S_{10} S_{00}^{-1} S_{01} S_{11}^{-1}$ defined as in Theorem \ref{theorem_J_stat} and Proposition \ref{Prop_gaussian_rotation}. Then for $\{\aa_i\}_{i=1}^{\infty}$ as in Proposition \ref{Proposition_Airy_Gauss}, we have convergence in finite-dimensional distributions for the largest eigenvalues:
      \begin{equation}
      \label{eq_J_stat_edge}
        \lim_{N\to\infty} \left\{ N^{2/3} c_+^{2/3} \left(\lambda_i- \lambda_+\right)  \right\}_{i=1}^{\infty}\to \{\aa_i\}_{i=1}^{\infty},
      \end{equation}
      where
      \begin{equation}
        \lambda_{\pm}=\frac{1}{(\p+\q)^2}\left[\sqrt{\p(\p+\q-1)}\pm \sqrt{\q}  \right]^2, \qquad
      c_+=\left(\p+\q\right) \frac{\sqrt{\lambda_+-\lambda_-}}{2\cdot \lambda_+\cdot (1-\lambda_+)},
      \end{equation}
      $$ \p=2-\frac{2}{N},\qquad \q=\frac{T}{N}-1-\frac{2}{N}.$$
\end{theorem}
\begin{proof}
  Theorem \ref{Theorem_main_approximation} implies that the asymptotics of largest eigenvalues $x_1,x_2,\dots$ is the same as the one for the largest eigenvalues of $\J(N; \frac{N}{2}, \frac{T-2N}{2})$. For the latter we use Theorem \ref{Theorem_Jacobi_as} with $\p=2-\frac{2}{N}$, $\q=\frac{T}{N}-1-\frac{2}{N}$.
\end{proof}
\begin{proof}[Proof of Theorem \ref{theorem_J_stat}]
  We use Theorem \ref{Theorem_Joh_largest} and the fact that for small $x$
  $$
   \ln\bigl(1-(\lambda_++x)\bigr)=\ln(1-\lambda_+)-\frac{1}{1-\lambda_+} x + o(x).\qedhere
  $$
\end{proof}

\subsection{Asymptotics under white noise assumption}

\label{Section_white_proof}

In this section we explain how the proof of Theorem \ref{Theorem_white_noise_approximation} is obtained. We follow the notations of Section \ref{Section_proof_of_main_approx} for the deterministic operators and spaces.
We take a uniformly-random orthogonal (or unitary if $\beta=2$) operator $\tilde O$ acting in $(T-1)$--dimensional space $V$ and define an operator $\tilde W$ acting in $V$:
$$
 \tilde W= \1_V- \tilde O (F_V)^{-1} \tilde O^*,
$$
 Let us introduce a symmetric (or Hermitian if $\beta=2$) $N\times N$ matrix $\widehat M$ through
\begin{equation}
\label{eq_tilde_M_white}
 \widehat M= [\tilde W]_{NN}^V ([\tilde W^* \tilde W]_{NN}^V)^{-1} [\tilde W^*]_{NN}^V.
\end{equation}
\begin{proposition}
\label{Prop_gaussian_rotation_white}
Choose an arbitrary positive definite covariance matrix $\Lambda$. Let $\eps$ be $N\times T$ matrix of random variables (real if $\beta=1$ and complex if $\beta=2$), such that $T$ columns of $\eps$ are i.i.d., and each of them is an $N$-dimensional mean zero Gaussian vector with covariance $\Lambda$. Fix an arbitrary $N$--dimensional vector $\mu$. Define the $N\times T$ matrix $X=(X_1,X_2,\dots,X_T)$ as in Eq.~\eqref{eq_wn_hypothesis} via
$$
 X_t=\mu+\eps_t, \quad t=1,\dots,T.
$$
Further set $\Delta^c X_t=X_{t+1}-X_{t}$, $t=1,\dots,T-1$ and $\Delta^c X_T=X_1-X_t$.
 Define $N\times N$ matrices:
\begin{equation}
\label{eq_modified_Joh_matrices_wn_proof}
 S_{00}^{w.n.}=\Delta^c X \mathcal P (\Delta^c X)^{*},\quad S_{01}^{w.n.}=\Delta^c X \mathcal P X^*, \quad S_{10}^{w.n.}=X \mathcal P (\Delta^c X)^*,\quad S_{11}^{w.n.}=X \mathcal P X^*,
\end{equation} Then the eigenvalues of the matrix  $S_{10}^{w.n.} (S_{00}^{w.n.})^{-1} S_{01}^{w.n.} (S_{11}^{w.n.})^{-1}$ have the same distribution as those of $\widehat M$ in \eqref{eq_tilde_M_white}.
\end{proposition}
\begin{proof}
The proof is similar to that of Proposition \ref{Prop_gaussian_rotation} and we omit many details. First, note that $\mu$ cancels out both in $X \mathcal P$ and in $\Delta^c X$. Hence, without loss of generality we can assume $\mu=0$.

Next, we define
$$
 G=\mathcal P ( F^{-1} - \1_{T}) \mathcal P.
$$
Note that $\mathcal P$ and $( F^{-1} - \1_{T})$ commute with each other. Using additionally $\mathcal P^2=\mathcal P$, we get:
$$
 S_{00}^{w.n.}=\eps G^* G \eps^{*},\quad S_{01}^{w.n.}=\eps G^* \eps^*, \quad S_{10}^{w.n.}=\eps G \eps^*,\quad S_{11}^{w.n.}=\eps \mathcal P \eps^*.
$$
At this point, arguing as in the proof of  Proposition \ref{Prop_gaussian_rotation}, we can assume that the covariance matrix $\Lambda$ is identical, since any other covariance matrix can be obtained at the cost of  conjugation of  $S_{10}^{w.n.} (S_{00}^{w.n.})^{-1} S_{01}^{w.n.} (S_{11}^{w.n.})^{-1}$, which leaves the eigenvalues unchanged.

We can further identify non-zero eigenvalues of  $S_{10}^{w.n.} (S_{00}^{w.n.})^{-1} S_{01}^{w.n.} (S_{11}^{w.n.})^{-1}$ with those of the product of two projectors $P_2 P_1 P_2$ (similarly to \eqref{eq_product_proj}), where $P_1$ is the orthogonal projector onto the space spanned by $N$ columns of $G\eps^*$ and $P_2$ is the orthogonal projector onto the space spanned by $N$ columns of $\mathcal P \eps^*$. Next, we rotate the space $V$, so that the span of $N$ columns of $\mathcal P \eps^*$ turns into the span of the first $N$ basis vectors $\tilde e_1,\dots,\tilde e_N$ in $V$. At this point we arrive at
\eqref{eq_tilde_M_white} with $\tilde W$ replaced by $-\tilde W$. It remains to notice that the introduction of prefactor $-1$ in front of $\tilde W$ does not change the matrix \eqref{eq_tilde_M_white}.
\end{proof}

\medskip

Let $\T=T-1$ and let $O$ be uniformly random real orthogonal of determinant $1$ (if $\beta=1$) or complex unitary (if $\beta=2$) $\T\times \T$ matrix. Define $W=\1_\T+ O$. Then we set
\begin{equation}
\label{eq_M_def_white}
 M=[W]_{NN} ([W^* W]_{NN})^{-1} [W^*]_{NN}.
\end{equation}
Recall that for a matrix $A$ with real spectrum,  $\lambda_i(A)$ is the $i$th largest eigenvalue of $A$.
\begin{proposition}
\label{Proposition_H1_error_white}
 Assume $\T=T-1\ge 2N$. One can couple $M$ from \eqref{eq_M_def_white} with $\widehat M$ from \eqref{eq_tilde_M_white} in such a way that  for each $\epsilon>0$ we have
 $$
   \mathrm{Prob}\left( \max_{1\le i \le N} |\lambda_i(M)-\lambda_i(\widehat M)|< \frac{1}{N^{1-\epsilon}}\right)\to 1,
 $$
 as $T,N\to\infty$ in such a way that the ratio $T/N$ remains bounded.
\end{proposition}

The proof is almost identical to that of Proposition \ref{Proposition_H1_error} and we omit it: the key idea is to notice that the only difference between $M$ and $\widehat M$ is in the replacement of $O$ by $- \tilde O (F_V)^{-1} \tilde O^*$; however, by the results of \citet{MeckesMeckes} the latter two matrices can be coupled so that their eigenvectors are the same, while eigenvalues are very close to each other.

\smallskip

Combining Proposition \ref{Prop_gaussian_rotation_white}, Proposition \ref{Proposition_H1_error_white}, and Theorem \ref{Theorem_Var_0}, we arrive at the statement of Theorem \ref{Theorem_white_noise_approximation}.

\subsection{Data}\label{data_subsection}
As of May 12, 2020, S$\&$P100 consists of the following companies: Apple Inc., AbbVie Inc.,	Abbott Laboratories	Accenture, Adobe Inc., American International Group	Allstate, Amgen Inc., American Tower, Amazon.com, American Express	Boeing Co., Bank of America Corp, Biogen, The Bank of New York Mellon, Booking Holdings, BlackRock Inc, Bristol-Myers Squibb, Berkshire Hathaway, Citigroup Inc, Caterpillar Inc., Charter Communications, Colgate-Palmolive, Comcast Corp., Capital One Financial Corp., ConocoPhillips, Costco Wholesale Corp., salesforce.com, Cisco Systems, CVS Health, Chevron Corporation, DuPont de Nemours Inc., Danaher Corporation, The Walt Disney Company, Dow Inc., Duke Energy, Emerson Electric Co., Exelon, Ford Motor Company, Facebook Inc., FedEx, General Dynamics, General Electric, Gilead Sciences,	General Motors,	Alphabet Inc. (Class C), Alphabet Inc. (Class A), Goldman Sachs, Home Depot, Honeywell, International Business Machines, Intel Corp., Johnson $\&$ Johnson, JPMorgan Chase $\&$ Co., Kraft Heinz, Kinder Morgan, The Coca-Cola Company,	Eli Lilly and Company, Lockheed Martin, Lowe's, MasterCard Inc, McDonald's Corp, Mondel\={e}z International, Medtronic plc, MetLife Inc., 3M Company, Altria Group, Merck $\&$ Co., Morgan Stanley, Microsoft, NextEra Energy, Netflix, Nike Inc., NVIDIA Corp., Oracle Corporation, Occidental Petroleum Corp., PepsiCo, Pfizer Inc, Procter $\&$ Gamble Co, Philip Morris International, PayPal Holdings, Qualcomm Inc., Raytheon Technologies, Starbucks Corp., Schlumberger, Southern Company, Simon Property Group, Inc., AT$\&$T Inc, Target Corporation, Thermo Fisher Scientific, Texas Instruments, UnitedHealth Group, Union Pacific Corporation, United Parcel Service, U.S. Bancorp, Visa Inc., Verizon Communications, Walgreens Boots Alliance, Wells Fargo, Walmart, Exxon Mobil Corp.

Eight of the above companies are not available for the entire period under consideration $01.01.2010-01.01.2020$. Those companies are AbbVie Inc. (founded in $2013$), Charter Communications (was bancrupt in $2009$ and got released on NASDAQ in the middle of $2010$), Dow Inc. (was spun off of DowDuPont on April 1, 2019), Facebook Inc. (went on IPO on February 1, 2012), General Motors (approached bankruptcy in 2009 and returned to NASDAQ as a new company at the end of $2010$), Kraft Heinz (Kraft Foods and H.J. Heinz merged into Kraft Heinz in $2015$), Kinder Morgan (was taken in a buyout and began trading again on the NYSE on February 11, 2011), PayPal Holdings (was part of eBay until $2015$).

\section*{Acknowledgements}  The authors are grateful to Donald Andrews, Alexei Borodin, Giuseppe Cavaliere, Alice Guionnet, Bruce Hansen, S{\o}ren Johansen, Grigori Olshanski, and anonymous referees for valuable suggestions. V.G.~acknowledges support from the NSF Grants DMS-1664619 and DMS-1949820.

\bibliographystyle{aer}
\bibliography{large_nT_bib_new}

\end{document}